%% file: PREPRINT.tex
\pdfoutput=1
\documentclass[11pt,fleqn]{article}

\usepackage[utf8]{inputenc}
\usepackage[T1]{fontenc}

\usepackage{comment}
\usepackage{fullpage}
\usepackage[indent,parfill]{parskip}
\usepackage{indentfirst}

\usepackage{amsmath,amsthm}
\usepackage{amssymb,amsfonts}
\usepackage{mathtools,thm-restate}

\usepackage{mathptmx}
\usepackage{esvect}  
\usepackage{stmaryrd}  
\usepackage[cal=pxtx,scr=boondox]{mathalpha}

\usepackage[style=alphabetic,natbib=true,sortcites=true,backref=true,maxnames=99,maxalphanames=99]{biblatex}

\DefineBibliographyStrings{english}{%
  backrefpage = {$\Lsh$~p.},
  backrefpages = {$\Lsh$~pp.},
}
\usepackage[colorlinks=true,citecolor=blue!65!black,linkcolor=red!75!black]{hyperref}
\usepackage[nameinlink]{cleveref}

\input{macro}

\let\Pr\relax\DeclareMathOperator*{\Pr}{\mathbb{P}}

\newtheorem{theorem}{Theorem}[section]
\newtheorem{proposition}[theorem]{Proposition}
\newtheorem{lemma}[theorem]{Lemma}
\newtheorem{claim}[theorem]{Claim}

\newtheorem{observation}[theorem]{Observation}

\theoremstyle{definition}
\newtheorem{definition}[theorem]{Definition}
\newtheorem{problem}[theorem]{Problem}

\newtheorem{remark}[theorem]{Remark}

\numberwithin{equation}{section}

\title{Asymptotically Optimal Inapproximability of \\ Maxmin $k$-Cut Reconfiguration}
\author{
Shuichi Hirahara \\
\small{National Institute of Informatics, Japan} \\
\small{\href{mailto:s\_hirahara@nii.ac.jp}{\texttt{s\_hirahara@nii.ac.jp}}}
\and
Naoto Ohsaka \\
\small{CyberAgent, Inc., Japan} \\
\small{\href{mailto:naoto.ohsaka@gmail.com}{\texttt{ohsaka\_naoto@cyberagent.co.jp}}}
}
\date{\today}

\begin{document}
\maketitle
\thispagestyle{empty}

\begin{abstract}\input{abstract}\end{abstract}

\clearpage
\tableofcontents

\clearpage
\input{main}
\input{pre}
\input{Cut-hard}
\input{Cut-alg}

\appendix
\input{app}

\printbibliography 

\end{document}

%% file: macro.tex
\usepackage[basic,sanserif,disableredefinitions]{complexity}

\usepackage{booktabs,colortbl,multirow}
\usepackage{diagbox,enumitem}
\usepackage{framed,fancybox,ascmac}
\usepackage{subfig}

\usepackage{comment}
\usepackage{url}

\usepackage[colorinlistoftodos,prependcaption]{todonotes}

\usepackage{xspace}
\usepackage{bm}
\usepackage{mleftright}
\mleftright

\usepackage{pifont}  

\makeatletter
\DeclareFontEncoding{LS1}{}{}
\makeatother
\DeclareFontSubstitution{LS1}{stix}{m}{n}
\DeclareSymbolFont{symbolsstix}{LS1}{stixscr}{m}{n}
\DeclareMathSymbol{\squarehfill}{\mathord}{symbolsstix}{"BB}
\DeclareMathSymbol{\squarevfill}{\mathord}{symbolsstix}{"BC}
\DeclareMathSymbol{\squarehvfill}{\mathord}{symbolsstix}{"BD}

\crefname{theorem}{Theorem}{Theorems}
\crefname{proposition}{Proposition}{Propositions}
\crefname{lemma}{Lemma}{Lemmas}
\crefname{claim}{Claim}{Claims}
\crefname{corollary}{Corollary}{Corollaries}
\crefname{remark}{Remark}{Remarks}
\crefname{observation}{Observation}{Observations}
\crefname{hypothesis}{Hypothesis}{Hypotheses}
\crefname{fact}{Fact}{Facts}
\crefname{definition}{Definition}{Definitions}
\crefname{problem}{Problem}{Problems}
\crefname{example}{Example}{Examples}

\crefname{appendix}{Appendix}{Appendices}
\crefname{section}{Section}{Sections}
\crefname{equation}{Eq.}{Eqs.}
\crefname{figure}{Figure}{Figures}
\crefname{table}{Table}{Tables}
\crefname{algorithm}{Algorithm}{Algorithms}

\usepackage[ruled]{algorithm}
\usepackage[commentColor=blue]{algpseudocodex}

\algrenewcommand\textproc{\textsl}

\renewcommand{\geq}{\geqslant}
\renewcommand{\leq}{\leqslant}
\renewcommand{\epsilon}{\varepsilon}
\renewcommand{\theta}{\vartheta}
\renewcommand{\phi}{\varphi}
\renewcommand{\bar}{\overline}

\renewcommand{\tilde}{\widetilde}
\renewcommand{\top}{\intercal}

\newcommand{\prb}[1]{\textsc{#1}\xspace}
\newcommand{\nth}[1]{#1\textsuperscript{th}\xspace}

\renewcommand{\vec}[1]{\mathbf{#1}}

\newcommand{\reco}{\leftrightsquigarrow}
\newcommand{\defeq}{\coloneq}

\DeclareMathOperator*{\argmax}{argmax}

\let\E\relax\DeclareMathOperator*{\E}{\mathbb{E}}  
\let\Pr\relax\DeclareMathOperator*{\Pr}{\mathbb{Pr}}
\DeclareMathOperator{\bigO}{\mathcal{O}}
\DeclareMathOperator{\val}{\mathsf{val}}
\DeclareMathOperator{\opt}{\mathsf{opt}}
\DeclareMathOperator{\rHam}{\mathrm{dist}}
\DeclareMathOperator{\enc}{\mathsf{enc}}
\DeclareMathOperator{\dec}{\mathsf{dec}}
\let\poly\relax\DeclareMathOperator*{\poly}{\mathrm{poly}}

\newcommand{\nei}{\calN}
\newcommand{\sss}{\mathsf{start}}
\newcommand{\ttt}{\mathsf{end}}

\newcommand{\Yes}{\textsc{Yes}\xspace}
\newcommand{\No}{\textsc{No}\xspace}

\newcommand{\V}{\calV}

\newcommand{\Accept}{\textsf{accept}\xspace}
\newcommand{\Reject}{\textsf{reject}\xspace}

\newcommand{\f}{f}
\newcommand{\g}{g}
\newcommand{\frnd}{F}
\newcommand{\sqcol}{\scrF}
\newcommand{\stsqcol}{\bbF}

\newcommand{\kzero}{10^3}
\newcommand{\rhozero}{10^{-8}}
\newcommand{\epsilonzero}{10^{-2}}
\newcommand{\hor}{\squarehfill}
\newcommand{\ver}{\squarevfill}
\newcommand{\str}{\squarehvfill}
\newcommand{\Vstripe}{\V_\mathrm{stripe}}
\newcommand{\Vcons}{\V_\mathrm{cons}}
\newcommand{\Vedge}{\V_\mathrm{edge}}
\newcommand{\Good}{\mathsf{Good}}
\newcommand{\Bad}{\mathsf{Bad}}
\newcommand{\Badgtr}{\Bad^\mathrm{lot}}
\newcommand{\Badlss}{\Bad^\mathrm{few}}
\newcommand{\BadlssL}{\Badlss_\mathrm{long}}
\newcommand{\BadlssS}{\Badlss_\mathrm{short}}
\newcommand{\SQ}{\text{{\Large\ding{112}}}}

\newcommand{\Vsix}{\V_{\text{\ref{lem:Cut-hard:6-to-2}}}}
\newcommand{\Vquad}{\V_{\text{\ref{lem:Cut-hard:quadratic}}}}

\newcommand{\Vl}{V_{\leq \Delta}}
\newcommand{\Vg}{V_{> \Delta}}

\newcommand{\kColReconf}{\prb{$k$-Coloring Reconfiguration}}
\newcommand{\kCutReconf}{\prb{$k$-Cut Reconfiguration}}
\newcommand{\MMkCutReconf}{\prb{Maxmin $k$-Cut Reconfiguration}}
\newcommand{\twoCutReconf}{\prb{2-Cut Reconfiguration}}
\newcommand{\sixCutReconf}{\prb{6-Cut Reconfiguration}}
\newcommand{\MMtwoCutReconf}{\prb{Maxmin 2-Cut Reconfiguration}}
\newcommand{\MMsixCutReconf}{\prb{Maxmin 6-Cut Reconfiguration}}

\newcommand{\calD}{\mathcal{D}}

\newcommand{\calN}{\mathcal{N}}

\newcommand{\calR}{\mathcal{R}}
\newcommand{\calS}{\mathcal{S}}

\newcommand{\calV}{\mathcal{V}}

\newcommand{\bbF}{\mathbb{F}}

\newcommand{\bbN}{\mathbb{N}}

\newcommand{\scrF}{\mathscr{F}}

\newcommand{\scrS}{\mathscr{S}}

\newcommand{\frakS}{\mathfrak{S}}

%% file: abstract.tex
\prb{$k$-Coloring Reconfiguration} is one of the most well-studied reconfiguration problems, which asks to
transform a given proper $k$-coloring of a graph to another by repeatedly recoloring a single vertex.
Its approximate version, \prb{Maxmin $k$-Cut Reconfiguration},
is defined as an optimization problem of maximizing the minimum fraction of bichromatic edges during the transformation between (not necessarily proper) $k$-colorings.
In this paper,
we prove that the optimal approximation factor of this problem is $1 - \Theta\left(\frac{1}{k}\right)$ for every $k \ge 2$.
Specifically, we show the $\PSPACE$-hardness of approximating the objective value within a factor of $1 - \frac{\varepsilon}{k}$ for some universal constant $\varepsilon > 0$,
whereas we present a deterministic polynomial-time algorithm that achieves the approximation factor of $1 - \frac{2}{k}$.

To prove the hardness result, we develop a new probabilistic verifier that tests a ``striped'' pattern.
Our polynomial-time algorithm is based on ``random reconfiguration via a random solution,'' i.e., 
the transformation that goes through one random $k$-coloring.

%% file: main.tex
\section{Introduction}

\emph{Reconfiguration} is an emerging field in theoretical computer science,
which studies reachability and connectivity problems over the space of solutions.
A \emph{reconfiguration problem} can be defined
for any combinatorial problem $\Pi$ and any transformation rule over the feasible solutions of $\Pi$.
The problem $\Pi$ is referred to as the \emph{source problem} of a reconfiguration problem.
For an instance $I$ of $\Pi$ and a pair of its feasible solutions,
the reconfiguration problem asks if
one solution can be transformed into the other 
by repeatedly applying the transformation rule
while always preserving that every intermediate solution is feasible.
Speaking differently,
the reconfiguration problem concerns the reachability over the \emph{configuration graph}, where
each node corresponds to a feasible solution of the given instance $I$ and 
each link represents that its endpoints are ``adjacent'' under the transformation rule.
Such a sequence of feasible solutions that form a path on the configuration graph is called a \emph{reconfiguration sequence}.
Over the past twenty years,
many reconfiguration problems have been defined from a variety of source problems,
including Boolean satisfiability, constraint satisfaction problems, and graph problems.

The computational complexity of reconfiguration problems has been extensively studied; e.g.,
reconfiguration problems of
\prb{3-SAT} \cite{gopalan2009connectivity},
\prb{Independent Set} \cite{hearn2005pspace,hearn2009games}, and
\prb{Set Cover} \cite{ito2011complexity}
are $\PSPACE$-complete, whereas
those of
\prb{2-SAT} \cite{gopalan2009connectivity},
\prb{Matching} \cite{ito2011complexity}, and
\prb{Spanning Tree} \cite{ito2011complexity}
belong to $\cP$.
We refer the readers to the surveys by \citet{nishimura2018introduction,heuvel13complexity,mynhardt2019reconfiguration,bousquet2024survey}
as well as
the Combinatorial Reconfiguration wiki \cite{hoang2023combinatorial} for
more algorithmic, hardness, and structural results of reconfiguration problems.

\subsection{\kColReconf and Its Approximate Version}
One of the most well-studied reconfiguration problems,
which we study in this paper, is 
\kColReconf \cite{bonsma2009finding,cereceda2008connectedness,cereceda2011finding,cereceda2009mixing,cereceda2007mixing},
whose source problem is \prb{$k$-Coloring}.
Recall that \prb{$k$-Coloring} is a graph coloring problem of 
deciding if a graph $G$ is \emph{$k$-colorable}; namely,
there is a \emph{proper} $k$-coloring $\f \colon V(G) \to [k]$ of $G$, which
renders every edge bichromatic.\footnote{
An edge is \emph{bichromatic} if its endpoints receive different colors.
}
In the \kColReconf problem,
for a $k$-colorable graph $G$ and a pair of its proper $k$-colorings $\f_\sss,\f_\ttt \colon V(G) \to [k]$,
we seek a reconfiguration sequence from $\f_\sss$ to $\f_\ttt$
consisting only of proper $k$-colorings of $G$, such that
every pair of neighboring $k$-colorings differ in a single vertex.
See \cref{fig:yes,fig:example-no} for \Yes and \No instances of \kColReconf.
If the number $k$ of available colors is sufficiently large
(e.g., the maximum degree of $G$ plus $2$ or more \cite{jerrum1995very,dyer2006randomly}),
the answer to this problem is always \Yes.
For a constant value of $k$, the following complexity results are known:
If $k \leq 3$, then \kColReconf belongs to $\cP$ \cite{cereceda2011finding}.\footnote{
Moreover, a reconfiguration sequence for \Yes instances can be found in polynomial time.
}
On the other hand, \kColReconf is $\PSPACE$-complete for every $k \geq 4$ \cite{bonsma2009finding}.
Quite interestingly, \prb{$3$-Coloring} ``becomes'' easy
in the reconfiguration regime even though \prb{$3$-Coloring} itself is $\NP$-complete
\cite{garey1976some,lovasz1973coverings,stockmeyer1973planar}.
Several existing work further investigate the parameterized complexity 
\cite{johnson2016finding,bonsma2014complexity}
and the complexity for restricted graph classes
\cite{bonamy2013recoloring,bonamy2014reconfiguration,wrochna2018reconfiguration,cereceda2009mixing,bonamy2011diameter,hatanaka2019coloring}.
Note that the configuration graph of \kColReconf is closely related to
the \emph{Glauber dynamics} \cite{jerrum1995very,molloy2004glauber,dyer2006randomly}.
See also \cref{sec:related} for related work.

\input{fig-yes}

\input{fig-no}

In this paper, we study \emph{approximability} of \kColReconf.
Since 2023,
approximability of reconfiguration problems has been studied actively from both hardness and algorithmic sides
\cite{ohsaka2023gap,ohsaka2024gap,ohsaka2024alphabet,ohsaka2025approximate,hirahara2024probabilistically,hirahara2024optimal,karthik2023inapproximability,ohsaka2024tight}.
For a reconfiguration problem, its \emph{approximate version} \cite{ito2011complexity}
allows to relax the feasibility of intermediate solutions, but
requires to optimize the ``worst'' feasibility during reconfiguration.
For example,
an approximate version of \prb{Set Cover Reconfiguration}
admits a $2$-factor approximation algorithm \cite{ito2011complexity},
which has been recently proven to be $\PSPACE$-hard to approximate within a factor of $2-o(1)$ \cite{hirahara2024optimal}.
There are two natural approximate versions of \kColReconf
since \prb{$k$-Coloring} has the following two approximate versions:

\begin{description}
    \item[1.~Maximizing the number of bichromatic edges:]
    For a (not necessarily $k$-colorable) graph $G$,
    the first problem asks to find a $k$-coloring of $G$
    that makes as many edges as possible bichromatic.
    This problem is known by the names of \prb{Max $k$-Cut} and \prb{Max $k$-Colorable Subgraph} \cite{papadimitriou1991optimization,guruswami2013improved}.\footnote{
        This problem is also called
        \prb{Max $k$-Coloring} \cite{austrin2014new,feige1998zero} or \prb{Max $k$-Colorability} \cite{petrank1994hardness}.
        We do not use these names to avoid confusion with other graph coloring problems.
    }
    \item[2.~Minimizing the number of used colors:]
    For a (not necessarily $k$-colorable) graph $G$,
    the second problem asks to find a proper coloring of $G$
    that uses as few colors as possible.
    This problem is known as \prb{Chromatic Number} and \prb{Graph Coloring}.
\end{description}
In this paper, we study a reconfiguration analogue of the first problem \prb{$k$-Cut}, which we call \MMkCutReconf.
In this problem, given a (not necessarily $k$-colorable) graph $G = (V,E)$ and a pair of its $k$-colorings $\f_\sss,\f_\ttt \colon V \to [k]$,
we shall construct a reconfiguration sequence $\sqcol$ from $\f_\sss$ to $\f_\ttt$
consisting of any (not necessarily proper) $k$-colorings of $G$
that maximizes the \emph{minimum fraction} of bichromatic edges of $G$,
where the minimum is taken over all $k$-colorings of $\sqcol$.

\begin{itembox}[l]{\MMkCutReconf}
\begin{tabular}{ll}
    \textbf{Input:}
    & a graph $G = (V,E)$ and
    a pair of $k$-colorings $\f_\sss,\f_\ttt \colon V \to [k]$ of $G$.
    \\
    \textbf{Output:}
    & a reconfiguration sequence $\sqcol$ from $\f_\sss$ to $\f_\ttt$.
    \\
    \textbf{Goal:}
    & maximize the minimum fraction of bichromatic edges of $G$ over all $k$-colorings of $\sqcol$.
\end{tabular}
\end{itembox}
See \cref{fig:example-no} for an example of \MMkCutReconf.
Solving this problem,
we may be able to find a ``reasonable'' reconfiguration sequence,
which consists of ``almost'' proper $k$-colorings, so that
we can manage \No instances of \kColReconf.

Here, we briefly review known results on \MMkCutReconf.
The $\PSPACE$-hardness of exactly solving \MMkCutReconf for every $k \geq 4$
follows from that of \kColReconf \cite{bonsma2009finding}.
For the $\PSPACE$-hardness of approximation,
the \emph{Probabilistically Checkable Reconfiguration Proof} (PCRP) theorem
due to \citet{hirahara2024probabilistically,karthik2023inapproximability},
along with a series of gap-preserving reductions due to \citet{ohsaka2023gap,bonsma2009finding},
implies that
\prb{Maxmin 4-Cut Reconfiguration} is $\PSPACE$-hard to approximate within some constant factor.\footnote{See \cref{sec:related} for other applications of the PCRP theorem in the $\PSPACE$-hardness of approximating reconfiguration problems.}
However, 
the \emph{asymptotic} behavior of approximability for \MMkCutReconf with respect to
the number $k$ of available colors is not well understood.

\subsection{Our Results}
In this paper, we find out that
the asymptotically optimal approximation factor of \MMkCutReconf is $1 - \Theta\left(\frac{1}{k}\right)$.
On the hardness side,
we demonstrate the $\PSPACE$-hardness of approximation
within a factor of $1-\Omega\left(\frac{1}{k}\right)$ for \emph{every} $k \geq 2$.

\begin{theorem}[informal; see \cref{thm:Cut-hard}]
\label{thm:intro:Cut-hard}
    There exist universal constants $\epsilon_c,\epsilon_s \in (0,1)$ with $\epsilon_c < \epsilon_s$ such that
    for every $k \geq 2$, a multigraph $G$, and a pair of its $k$-colorings $\f_\sss, \f_\ttt$,
    it is $\PSPACE$-hard to distinguish between the following cases\textup{:}
    \begin{itemize}
        \item \textup{(}Completeness\textup{)}
            There exists a reconfiguration sequence from $\f_\sss$ to $\f_\ttt$
            consisting of $k$-colorings that make at least
            $\left(1-\frac{\epsilon_c}{k}\right)$-fraction of edges of $G$ bichromatic.
        \item \textup{(}Soundness\textup{)}
            Every reconfiguration sequence contains a $k$-coloring that makes
            more than $\frac{\epsilon_s}{k}$-fraction of edges of $G$ monochromatic.
    \end{itemize}
    In particular,
    \MMkCutReconf is $\PSPACE$-hard to approximate within a factor of
    $1-\frac{\epsilon}{k}$ for every $k \geq 2$
    for some universal constant $\epsilon \in (0,1)$.
\end{theorem}\noindent
On the algorithmic side,
we develop a deterministic $\left(1-\frac{2}{k}\right)$-factor approximation algorithm for \emph{every} $k \geq 2$.\footnote{
Although $1-\frac{2}{k} = 0$ if $k=2$,
the actual approximation factor can be arbitrarily close to $\frac{1}{4}$.
See \cref{sec:Cut-alg}.
}

\begin{theorem}[informal; see \cref{thm:Cut-alg}]
\label{thm:intro:Cut-alg}
    For every $k \geq 2$,
    there exists a deterministic $\left(1-\frac{2}{k}\right)$-factor
    approximation algorithm for \MMkCutReconf.
\end{theorem}\noindent
To the best of our knowledge, this is the first non-trivial approximation algorithm for \MMkCutReconf.

\cref{thm:intro:Cut-hard,thm:intro:Cut-alg} provide
asymptotically tight lower and upper bounds for approximability of \MMkCutReconf.

\subsection{Organization}
The rest of this paper is organized as follows.
In \cref{sec:overview-Cut-hard,sec:overview-Cut-alg},
    we present an overview of the proof of \cref{thm:intro:Cut-hard,thm:intro:Cut-alg},
    respectively.
In \cref{sec:related},
    we review related work on
    variants of \kColReconf, and
    approximability of \prb{Max $k$-Cut} and reconfiguration problems.
In \cref{sec:pre},
    we formally define \kColReconf as well as \MMkCutReconf.
In \cref{sec:Cut-hard},
    we prove that
    \MMkCutReconf is $\PSPACE$-hard to approximate within a factor of $1-\Omega\left(\frac{1}{k}\right)$ 
    (\cref{thm:intro:Cut-hard}).
In \cref{sec:Cut-alg},
    we develop a deterministic $\left(1-\frac{2}{k}\right)$-factor approximation algorithm for 
    \MMkCutReconf
    (\cref{thm:intro:Cut-alg}).
Some technical proofs are deferred to \cref{app:Cut-hard}.

\subsection{Notations}
For a nonnegative integer $n \in \bbN$, let $[n] \defeq \{1,2,\ldots,n\}$.
We use the Iverson bracket $\llbracket \cdot \rrbracket$; i.e.,
$\llbracket P \rrbracket$ for a statement $P$
is defined as $1$ if $P$ is true and $0$ otherwise.
A \emph{sequence} $\scrS$ of a finite number of objects, $s^{(1)}, \ldots, s^{(T)}$,
is denoted by $(s^{(1)}, \ldots, s^{(T)})$, and
we write $s \in \scrS$ to indicate that $s$ appears in $\scrS$.
The symbol $\circ$ stands for a concatenation of two sequences or functions, and
$\mathfrak{S}_n$ for the set of all permutations over $[n]$.
For a set $\calS$,
we write $X \sim \calS$ to mean that 
$X$ is a random variable uniformly drawn from $\calS$.
For two functions $f,g \colon \calD \to \calR$ over a finite domain $\calD$,
the \emph{relative Hamming distance} between $f$ and $g$,
denoted by $\rHam(f,g)$,
is defined as the fraction of positions on which $f$ and $g$ differ; namely,
\begin{align}
    \rHam(f,g) \defeq
    \Pr_{x \sim \calD}\Bigl[ f(x) \neq g(x) \Bigr]
    = |\calD|^{-1} \cdot \left|\Bigl\{x \in \calD \Bigm| f(x) \neq g(x) \Bigr\}\right|.
\end{align}
We say that $f$ is \emph{$\epsilon$-close to} $g$ if $\rHam(f,g) \leq \epsilon$ and
\emph{$\epsilon$-far from} $g$ if $\rHam(f,g) > \epsilon$.
Similar notations are used
for a set of function $G$ from $\calD$ to $\calR$; e.g.,
$\rHam(f,G) \defeq \min_{g \in G} \rHam(f,g)$ and
$f$ is $\epsilon$-close to $G$ if $\rHam(f,G) \leq \epsilon$.

\section{Proof Overview of $\PSPACE$-hardness of Approximation}
\label{sec:overview-Cut-hard}

In this section,
we give an overview of the proof of \cref{thm:intro:Cut-hard}; i.e.,
\MMkCutReconf is $\PSPACE$-hard to approximate within a factor of $1-\Omega\left(\frac{1}{k}\right)$.
For a graph $G$ and a pair of its $k$-colorings $\f_\sss,\f_\ttt \colon V(G) \to [k]$,
let $\opt_G(\f_\sss \reco \f_\ttt)$ denote
the \emph{optimal value} of \MMkCutReconf; namely,
the maximum of the minimum fraction of bichromatic edges of $G$,
where the maximum is taken over all possible reconfiguration sequences
from $\f_\sss$ to $\f_\ttt$.
For any reals $0 \leq s \leq c \leq 1$,
\prb{Gap$_{c,s}$ \kCutReconf} asks
whether $\opt_G(\f_\sss \reco \f_\ttt) \geq c$ or
$\opt_G(\f_\sss \reco \f_\ttt) < s$.
See \cref{sec:pre} for the formal definition.

Our starting point is the $\PSPACE$-hardness of approximating \MMtwoCutReconf,
whose proof is based on \cite{bonsma2009finding,hirahara2024probabilistically,ohsaka2023gap}.

\begin{proposition}[informal; see \cref{prp:Cut-hard:2Cut}]
\label{prp:intro:Cut-hard:2Cut}
There exist universal constants $\epsilon_c,\epsilon_s \in (0,1)$ with $\epsilon_c < \epsilon_s$ such that
\prb{Gap$_{1-\epsilon_c,1-\epsilon_s}$ \twoCutReconf}
is $\PSPACE$-hard.
\end{proposition}

\noindent
We construct the following two gap-preserving reductions from \MMtwoCutReconf to \MMkCutReconf,
the former for all sufficiently large $k$ and
the latter for finitely many $k$.

\begin{lemma}[informal; see \cref{lem:Cut-hard:crazy}]
\label{lem:intro:Cut-hard:crazy}
    For every reals $\epsilon_c,\epsilon_s \in (0,1)$ with $\epsilon_c < \epsilon_s$,
    there exist reals $\delta_c, \delta_s \in (0,1)$ with $\delta_c < \delta_s$
    such that
    for all sufficiently large $k \geq k_0 \defeq \kzero$,
    there exists a gap-preserving reduction from
    \prb{Gap$_{1-\epsilon_c,1-\epsilon_s}$ \twoCutReconf}
    to
    \prb{Gap$_{1-\frac{\delta_c}{k},1-\frac{\delta_s}{k}}$ \kCutReconf}.
\end{lemma}

\begin{lemma}[informal; see \cref{lem:Cut-hard:quadratic}]
\label{lem:intro:Cut-hard:quadratic}
    For every integer $k \geq 3$ and
    every reals $\epsilon_c,\epsilon_s \in (0,1)$ with $\epsilon_c < \epsilon_s$,
    there exist universal constants $\delta_c,\delta_s \in (0,1)$ with $\delta_c<\delta_s$ such that
    there exists a gap-preserving reduction from
    \prb{Gap$_{1-\epsilon_c, 1-\epsilon_s}$ \twoCutReconf}
    to
    \prb{Gap$_{1-\delta_c, 1-\delta_s}$ \kCutReconf}.
\end{lemma}

\noindent
We obtain \cref{thm:intro:Cut-hard} as a corollary of 
\cref{prp:intro:Cut-hard:2Cut,lem:intro:Cut-hard:crazy,lem:intro:Cut-hard:quadratic}.
Since the most technical part in the proof of \cref{thm:intro:Cut-hard} is \cref{lem:intro:Cut-hard:crazy},
we will outline its proof in the remainder of this section.
See \cref{app:Cut-hard} for
the proofs of \cref{prp:intro:Cut-hard:2Cut,lem:intro:Cut-hard:quadratic}.

\subsection{Failed Attempt: Why \texorpdfstring{\cite{kann1997hardness,guruswami2013improved}}{[GS13, KKLP97]} Do Not Work for Proving \texorpdfstring{\cref{lem:intro:Cut-hard:crazy}}{Lemma~\protect\ref{lem:intro:Cut-hard:crazy}}}
\label{sec:overview-Cut-hard:failed}

To prove \cref{lem:intro:Cut-hard:crazy},
one might think of applying the existing proof techniques for
the $\NP$-hardness of approximating \prb{Max $k$-Cut},
which has the (asymptotically) same approximation threshold of $1-\Theta\left(\frac{1}{k}\right)$ as \MMkCutReconf
\cite{kann1997hardness,guruswami2013improved,austrin2014new,frieze1997improved}
(see \cref{sec:related} for related work).
However, this approach \emph{does not work} for proving \cref{lem:intro:Cut-hard:crazy}:
when a gap-preserving reduction from \prb{Max 2-Cut} to \prb{Max $k$-Cut}
due to \cite{kann1997hardness,guruswami2013improved}
is used to reduce \MMtwoCutReconf to \MMkCutReconf,
the ratio between completeness and soundness becomes $1 - \bigO\left(\frac{1}{k^2}\right)$.

To explain the detail,
we briefly review the gap-preserving reduction of \citet{kann1997hardness}.\footnote{
    The reduction of \citet{guruswami2013improved} differs from that of \cite{kann1997hardness} in that
    it starts from \prb{Max $3$-Cut} to preserve the perfect completeness.
}
For a graph $G=(V,E)$ and a positive even integer $k$,
a new weighted graph $H$ is constructed as follows.
\begin{itemize}
\item 
    Create fresh $\frac{k}{2}$ copies of each vertex $v$ of $G$,
    denoted by $v_{1}, \ldots, v_{\frac{k}{2}}$.
\item
    For each edge $(v,w)$ of $G$ and pair $i,j \in \left[\frac{k}{2}\right]$,
    create an edge $(v_i, w_j)$ of weight $1$.
\item
    For each vertex $v$ of $G$ and pair $i \neq j \in \left[\frac{k}{2}\right]$,
    create an edge $(v_i, v_j)$ of weight equal to the degree of $v$.
\end{itemize}
See \cref{fig:failed:G,fig:failed:H} for illustration.
The total edge weight of $H$ is equal to ${k \choose 2} \cdot |E|$.
By \cite{kann1997hardness}, this construction is
an approximation-preserving reduction from \prb{Max 2-Cut} to \prb{Max $k$-Cut},
implying the $\NP$-hardness of
$\left(1-\Omega\left(\frac{1}{k}\right)\right)$-factor approximation for \prb{Max $k$-Cut}.

\input{fig-failed}

Let us apply the above reduction to reduce
\MMtwoCutReconf to \MMkCutReconf.
Given a graph $G=(V,E)$ and a pair of its proper $2$-colorings $\f_\sss,\f_\ttt \colon V \to [2]$
as an instance of \MMtwoCutReconf,
we construct an instance of \MMkCutReconf as follows.
First, create a weighted graph $H$ from $G$ according to \cite{kann1997hardness}.
Then, create a pair of $k$-colorings $\f'_\sss,\f'_\ttt \colon V(H) \to [k]$ of $H$
in a natural manner such that
$\f'_\sss(v_i) \defeq \f_\sss(v) + 2(i-1)$ and
$\f'_\ttt(v_i) \defeq \f_\ttt(v) + 2(i-1)$ for each vertex $v_i$ of $H$.
See \cref{fig:failed:G1,fig:failed:G2,fig:failed:H1,fig:failed:H2} for illustration.
For each $i \in \left[\frac{k}{2}\right]$, we define $V_i \defeq \{v_i \mid v \in V \}$.
Observe that
$\f'_\sss$ and $\f'_\ttt$ are proper if so are $\f_\sss$ and $\f_\ttt$, and
every vertex of $V_i$ is colored in $2i-1$ or $2i$.
Here, we claim that
$\opt_H(\f'_\sss \reco \f'_\ttt) \geq 1 - \frac{2}{k(k-1)}$
independent of the value of $\opt_G(\f_\sss \reco \f_\ttt)$.
Consider a reconfiguration sequence $\sqcol'$ from $\f'_\sss$ to $\f'_\ttt$
obtained by recoloring vertices of $V_1, V_2, \ldots, V_{\frac{k}{2}}$ in this order.
Suppose we are on the way of recoloring the vertices of $V_i$.
The subgraph $H[V_i]$ may contain (at most) $|E|$ monochromatic edges, but
all other $(\frac{k}{2}-1)$ subgraphs $H[V_j]$ for $j \neq i$ \emph{do not} contain any monochromatic edges,
deriving that
\begin{align}
    \opt_H\bigl(\f'_\sss \reco \f'_\ttt\bigr)
    \geq 1 - \frac{1 \cdot |E| + \left(\frac{k}{2}-1\right) \cdot 0}{{k \choose 2} \cdot |E|}
    \geq 1 - \frac{2}{k(k-1)}.
\end{align}
This is undesirable because
the ratio between completeness and soundness is at least 
$1 - \bigO\left(\frac{1}{k^2}\right)$.

\subsection{Our Reduction in the Proof of \texorpdfstring{\cref{lem:intro:Cut-hard:crazy}}{Lemma~\protect\ref{lem:intro:Cut-hard:crazy}}}
\label{sec:overview-Cut-hard:our}

\input{fig-motive}

Our gap-preserving reduction from \MMtwoCutReconf to \MMkCutReconf
is completely different from those of \cite{kann1997hardness,guruswami2013improved}.
Briefly speaking,
    we shall encode a $2$-coloring of each vertex $v$ of a graph $G$ by a $k$-coloring of a $k\times k$ grid $[k]^2$.
Our proposed encoding is motivated by the following scenario:
Suppose that for a graph $G=(V,E)$ and a pair of its proper $k$-colorings $\f,\g \colon V \to [k]$,
    we would like to find an optimal reconfiguration sequence from $\f$ to $\g$
    (see \cref{fig:motive:f,fig:motive:g}).
For each pair of colors $\alpha, \beta \in [k]$,
    let $V_{\alpha,\beta}$ be the set of vertices in $V$ colored $\alpha$ by $\f$ and $\beta$ by $\g$
    (see \cref{fig:motive:Vs}); namely,
\begin{align}
    V_{\alpha,\beta} \defeq \Bigl\{
        v \in V \Bigm| \f(v) = \alpha \text{ and } \g(v) = \beta
    \Bigr\}.
\end{align}
If $V_{\alpha,\beta}$'s are placed on a $k \times k$ grid,
    $\f$ looks ``horizontally striped'' while
    $\g$ looks ``vertically striped''
    (see \cref{fig:motive:horizontal,fig:motive:vertical}).
Since both $\f$ and $\g$ are proper,
there may exist edges between $V_{\alpha_1,\beta_1}$ and $V_{\alpha_2,\beta_2}$ \emph{only if}
$\alpha_1 \neq \alpha_2$ and $\beta_1 \neq \beta_2$ (see \cref{fig:motive:grid}).
On the other hand, any reconfiguration sequence from $\f$ to $\g$ seems to
make a nonnegligible fraction of edges into monochromatic.
The above structural observation motivates the following two ideas:

\begin{description}
    \item[Idea 1:]
        Consider the ``striped'' pattern represented by
        a $k$-coloring of $[k]^2$
        as if it were encoding $[2]$; i.e.,
        the ``horizontally striped'' pattern represents $1$, whereas
        the ``vertically striped'' pattern represents $2$.
        This encoding can be thought of as a \emph{very redundant} error-correcting code
        from $[2]$ to $[k]^{[k]^2}$.
    
    \item[Idea 2:]
        Given a graph $G=(V,E)$ and
        a collection of $|V|$ $k$-colorings of $[k]^2$ for each vertex of $G$,
        we test if these $k$-colorings encode a \emph{proper} $2$-coloring of $G$.
        Specifically, we will design a probabilistic verifier that checks if
        (1) a $k$-coloring of $[k]^2$ associated with each vertex of $G$ is close to a striped pattern, and
        (2) a pair of $k$-colorings of $[k]^2$ corresponding to each edge of $G$ encode different colors.
        In the subsequent sections,
        we will introduce the following three auxiliary verifiers to achieve this requirement:
        \textbf{Stripe}, \textbf{consistency}, and \textbf{edge verifiers}.
\end{description}

We will say that a $k$-coloring $\f \colon [k]^2 \to [k]$
is \emph{horizontally striped} if 
$\f(x,y) = \sigma(y)$ for all $(x,y) \in [k]^2$
for some permutation $\sigma \in \frakS_k$,
\emph{vertically striped} if
$\f(x,y) = \sigma(x)$ for all $(x,y) \in [k]^2$
for some permutation $\sigma \in \frakS_k$, and
\emph{striped} if
it is horizontally or vertically striped.

\subsubsection{Stripe Test (\cref{sec:Cut-hard:tests:stripe})}
Our first, most important verifier is the \emph{stripe verifier} $\Vstripe$, which checks
if a $k$-coloring $\f$ of $[k]^2$ is close to a striped pattern.
Specifically,
$\Vstripe$ samples a pair of vertices from $[k]^2$ that forms a \emph{diagonal line} in a $k \times k$ grid, and
it accepts if they have different colors, as follows:

\begin{itembox}[l]{\textbf{Stripe verifier $\Vstripe$.}}
\begin{algorithmic}[1]
    \item[\textbf{Oracle access:}]
        a $k$-coloring $\f \colon [k]^2 \to [k]$.
    \State select $(x_1,y_1) \in [k]^2$ and $(x_2,y_2) \in [k]^2$
        s.t.~$x_1 \neq x_2$ and $y_1 \neq y_2$ uniformly at random.
    \If{$\f(x_1,y_1) = \f(x_2,y_2)$}
        \State declare \Reject.
    \Else
        \State declare \Accept.
    \EndIf
\end{algorithmic}
\end{itembox}

\noindent
We say that a $k$-coloring $\f \colon [k]^2 \to [k]$
is \emph{$\epsilon$-far from being striped} if
$\f$ is $\epsilon$-far from every striped $k$-coloring, and
is \emph{$\epsilon$-close to being striped} if
$\f$ is $\epsilon$-close to some striped $k$-coloring.
See \cref{fig:far} for an example of a $k$-coloring of $[k]^2$ far from being striped.

\input{fig-far}

The following lemma is the crux of the proof of \cref{lem:intro:Cut-hard:crazy},
which bounds $\Vstripe$'s rejection probability with respect to the distance from $\f$ to the striped pattern:

\begin{lemma}[informal; see \cref{lem:Cut-hard:stripe:striped,lem:Cut-hard:stripe:far}]
\label{lem:intro:Cut-hard:stripe}
The following hold\textup{:}
\begin{itemize}
    \item
    if $f$ is striped,
    $\Vstripe$ accepts with probability $1$\textup{;}
    \item
    if $f$ is $\epsilon$-far from being striped,
    $\Vstripe$ rejects with probability
    $\Omega\left(\frac{\epsilon}{k}\right)$.
\end{itemize}
\end{lemma}\noindent
The rejection probability ``$\Omega\left(\frac{\epsilon}{k}\right)$'' is critical for deriving
a $\left(1-\Omega\left(\frac{1}{k}\right)\right)$-factor gap between completeness and soundness.
The latter statement of \cref{lem:intro:Cut-hard:stripe} presented in \cref{sec:Cut-hard:stripe:far}
involves the most technical proof in this paper,
exploiting the nontrivial structure of a $k$-coloring of $[k]^2$ far from being striped.

Observe that $\Vstripe$ is only allowed to
sample a pair $(v,w)$ of vertices from $[k]^2$ (nonadaptively) and
accepts (resp.~rejects) if $\f(v) \neq f(w)$ (resp.~$\f(v) = \f(w)$).
Thus, $\Vstripe$ can be ``emulated'' by a graph $H$ such that
\begin{align}
    V(H) & \defeq [k]^2, \\
    E(H) & \defeq \Bigl\{
        \bigl((x_1, y_1), (x_2, y_2)\bigr) \in \left([k]^2\right)^2
        \Bigm| x_1 \neq x_2 \text{ and } y_1 \neq y_2
    \Bigr\},
\end{align}
in a sense that 
for any $k$-coloring $\f$ of $[k]^2$,
the probability that $\Vstripe$ accepts (resp.~rejects) $\f$
is equal to
the fraction of edges in $H$ that are made bichromatic (resp.~monochromatic) by $\f$.
In fact, the graph structure of \cref{fig:motive:grid} coincides with $H$.
The remaining two verifiers can also be emulated by (multi)graphs.

\subsubsection{Consistency Test (\cref{sec:Cut-hard:tests:cons})}
Our next verifier is the \emph{consistency verifier} $\Vcons$, which checks
if a pair of $k$-colorings $\f,\g$ of $[k]^2$ share the \emph{same} striped pattern
(given that both $\f$ and $\g$ are close to being striped).
Specifically,
$\Vcons$ runs
the \emph{row test} and \emph{column test} with equal probability,
the former for the horizontally striped pattern and 
the latter for the vertically striped pattern,
as follows:

\begin{itembox}[l]{\textbf{Consistency verifier $\Vcons$.}}
\begin{algorithmic}[1]
    \item[\textbf{Oracle access:}]
        two $k$-colorings $\f, \g \colon [k]^2 \to [k]$.
    \State sample $r \sim [0,1]$.
    \If{$0 \leq r < \frac{1}{2}$} \Comment{run the row test with probability $\frac{1}{2}$.}
        \State select $(x_1,y_1) \in [k]^2$ and $(x_2,y_2) \in [k]^2$
            s.t.~$y_1 \neq y_2$ uniformly at random.
    \Else \Comment{run the column test with probability $\frac{1}{2}$.}
        \State select $(x_1,y_1) \in [k]^2$ and $(x_2,y_2) \in [k]^2$
            s.t.~$x_1 \neq x_2$ uniformly at random.
    \EndIf
    \If{$\f(x_1,y_1) = \g(x_2,y_2)$}
        \State declare \Reject.
    \Else
        \State declare \Accept.
    \EndIf
\end{algorithmic}
\end{itembox}

\noindent
A pair of $k$-colorings $\f,\g \colon [k]^2 \to [k]$ are said to be
\emph{consistent} if 
both $\f$ and $\g$ are
closest to horizontally striped $k$-colorings or 
closest to vertically striped $k$-colorings, and
\emph{inconsistent} otherwise.
The following lemma bounds $\Vcons$'s rejection probability.

\begin{lemma}[informal;  see \cref{lem:Cut-hard:cons:striped,lem:Cut-hard:cons:far}]
\label{lem:intro:Cut-hard:cons}
\mbox{}

    Suppose $\f$ and $\g$ are striped.
    Then, the following hold\textup{:}
    \begin{itemize}
        \item if $\f = \g$
            (i.e., $\f$ and $\g$ have the same striped $k$-coloring),
            $\Vcons$ rejects 
            with probability exactly $\frac{1}{2k}$\textup{;}
        \item if $\f$ and $\g$ are inconsistent,
            $\Vcons$ rejects 
            with probability exactly $\frac{1}{k}$.
    \end{itemize}

    Suppose $\f$ and $\g$ are $\epsilon$-close to being striped.
    Then, the following hold\textup{:}
    \begin{itemize}
    \item if $\f$ and $\g$ are consistent,
        $\Vcons$ rejects 
        with probability more than
        $( 1-4 \epsilon ) \cdot \frac{1}{2k}$\textup{;}
    \item if $\f$ and $\g$ are inconsistent,
        $\Vcons$ rejects 
        with probability more than
        $( 1-4 \epsilon ) \cdot \frac{1}{k}$.
    \end{itemize}
\end{lemma}

Since \cref{lem:intro:Cut-hard:cons} does not bound $\Vcons$'s rejection probability from below
if $\f$ and $\g$ are far from being striped,
we will combine $\Vstripe$ and $\Vcons$ in the third test.

\subsubsection{Edge Test (\cref{sec:Cut-hard:tests:edge})}
Our final verifier is the \emph{edge verifier} $\Vedge$,
which checks if a pair of $k$-colorings $\f,\g$ of $[k]^2$ are \emph{close to} the same striped pattern.
To this end, $\Vedge$ calls the stripe verifier $\Vstripe$ and the consistency verifier $\Vcons$
with a carefully designed probability, as follows:\footnote{
The value of $\rho$ denotes the hidden constant in $\Omega\left(\frac{\epsilon}{k}\right)$ of \cref{lem:intro:Cut-hard:stripe}.
}

\begin{itembox}[l]{\textbf{Edge verifier $\Vedge$.}}
\begin{algorithmic}[1]
    \item[\textbf{Oracle access:}]
        two $k$-colorings $\f,\g \colon [k]^2 \to [k]$.
    \State let $\rho \defeq \rhozero$ \textbf{and} $Z \defeq \frac{2}{\rho} + \frac{2}{\rho} + 1$.
    \State sample $r \sim [0,1]$.
    \If{$0 \leq r < \frac{2}{\rho Z}$} \Comment{with probability $\frac{2}{\rho Z}$}
        \State execute $\Vstripe$ on $\f$.
    \ElsIf{$\frac{2}{\rho Z} \leq r < \frac{2}{\rho Z} + \frac{2}{\rho Z}$} \Comment{with probability $\frac{2}{\rho Z}$}
        \State execute $\Vstripe$ on $\g$.
    \Else \Comment{with probability $\frac{1}{Z}$}
        \State execute $\Vcons$ on $\f \circ \g$.
    \EndIf
\end{algorithmic}
\end{itembox}

\noindent
The following lemma bounds $\Vedge$'s rejection probability.

\begin{lemma}[informal; see \cref{lem:Cut-hard:edge:striped,lem:Cut-hard:edge:mismatch,lem:Cut-hard:edge:any}]
\label{lem:intro:Cut-hard:edge}
    The following hold\textup{:}
    \begin{itemize}
        \item if $\f$ and $\g$ are striped and $\f = \g$
            (i.e., $\f$ and $\g$ have the same striped $k$-coloring),
            $\Vedge$ rejects with probability at most $\frac{1}{2Z \cdot k}$\textup{;}
        \item if
            $\f$ and $\g$ are inconsistent,
            $\Vedge$ rejects with probability at least $\frac{1}{Z \cdot k}$\textup{;}
        \item 
            $\Vedge$ always rejects with probability at least $\frac{1}{2Z \cdot k}$.
    \end{itemize}
\end{lemma}\noindent

\subsubsection{Putting Them Together (\cref{sec:Cut-hard:crazy})}
We are now ready to
reduce \MMtwoCutReconf to \MMkCutReconf
to accomplish the proof of \cref{lem:intro:Cut-hard:crazy}.
Given a graph $G=(V,E)$ and a pair of its $2$-colorings $\f_\sss,\f_\ttt \colon V \to [2]$
as an instance of \MMtwoCutReconf,
we construct a new (multi)graph $H$ and a pair of its $k$-colorings $\f'_\sss, \f'_\ttt \colon V(H) \to [k]$
as an instance of \MMkCutReconf as follows.
For each vertex $v$ of $G$, we create a fresh copy of a $k\times k$ grid $[k]^2$; namely,
the vertex set of $H$ is defined as
\begin{align}
    V(H) \defeq V \times [k]^2.
\end{align}
Since a $k$-coloring $\f' \colon V \times [k]^2 \to [k]$ of $H$
consists of $|V|$ $k$-colorings of $[k]^2$,
we will think of it as a function $\f' \colon V \to ([k]^2 \to [k])$ such that
$f'(v)$ gives a $k$-coloring of $[k]^2$ associated with $v \in V$.

Consider the following verifier $\V_G$,
given oracle access to a function $\f' \colon V \to ([k]^2 \to [k])$,
which
samples an edge $(v,w)$ from $G$ and
runs $\Vedge$ on $\f'(v) \circ \f'(w)^\top$:\footnote{
$\f'(w)^\top$ is the \emph{transposition} of $\f'(w)$; i.e.,
$\f'(w)^\top(x,y) = \f'(w)(y,x)$ for all $(x,y) \in [k]^2$.
The transposition comes from the design of $\Vedge$
to check the \emph{consistency} between a pair of $k$-colorings,
whereas we here need to check the \emph{inconsistency}.
}

\begin{itembox}[l]{\textbf{Overall verifier $\V_G$.}}
\begin{algorithmic}[1]
    \item[\textbf{Input:}]
        a graph $G = (V,E)$.
    \item[\textbf{Oracle access:}]
        a function $\f' \colon V \to ([k]^2 \to [k])$.
    \State select an edge $(v,w)$ of $G$ uniformly at random.
    \State execute $\Vedge$ on $\f'(v) \circ \f'(w)^\top$.
\end{algorithmic}
\end{itembox}

\noindent
It is not hard to generate the set $E(H)$ of (parallel) edges between $V(H)$ so as to
emulate $\V_G$ in that for any $k$-coloring $\f' \colon V \times [k]^2 \to [k]$,
the fraction of bichromatic edges in $E(H)$
is equal to 
the acceptance probability of $\V_G$.
Construct a pair of $k$-colorings $\f'_\sss, \f'_\ttt \colon V \to ([k]^2 \to [k])$ of $H$
such that for each vertex $v$ of $G$,
we define
$\f'_\sss(v)$ (resp.~$\f'_\ttt(v)$) 
to be
horizontally striped if $\f_\sss(v)$ (resp.~$\f_\ttt(v)$) is $1$ and
  vertically striped if $\f_\sss(v)$ (resp.~$\f_\ttt(v)$) is $2$.
This completes the description of the reduction.
See \cref{fig:crazy} for illustration.
Our reduction enjoys the following gap-preserving property:

\input{fig-crazy}

\begin{lemma}[informal; see \cref{lem:Cut-hard:complete,lem:Cut-hard:sound}]
\label{lem:intro:Cut-hard:complete-sound}
    The following hold\textup{:}
    \begin{align}
        \opt_G\bigl(\f_\sss \reco \f_\ttt\bigr) \geq 1-\epsilon_c
        & \implies \opt_H\bigl(\f'_\sss \reco \f'_\ttt\bigr) \geq 1-\frac{1+\epsilon_c}{2Z \cdot k} - o(1), \\
        \opt_G\bigl(\f_\sss \reco \f_\ttt\bigr) < 1-\epsilon_s
        & \implies \opt_H\bigl(\f'_\sss \reco \f'_\ttt\bigr) < 1 - \frac{1+\epsilon_s}{2Z \cdot k},
    \end{align}
    where
    $Z \defeq \frac{2}{\rho} + \frac{2}{\rho} + 1$ and $\rho \defeq \rhozero$.
\end{lemma}\noindent
The proof of \cref{lem:intro:Cut-hard:complete-sound} relies on \cref{lem:intro:Cut-hard:edge},
and the proof of \cref{lem:intro:Cut-hard:crazy} is almost immediate from \cref{lem:intro:Cut-hard:complete-sound}.

\begin{remark}
Our reduction can also be used to reduce
\prb{Max 2-Cut} to \prb{Max $k$-Cut} in a gap-preserving manner (\cref{lem:Cut-hard:KKLP97}),
which reproves that
\prb{Max $k$-Cut} is $\NP$-hard to approximate within a factor of $1 - \Omega\left(\frac{1}{k}\right)$ \cite{kann1997hardness,guruswami2013improved}.
Since the existing reductions for \prb{Max $k$-Cut} due to \cite{kann1997hardness,guruswami2013improved}
do not work for \MMkCutReconf,
the present study demonstrates
the inherent difficulty in designing approximation-preserving reductions between reconfiguration problems.
\end{remark}

\section{Proof Overview of Approximation Algorithm}
\label{sec:overview-Cut-alg}

In this section, we present a highlight of the proof of \cref{thm:intro:Cut-alg}, i.e.,
a deterministic $\left(1-\frac{2}{k}\right)$-factor approximation algorithm for \MMkCutReconf.
Our approximation algorithm is based on \emph{a random reconfiguration via a random solution}.
Let $G = (V,E)$ be a graph and $\f_\sss, \f_\ttt \colon V \to [k]$ be a pair of its $k$-colorings.
We assume $\f_\sss$ and $\f_\ttt$ are proper for the sake of simplicity
(see \cref{lem:Cut-alg:low-value} for how to address when
$\f_\sss$ and $\f_\ttt$ contain many monochromatic edges).
Let $\frnd \colon V \to [k]$ be a \emph{random} $k$-coloring of $G$,
which makes $\left(1-\frac{1}{k}\right)$-fraction of edges of $G$ bichromatic in expectation.
Consider now the following two \emph{random} reconfiguration sequences:
\begin{itemize}
    \item a reconfiguration sequence $\sqcol_1$ from $\f_\sss$ to $\frnd$
obtained by recoloring vertices at which
$\f_\sss$ and $\frnd$ differ in a random order, and
    \item a reconfiguration sequence $\sqcol_2$ from $\frnd$ to $\f_\ttt$
obtained by recoloring vertices at which
$\frnd$ and $\f_\ttt$ differ in a random order.
\end{itemize}
Concatenating $\sqcol_1$ and $\sqcol_2$,
we obtain a random reconfiguration sequence $\sqcol$ from $\f_\sss$ to $\f_\ttt$.
It is easy to prove that
for each edge $e$ of $G$,
\emph{all} $k$-colorings of $\sqcol$ make $e$ bichromatic with probability at least $1-\frac{9}{k}$
(\cref{obs:Cut-alg:simple}).
In particular, $\sqcol$ already achieves a $\left(1-\frac{9}{k}\right)$-factor approximation for \MMkCutReconf in expectation.
Note that \citet{karthik2023inapproximability}
used a similar strategy to approximate \prb{Maxmin 2-CSP Reconfiguration},
which constructs a reconfiguration sequence that goes through a random assignment in a \emph{greedy} manner.

Separately deriving \emph{concentration bounds} for each $\sqcol_1$ and $\sqcol_2$,
we further improve the approximation factor from $1-\frac{9}{k}$ to $1-\frac{2}{k}$.
Our crucial insight for this purpose is
to partition the vertex set of $G$ into the low-degree and high-degree sets.
We say that a vertex of $G$ is
\emph{low degree} if its degree is less than $|E|^\frac{2}{3}$ and
\emph{high degree} otherwise.

\begin{itemize}
\item
Suppose first $G$ contains only low-degree vertices.
By case analysis, we can show that
each edge is always bichromatic within $\sqcol_1$
with probability at least $\left(1-\frac{1}{k}\right)^2 = 1-\frac{2}{k}+\frac{1}{k^2}$.
By applying the read-$k$ Chernoff bound \cite{gavinsky2015tail} with parameter $|E|^\frac{2}{3}$,
we obtain that
every $k$-coloring of $\sqcol_1$
makes at least $\left(1-\frac{2}{k}\right)$-fraction of edges bichromatic
with high probability.
The same result holds for $\sqcol_2$.

\item
Suppose now $G$ contains high-degree vertices,
for which a direct application of the read-$k$ Chernoff bound does not yield useful concentration bounds.
We resort to the following ad-hoc observations,
which are reminiscent of those for \prb{Maxmin $2$-CSP Reconfiguration} due to \cite{karthik2023inapproximability}:
\begin{enumerate}
    \item
    Since there are ``few'' high-degree vertices, the number of edges between them is negligible.
    \item
    Each high-degree vertex has ``many'' low-degree neighbors,
    whose colors assigned by $\frnd$ are distributed almost evenly;
    thus, a nearly $\left(1-\frac{1}{k}\right)$-fraction of edges between
    high-degree vertices and low-degree vertices
    are bichromatic with high probability.
\end{enumerate}
In light of the second observation,
we generate a reconfiguration sequence $\sqcol_1$ from $\f_\sss$ to $\frnd$
by first recoloring low-degree vertices followed by high-degree vertices, and
a reconfiguration sequence $\sqcol_2$ from $\frnd$ to $\f_\ttt$
by first recoloring high-degree vertices followed by low-degree vertices.
\end{itemize}

The following randomized algorithm 
generates a random reconfiguration sequence $\sqcol$ from $\f_\sss$ to $\f_\ttt$, which
guarantees a $\left(1-\frac{2}{k}\right)$-factor approximation for \MMkCutReconf
with high probability:

\begin{itembox}[l]{\textbf{Generating a random reconfiguration sequence $\sqcol$ from $\f_\sss$ to $\f_\ttt$.}}
\begin{algorithmic}[1]
    \item[\textbf{Input:}]
        a graph $G = (V,E)$ and two $k$-colorings $\f_\sss, \f_\ttt \colon V \to [k]$ of $G$.
    \State sample a random $k$-coloring $\frnd \colon V \to [k]$ of $G$.
    \LComment{start from $\f_\sss$.}
    \State recolor each low-degree vertex $v$ from $\f_\sss(v)$ to $\frnd(v)$ in a random order.
    \State recolor each high-degree vertex $v$ from $\f_\sss(v)$ to $\frnd(v)$ in a random order.
    \LComment{obtain $\frnd$.}
    \State recolor each high-degree vertex $v$ from $\frnd(v)$ to $\f_\ttt(v)$ in a random order.
    \State recolor each low-degree vertex $v$ from $\frnd(v)$ to $\f_\ttt(v)$ in a random order.
    \LComment{end at $\f_\ttt$.}
\end{algorithmic}
\end{itembox}

\noindent
Our deterministic algorithm is obtained by derandomizing the above algorithm,
which can be done by the method of conditional expectations \cite{alon2016probabilistic}.

\section{Related Work}
\label{sec:related}

\subsection{Variants of \kColReconf}
Other than reachability problems,
there are several types of reconfiguration problems \cite{mouawad2015reconfiguration,nishimura2018introduction,heuvel13complexity}.
One is \emph{connectivity problems},
which ask if the configuration graph is connected.
In the connectivity variant of \kColReconf,
we are asked to decide if 
\emph{every} pair of proper $k$-colorings of a graph $G$ are reconfigurable each other.
Such a graph $G$ is said to be \emph{$k$-mixing}.
On the complexity side,
it is $\coNP$-hard to decide if a graph is $k$-mixing for every $k \geq 3$ \cite{cereceda2009mixing,bousquet2024note}.
The name of $k$-mixing comes from the relation to
the (rapid) mixing of the Glauber dynamics
\cite{jerrum1995very,molloy2004glauber,dyer2006randomly}.
The \emph{Glauber dynamics} is a Markov Chain such that
starting from a graph $G$ and a proper $k$-coloring of $G$,
we repeatedly recolor a random vertex with a random color
(as long as it yields a proper $k$-coloring).
The Glauber dynamics is \emph{ergodic} only if $G$ is $k$-mixing.

Other algorithmic and structural problems related to \kColReconf include
finding the shortest reconfiguration sequence \cite{cereceda2011finding,bonamy2020shortest,johnson2016finding} and
bounding the diameter of the configuration graph
\cite{cereceda2011finding,bonsma2009finding,bonamy2011diameter,bonamy2014reconfiguration}, respectively.
See also \citet[\S6]{nishimura2018introduction},
\citet[\S3]{heuvel13complexity}, and
\citet{mynhardt2019reconfiguration}.

\subsection{Approximability of \prb{Max $k$-Cut}}
The \prb{Max $k$-Cut} problem
(a.k.a.~\prb{Max $k$-Colorable Subgraph} \cite{papadimitriou1991optimization,guruswami2013improved})
seeks a $k$-coloring of a graph that makes the maximum number of edges bichromatic.
Observe easily that
a random $k$-coloring makes a $\left(1-\frac{1}{k}\right)$-fraction of edges bichromatic in expectation;
moreover, \citet{frieze1997improved} developed a $\left( 1-\frac{1}{k}+\frac{2 \ln k}{k^2} \right)$-factor approximation algorithm.
On the hardness side,
$\left( 1-\frac{1}{17k + \bigO(1)} \right)$-factor approximation is $\NP$-hard
\cite{kann1997hardness,guruswami2013improved,austrin2014new}.
For the special case of $k=2$, i.e., \prb{Max Cut},
the current best approximation factor is $\approx 0.878$ \cite{goemans1995improved},
which is proven to be optimal \cite{khot2007optimal,mossel2010noise}
under the Unique Games Conjecture \cite{khot2002power}.

\subsection{Approximability of Reconfiguration Problems}
\citet{ito2011complexity}
proved that several reconfiguration problems
(e.g., \prb{Maxmin SAT Reconfiguration})
are $\NP$-hard to approximate
relying on the $\NP$-hardness of approximating the source problems (e.g., \prb{Max SAT}).
Since most reconfiguration problems are $\PSPACE$-complete,
$\NP$-hardness results are not optimal.
In fact, \cite{ito2011complexity} posed $\PSPACE$-hardness of approximation as an open problem.

Motivated by $\PSPACE$-hardness of approximation for reconfiguration problems,
\citet{ohsaka2023gap} postulated
a reconfiguration analogue of the PCP theorem \cite{arora1998probabilistic,arora1998proof},
called the \emph{Reconfiguration Inapproximability Hypothesis} (RIH).
Under RIH, (approximate versions of) several reconfiguration problems are $\PSPACE$-hard to approximate,
including those of 
\prb{3-SAT}, \prb{Independent Set}, \prb{Vertex Cover}, and \prb{Clique}.
Very recently,
\citet{hirahara2024probabilistically} and \citet{karthik2023inapproximability}
independently gave a proof of RIH
by establishing the \emph{Probabilistically Checkable Reconfiguration Proof} (PCRP) theorem,
which provides a new PCP-type characterization of $\PSPACE$.
The PCRP theorem, along with a series of gap-preserving reductions 
\cite{ohsaka2023gap,ohsaka2024gap,hirahara2024probabilistically,hirahara2024optimal},
implies
\emph{unconditional} $\PSPACE$-hardness of approximation results for many reconfiguration problems,
thereby resolving the open problem of \cite{ito2011complexity} affirmatively. 

One recent trend regarding approximability of reconfiguration problems is to 
prove an explicit factor of $\PSPACE$-hardness of approximation.
In the $\NP$ regime,
the \emph{parallel repetition theorem} \cite{raz1998parallel} can be used to derive
explicit, strong inapproximability results
\cite{hastad1999clique,hastad2001some,feige1998threshold,bellare1998free,zuckerman2007linear}.
Unfortunately, a naive parallel repetition does not
reduce the soundness error of a reconfiguration analogue of two-prover games \cite{ohsaka2025approximate}.
\citet{ohsaka2024gap} adapted \citeauthor{dinur2007pcp}'s gap amplification \cite{dinur2007pcp,radhakrishnan2006gap,radhakrishnan2007dinurs}
to show that
\prb{Maxmin 2-CSP Reconfiguration} and \prb{Minmax Set Cover Reconfiguration} are 
$\PSPACE$-hard to approximate within a factor of $0.9942$ and $1.0029$, respectively.
Subsequently, \citet{karthik2023inapproximability} showed that
\prb{Minmax Set Cover Reconfiguration} is $\NP$-hard to approximate within
a factor of $2-\epsilon$ for every $\epsilon > 0$.
\citet{hirahara2024optimal} demonstrated that 
\prb{Minmax Set Cover Reconfiguration}
is $\PSPACE$-hard to approximate within a factor of $2-o(1)$,
improving upon \cite{karthik2023inapproximability,ohsaka2024gap}.
Since \prb{Minmax Set Cover Reconfiguration} admits
a $2$-factor approximation algorithm \cite{ito2011complexity},
this is the first optimal $\PSPACE$-hardness result for approximability of any reconfiguration problem.

Approximation algorithms have been developed for several reconfiguration problems; e.g.,
\prb{Maxmin 2-CSP Reconfiguration} admits
    a $\left(\frac{1}{2} - \epsilon\right)$-factor approximation \cite{karthik2023inapproximability},
\prb{Subset Sum Reconfiguration} admits
    a PTAS \cite{ito2014approximability}, and
\prb{Submodular Reconfiguration} admits
    a constant-factor approximation \cite{ohsaka2022reconfiguration}.

%% file: fig-yes.tex
\begin{figure}[t]
    \null\hfill
    \subfloat[$3$-coloring $\f_\sss$.]{\includegraphics[scale=0.3]{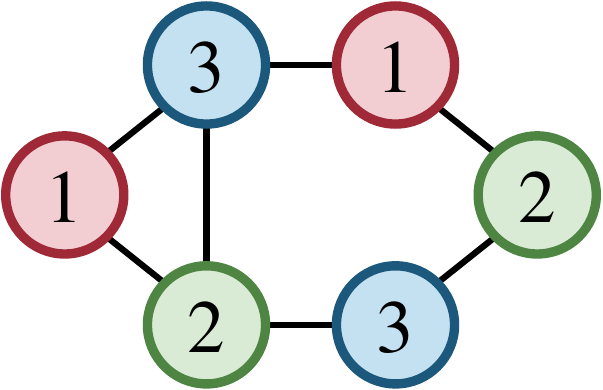}}
    \hfill
    \subfloat[$3$-coloring $\f^{(2)}$.]{\includegraphics[scale=0.3]{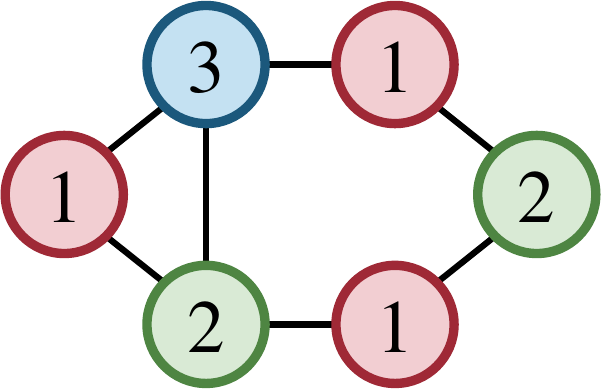}}
    \hfill
    \subfloat[$3$-coloring $\f^{(3)}$.]{\includegraphics[scale=0.3]{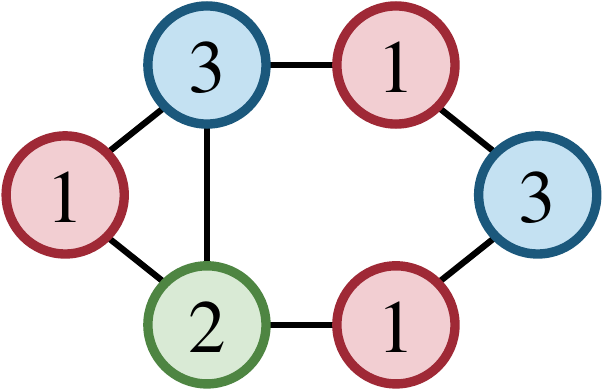}}
    \hfill
    \subfloat[$3$-coloring $\f_\ttt$.]{\includegraphics[scale=0.3]{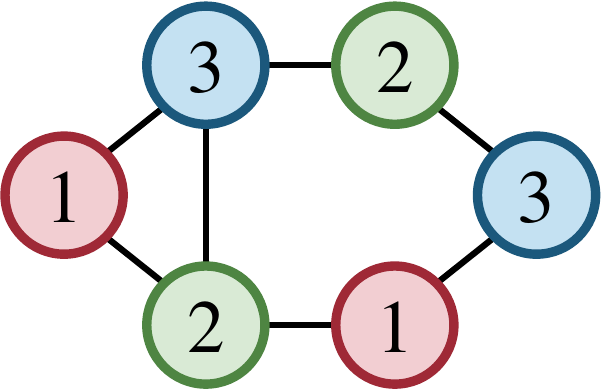}}
    \hfill\null
    \caption{
        A \Yes instance of \prb{3-Coloring Reconfiguration}.
        There is a reconfiguration sequence
        $(\f_\sss = \f^{(1)}, \f^{(2)}, \f^{(3)}, \f^{(4)} = \f_\ttt)$ such that
        each $3$-coloring is proper and
        is obtained by the previous one by recoloring a single vertex.
    }
    \label{fig:yes}
\end{figure}

%% file: fig-no.tex
\begin{figure}[t]
    \null\hfill
    \subfloat[$3$-coloring $\g_\sss$.]{\includegraphics[scale=0.3]{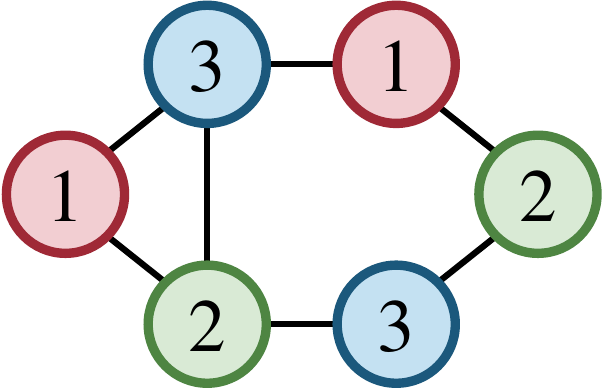}}
    \hfill
    \subfloat[$3$-coloring $\g^{(2)}$.]{\includegraphics[scale=0.3]{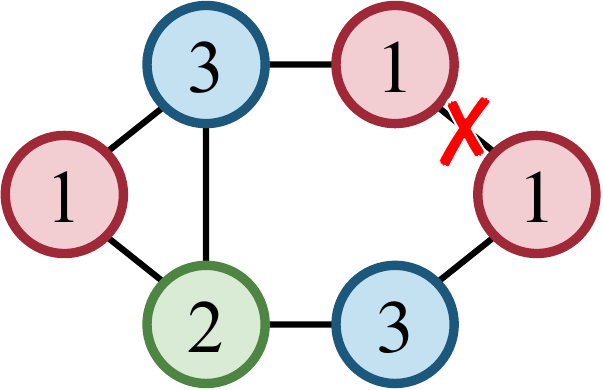}}
    \hfill
    \subfloat[$3$-coloring $\g_\ttt$.]{\includegraphics[scale=0.3]{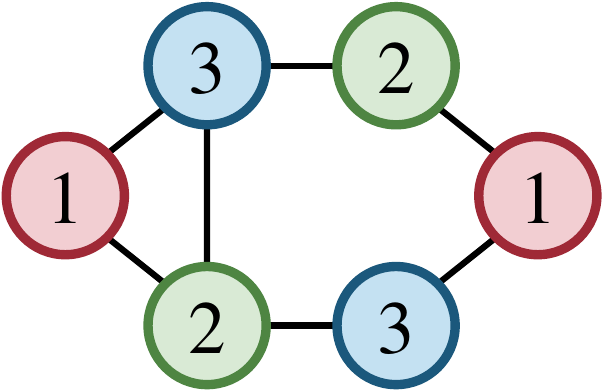}}
    \hfill\null
    \caption{
        A \No instance of \prb{3-Coloring Reconfiguration}.
        There is no reconfiguration sequence from $\g_\sss$ to $\g_\ttt$, because
        $\g_\ttt$ is ``frozen'' in that any vertex cannot recolored.
        Considering this input as an instance of \prb{Maxmin 3-Cut Reconfiguration},
        we can transform $\g_\sss$ into $\g_\ttt$ via $\g^{(2)}$,
        which contains a single monochromatic edge.
    }
    \label{fig:example-no}
\end{figure}

%% file: fig-failed.tex
\begin{figure}[t]
    \centering
    \null\hfill
    \subfloat[
        Graph $G$.
        \label{fig:failed:G}
    ]{\includegraphics[width=0.25\linewidth]{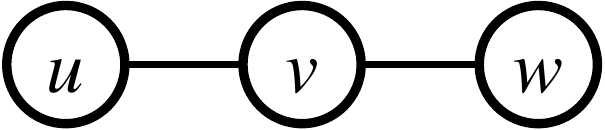}}
    \hfill
    \subfloat[{
        $2$-coloring $\f_\sss$ of $G$.
        \label{fig:failed:G1}
    }]{\includegraphics[width=0.25\linewidth]{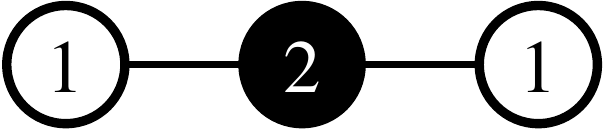}}
    \hfill
    \subfloat[{
        $2$-coloring $f_\ttt$ of $G$.
        \label{fig:failed:G2}
    }]{\includegraphics[width=0.25\linewidth]{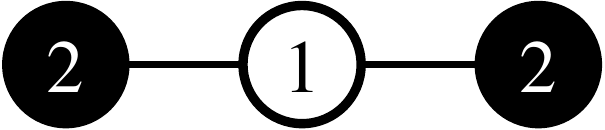}}
    \hfill\null\\
    \null\hfill
    \subfloat[
        Graph $H$.
        \label{fig:failed:H}
    ]{\includegraphics[width=0.3\linewidth]{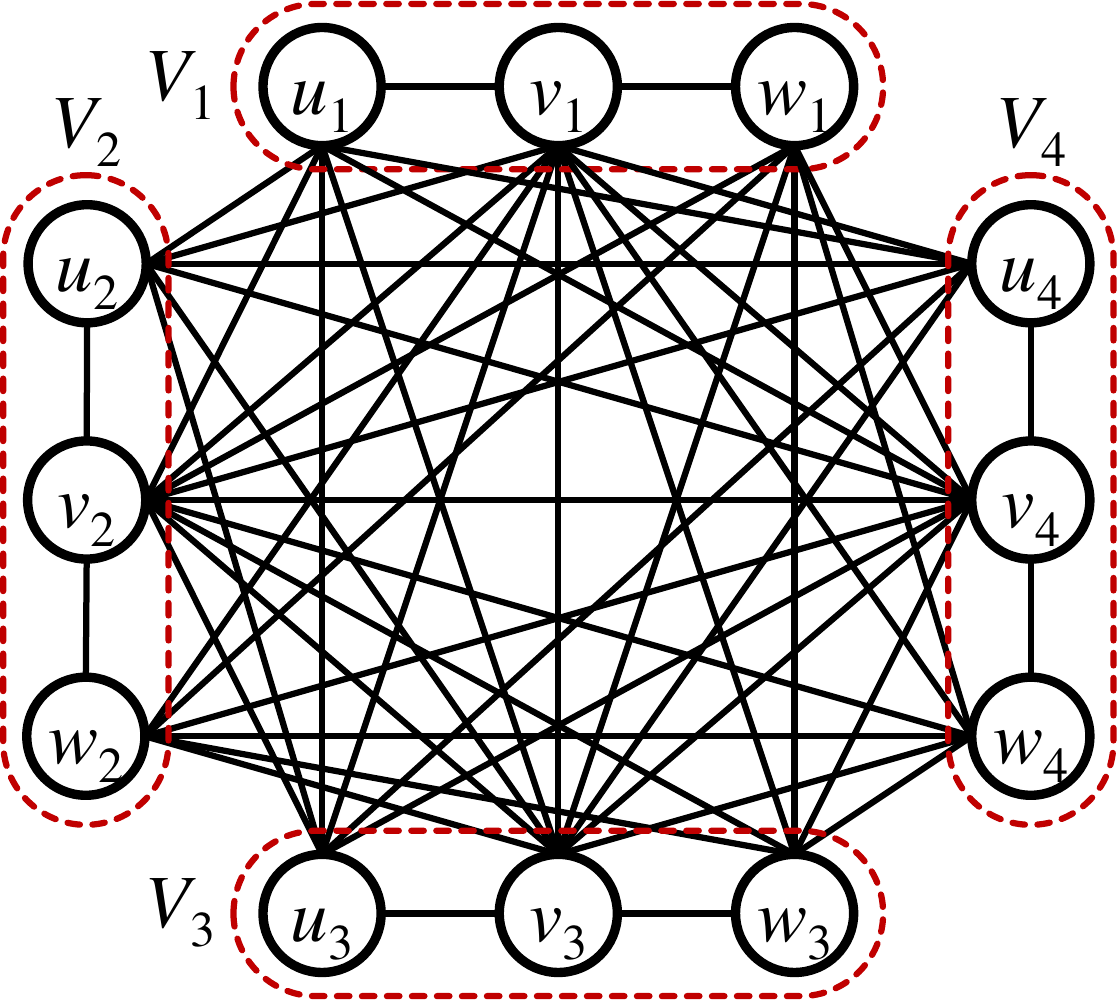}}
    \hfill
    \subfloat[
        $k$-coloring $\f'_\sss$ of $H$.
        \label{fig:failed:H1}
    ]{\includegraphics[width=0.3\linewidth]{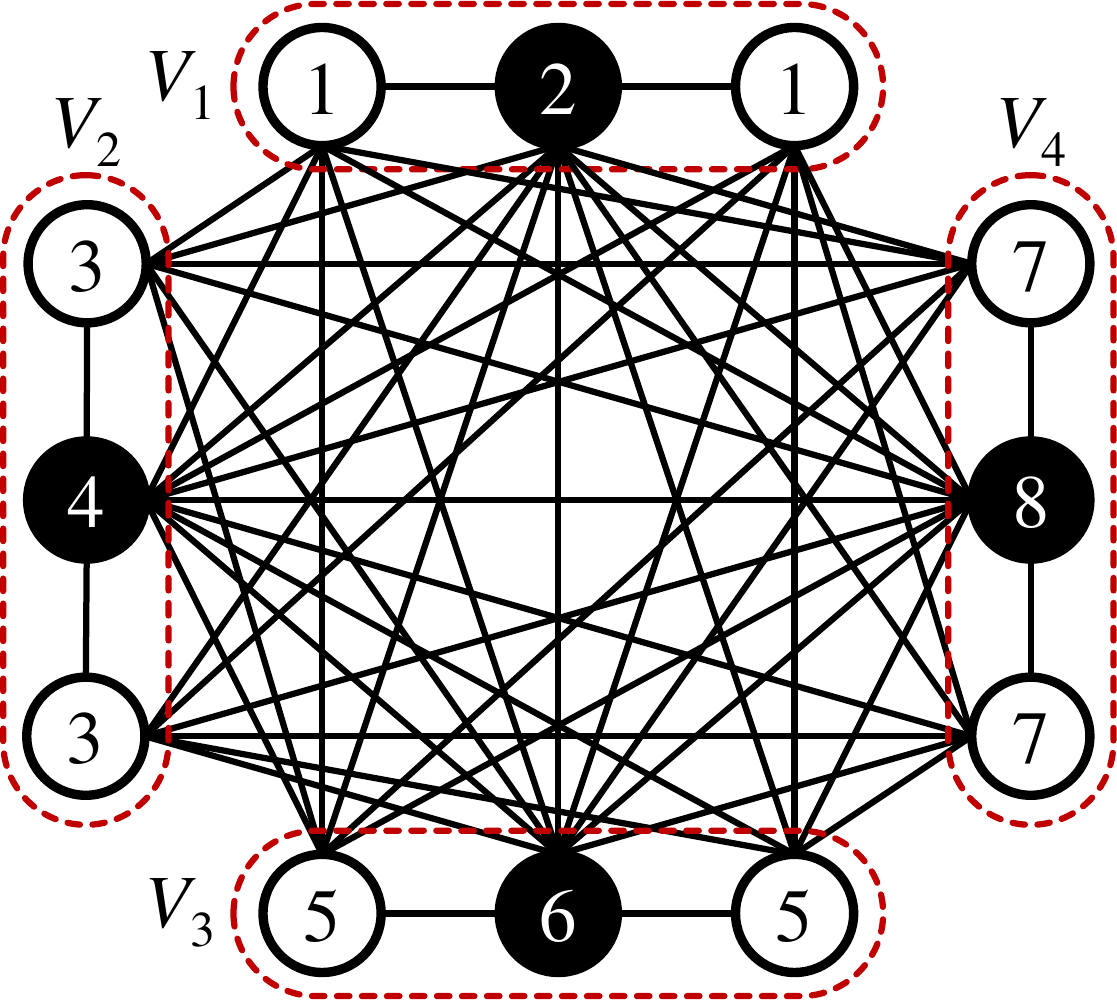}}
    \hfill
    \subfloat[
        $k$-coloring $\f'_\ttt$ of $H$.
        \label{fig:failed:H2}
    ]{\includegraphics[width=0.3\linewidth]{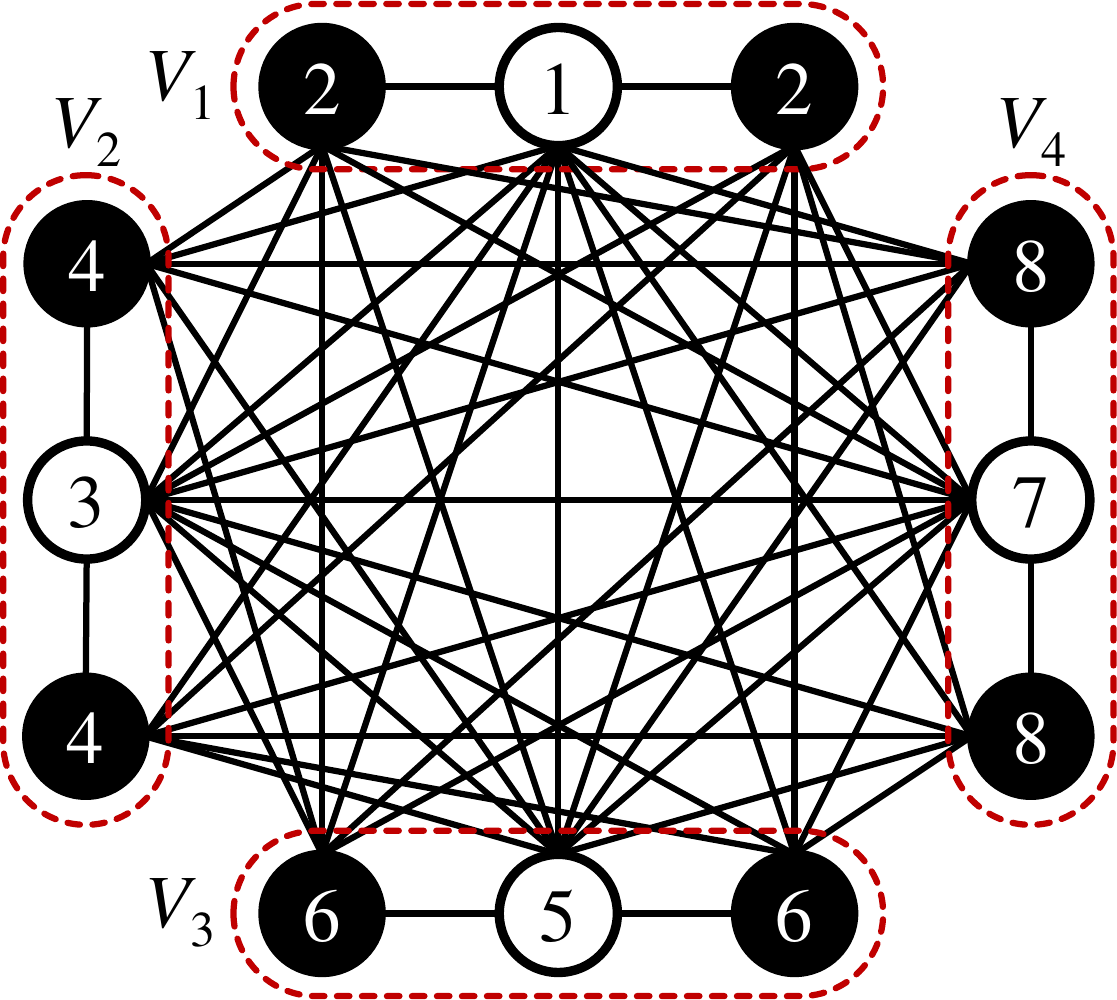}}
    \hfill\null
    \caption{
        A failed attempt to reduce \MMtwoCutReconf to \MMkCutReconf ($k=8$) using \cite{kann1997hardness}.
        Given a graph $G$ and a pair of its $2$-colorings $\f_\sss,\f_\ttt$,
        we construct
        a new graph $H$ and a pair of its $k$-colorings $\f'_\sss,\f'_\ttt$.
        Consider a reconfiguration sequence $\sqcol'$ from $\f'_\sss$ to $\f'_\ttt$
        obtained by recoloring vertices of $V_1, V_2, \ldots, V_{\frac{k}{2}}$ in this order.
        For any intermediate $k$-coloring of $\sqcol'$,
        all but one induced subgraph $H[V_i]$ do not contain any monochromatic edges.
    }
    \label{fig:failed}
\end{figure}

%% file: fig-motive.tex
\begin{figure}[t]
    \centering
    \null\hfill
    \subfloat[
        Coloring of $G$ by $\f$.
        \label{fig:motive:f}
    ]{\includegraphics[width=0.25\linewidth]{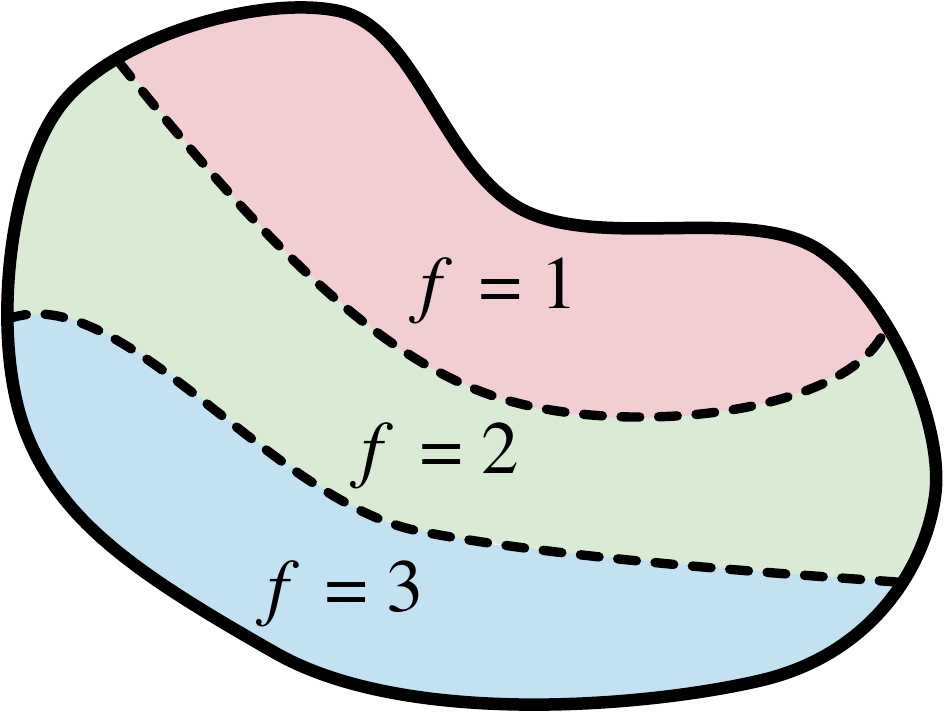}}
    \hfill
    \subfloat[
        Coloring of $G$ by $\g$.
        \label{fig:motive:g}
    ]{\includegraphics[width=0.25\linewidth]{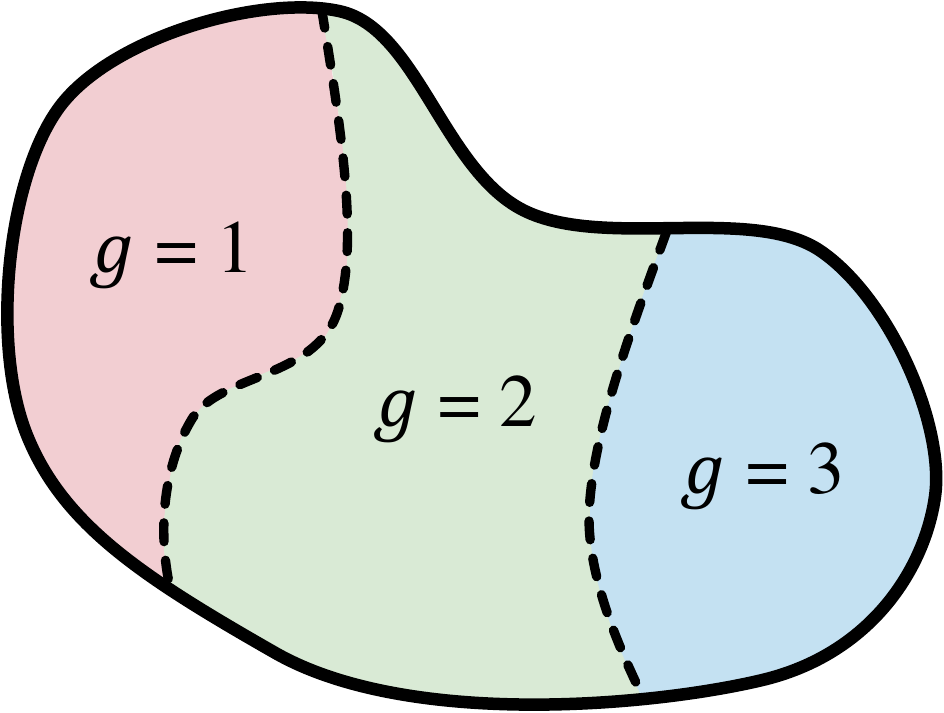}}
    \hfill
    \subfloat[
        $V_{\alpha,\beta}$'s partition the vertex set into $k^2$ groups.
        \label{fig:motive:Vs}
    ]{\includegraphics[width=0.25\linewidth]{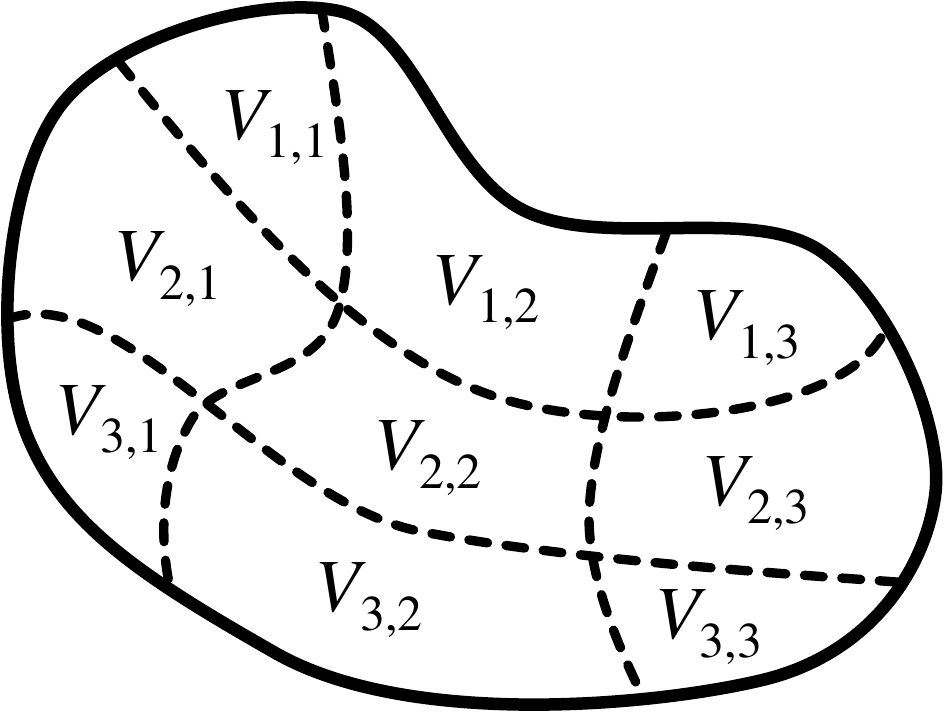}}
    \hfill\null\\
    \null\hfill
    \subfloat[
        Coloring of $V_{\alpha,\beta}$'s by $\f$ looks horizontally striped.
        \label{fig:motive:horizontal}
    ]{\includegraphics[scale=0.3]{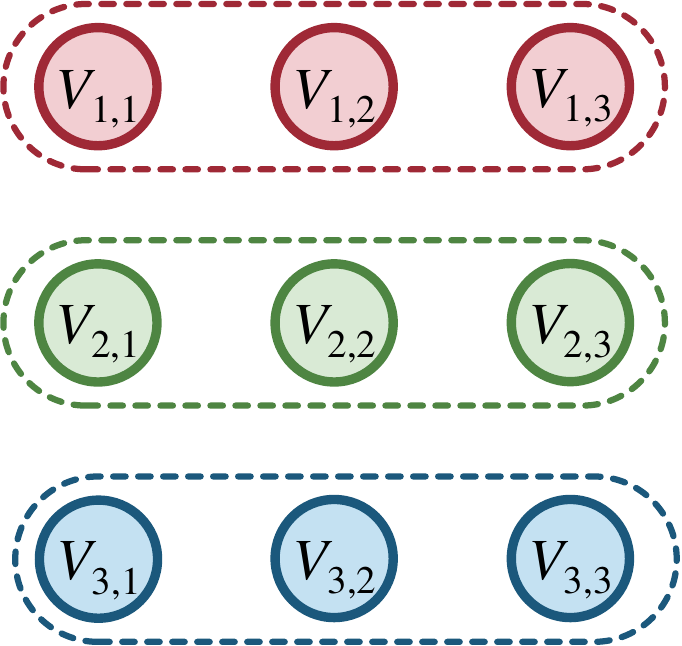}}
    \hfill
    \subfloat[
        Coloring of $V_{\alpha,\beta}$'s by $\g$ looks vertically striped.
        \label{fig:motive:vertical}
    ]{\includegraphics[scale=0.3]{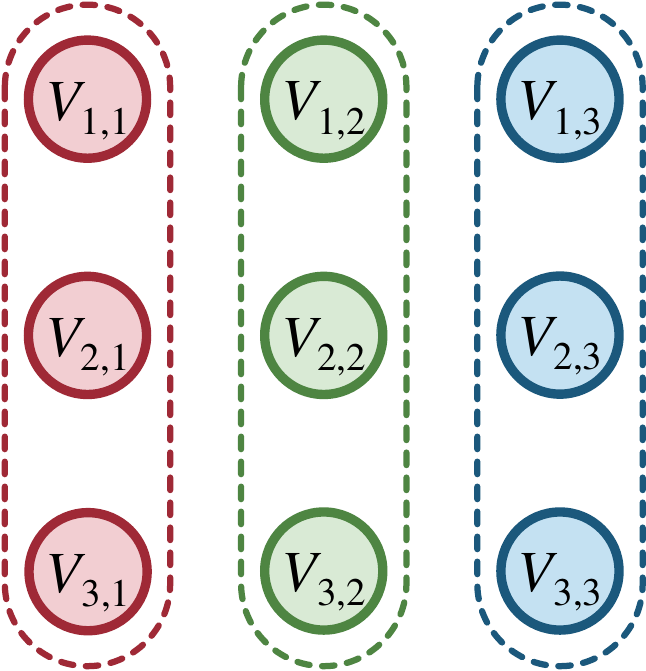}}
    \hfill
    \subfloat[
        Graph structure of $V_{\alpha,\beta}$'s.
        \label{fig:motive:grid}
    ]{\includegraphics[scale=0.3]{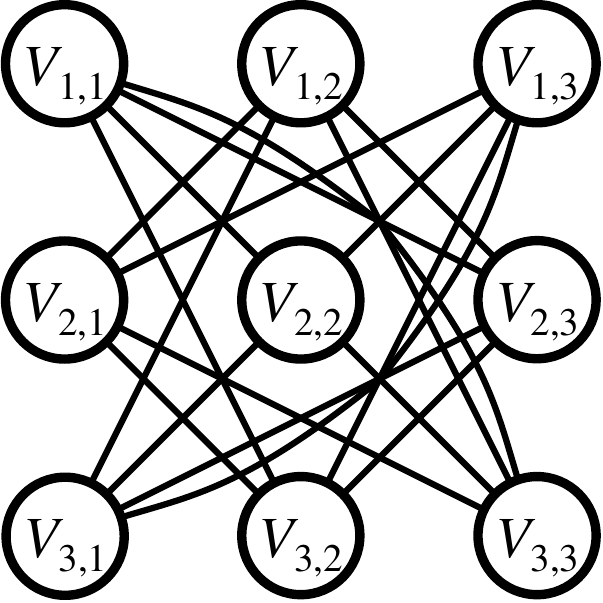}}
    \hfill\null
    \caption{
        Our proposed encoding and the stripe test are motivated by
        the graph structure formed by two different proper $k$-colorings.
    }
    \label{fig:motive}
\end{figure}

%% file: fig-far.tex
\begin{figure}
    \centering
    \includegraphics[width=0.2\linewidth]{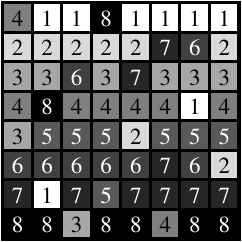}
    \caption{
        A $k$-coloring $\f$ of $[k]^2$ that is far from being striped.
        Obviously, $\f$ is closest to an $8 \times 8$ horizontally striped pattern but 
        differs in $16$ entries; thus, $\f$ is $0.25$-far from being striped.
    }
    \label{fig:far}
\end{figure}

%% file: fig-crazy.tex
\begin{figure}[t]
    \centering
    \null\hfill
    \subfloat[
        Graph $G$.
        \label{fig:crazy:G}
    ]{\includegraphics[width=0.25\linewidth]{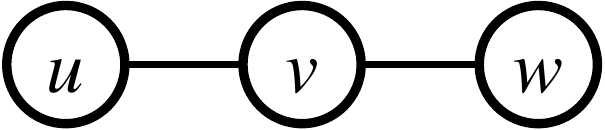}}
    \hfill
    \subfloat[{
        $2$-coloring $\f_\sss$ of $G$.
        \label{fig:crazy:G1}
    }]{\includegraphics[width=0.25\linewidth]{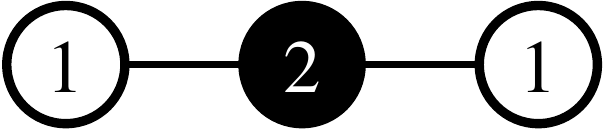}}
    \hfill
    \subfloat[{
        $2$-coloring $f_\ttt$ of $G$.
        \label{fig:crazy:G2}
    }]{\includegraphics[width=0.25\linewidth]{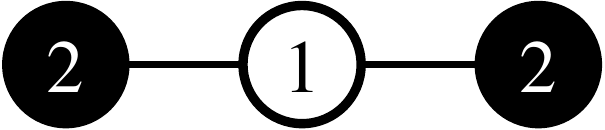}}
    \hfill\null\\
    \null\hfill
    \subfloat[
        Graph $H$ of \cref{lem:intro:Cut-hard:crazy}.
        \label{fig:crazy:H}
    ]{\includegraphics[width=0.31\linewidth]{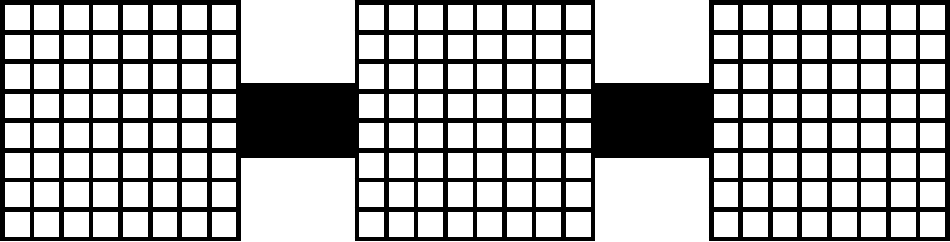}}
    \hfill
    \subfloat[
        $k$-coloring $\f'_\sss$ of $H$.
        \label{fig:crazy:H1}
    ]{\includegraphics[width=0.31\linewidth]{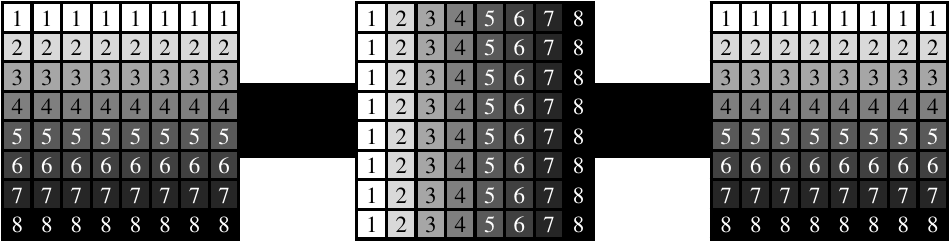}}
    \hfill
    \subfloat[
        $k$-coloring $\f'_\ttt$ of $H$.
        \label{fig:crazy:H2}
    ]{\includegraphics[width=0.31\linewidth]{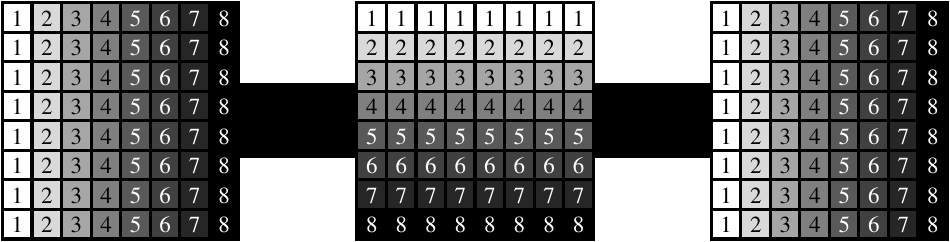}}
    \hfill\null
    \caption{
        Our reduction from \MMtwoCutReconf to \MMkCutReconf
        in the proof of \cref{lem:intro:Cut-hard:crazy} ($k=8$).
        Given a graph $G$ and a pair of its $2$-colorings $\f_\sss,\f_\ttt$,
        we construct
        a new (multi)graph $H$ and a pair of its $k$-colorings $\f'_\sss,\f'_\ttt$, where
        the vertex set of $H$ consists of $|V|$ $k \times k$ grids and
        the edge set of $H$ emulates $\V_G$.
        (The edges are represented by the thick lines in the above figures because
        they are too complicated to be drawn).
        Each $k \times k$ grid is colored in either a horizontally or vertically striped pattern
        depending on $\f_\sss$ or $\f_\ttt$.
        If any reconfiguration sequence from $\f_\sss$ to $\f_\ttt$ makes
        $\epsilon$-fraction of edges $G$ monochromatic,
        any reconfiguration sequence from $\f'_\sss$ to $\f'_\ttt$ makes
        $\Omega\left(\frac{\epsilon}{k}\right)$-fraction of edges of $H$ monochromatic.
    }
    \label{fig:crazy}
\end{figure}

%% file: pre.tex
\section{Preliminaries}
\label{sec:pre}

\subsection{\kColReconf and \MMkCutReconf}
We formulate \kColReconf and its approximate version.
Throughout this paper, all graphs are \emph{undirected}.
For a graph $G=(V,E)$,
let $V(G)$ and $E(G)$ denote the vertex set and edge set of $G$, respectively.
For a vertex $v$ of $G$,
let $\nei_G(v)$ denote the set of the neighbors of $v$ and
$d_G(v)$ denote the degree of $v$.
For a vertex set $S \subseteq V(G)$, we write $G[S]$ for the subgraph of $G$ induced by $S$.
Unless otherwise stated, graphs appearing in this paper are \emph{multigraphs}; namely, 
the edge set is a multiset consisting of \emph{parallel edges}.

For a graph $G=(V,E)$ and a positive integer $k \in \bbN$,
a \emph{$k$-coloring} of $G$ is a function $\f \colon V \to [k]$
that assigns a color of $[k]$ to each vertex of $G$.
We call $\f(v)$ the \emph{color} of $v$.
An edge $(v,w)$ of $G$ is said to be
\emph{bichromatic} on $\f$ if $\f(v) \neq \f(w)$ and
\emph{monochromatic} on $\f$ if $\f(v) = \f(w)$.
We say that a $k$-coloring $\f$ of $G$ is \emph{proper}
if every edge of $G$ is bichromatic on $f$.
A graph $G$ is said to be \emph{$k$-colorable} if there is a proper $k$-coloring of $G$.
The \emph{value} of $\f$
is defined as the fraction of edges of $G$ that are bichromatic on $\f$; namely,
\begin{align}
    \val_G(\f) \defeq \frac{1}{|E|} \cdot \left|\Bigl\{
        (v,w) \in E \Bigm| \f(v) \neq \f(w)
    \Bigr\}\right|.
\end{align}
Recall that
\prb{$k$-Coloring} asks to decide if a graph $G$ is $k$-colorable, and
its approximate version called \prb{Max $k$-Cut}
(a.k.a.~\prb{Max $k$-Colorable Subgraph} \cite{papadimitriou1991optimization,guruswami2013improved}\footnote{
In \cite{guruswami2013improved},
\prb{Max $k$-Colorable Subgraph} always refers to the perfect completeness case;
i.e., $G$ is promised to be $k$-colorable.
})
requires to find a $k$-coloring $\f$ of $G$ that maximizes $\val_G(\f)$.

Subsequently, we formulate a reconfiguration version of
\prb{$k$-Coloring} as well as \prb{Max $k$-Cut}.
For a graph $G=(V,E)$ and a pair of its $k$-colorings
$\f_\sss,\f_\ttt \colon V \to [k]$,
a \emph{reconfiguration sequence from $\f_\sss$ to $\f_\ttt$}
is any sequence
$\sqcol = (\f^{(1)}, \ldots, \f^{(T)})$
over $k$-colorings of $G$
such that
$\f^{(1)} = \f_\sss$,
$\f^{(T)} = \f_\ttt$, and
every pair of adjacent $k$-colorings differ in at most one vertex.
The \kColReconf problem \cite{bonsma2009finding,cereceda2008connectedness,cereceda2011finding,cereceda2009mixing,cereceda2007mixing}
asks to decide
if there is a reconfiguration sequence from $\f_\sss$ to $\f_\ttt$
consisting only of proper $k$-colorings of $G$.
Note that
\kColReconf belongs to $\cP$ if $k \leq 3$ \cite{cereceda2011finding} whereas
it becomes $\PSPACE$-complete for every $k \geq 4$ \cite{bonsma2009finding}.

Since we are concerned with approximability of \kCutReconf, we formulate its approximate version.
For a reconfiguration sequence $\sqcol = (\f^{(1)}, \ldots, \f^{(T)})$ over $k$-colorings of $G$,
let $\val_G(\sqcol)$ denote the \emph{minimum fraction} of bichromatic edges 
over all $\f^{(t)}$'s in $\sqcol$; namely,
\begin{align}
    \val_G(\sqcol) \defeq \min_{\f^{(t)} \in \sqcol} \val_G(\f^{(t)}).
\end{align}
For a graph $G=(V,E)$ and a pair of its $k$-colorings $\f_\sss,\f_\ttt$,
the \MMkCutReconf problem requires to maximize $\val_G(\sqcol)$ subject to
$\sqcol = (\f_\sss, \ldots, \f_\ttt)$.
\MMkCutReconf is $\PSPACE$-hard because so is \kColReconf.
For a pair of $k$-colorings $\f_\sss,\f_\ttt$ of $G$,
let $\opt_G(\f_\sss \reco \f_\ttt)$ denote the maximum value of 
$\val_G(\sqcol)$ over all possible reconfiguration sequences $\sqcol$
from $\f_\sss$ to $\f_\ttt$; namely,
\begin{align}
    \opt_G\bigl(\f_\sss \reco \f_\ttt\bigr)
    \defeq \max_{\sqcol = (\f_\sss, \ldots, \f_\ttt)}
    \val_G(\sqcol).
\end{align}
Note that $\opt_G(\f_\sss \reco \f_\ttt) \leq \min\{\val_G(\f_\sss), \val_G(\f_\ttt)\}$.
The gap version of \MMkCutReconf is defined as follows:

\begin{problem}
For every reals $0 \leq s \leq c \leq 1$ and positive integer $k \in \bbN$,
\prb{Gap$_{c,s}$ \kCutReconf} requires to determine for 
a graph $G$ and a pair of its $k$-colorings $\f_\sss,\f_\ttt$,
whether $\opt_G(\f_\sss \reco \f_\ttt) \geq c$ or
$\opt_G(\f_\sss \reco \f_\ttt) < s$.
Here, $c$ and $s$ are respectively called \emph{completeness} and \emph{soundness}.
\end{problem}\noindent

We say that a reconfiguration sequence
$\sqcol = (\f^{(1)}, \ldots, \f^{(T)})$
from $\f_\sss$ to $\f_\ttt$ is \emph{irredundant} if 
(1) no pair of adjacent $k$-colorings are identical, and
(2) for each vertex $v$ of $G$, there is a unique index $\tau_v \in [T]$ such that
\begin{align}
    \f^{(t)}(v) =
    \begin{cases}
        \f_\sss(v) & \text{if } 1 \leq t \leq \tau_v, \\
        \f_\ttt(v) & \text{if } \tau_v < t \leq T.
    \end{cases}
\end{align}
Informally, irredundancy ensures that each vertex is recolored at most once;
in particular,
the length of $\sqcol$ must be the number of vertices on which $\f_\sss$ and $\f_\ttt$ differ.
Let $\stsqcol(\f_\sss \reco \f_\ttt)$ denote the set of 
    all irredundant reconfiguration sequences from $\f_\sss$ to $\f_\ttt$.
The size of $\stsqcol(\f_\sss \reco \f_\ttt)$ is equal to
$d!$, where $d$ is the number of vertices on which $\f_\sss$ and $\f_\ttt$ differ. 
For any $\ell$ $k$-colorings $\f_1, \f_2, \ldots, \f_\ell$ of a graph $G$,
let $\stsqcol(\f_1 \reco \f_2 \reco \cdots \reco \f_\ell)$ denote
the set of reconfiguration sequences obtained by concatenating
any $\ell-1$ irredundant reconfiguration sequences of
$\stsqcol(\f_i \reco \f_{i+1})$ for all $i \in [\ell-1]$,
which can be defined recursively as follows:
\begin{align}
    \stsqcol\bigl(\f_1 \reco \f_2 \reco \cdots \reco \f_\ell\bigr) \defeq
    \Bigl\{
        \sqcol \circ \sqcol' \Bigm|
        \sqcol \in \stsqcol\bigl(\f_1 \reco \f_2\bigr) \text{ and }
        \sqcol' \in \stsqcol\bigl(\f_2 \reco \cdots \reco \f_\ell\bigr)
    \Bigr\}.
\end{align}

\subsection{Some Concentration Inequalities}
Here, we introduce some concentration inequalities.
The Chernoff bound is first introduced below.

\begin{theorem}[Chernoff bound]
\label{thm:Chernoff}
    Let $X_1, \ldots, X_n$ be independent Bernoulli random variables, and
    $X \defeq \sum_{i \in [n]} X_i$.
    Then, for any real $\epsilon \in (0,1)$, it holds that
    \begin{align}
    \begin{aligned}
        \Pr\Bigl[ X \geq (1+\epsilon)\E[X] \Bigr]
        & \leq \exp\left(-\frac{\epsilon^2 \cdot \E[X]}{3}\right), \\
        \Pr\Bigl[ X \leq (1-\epsilon)\E[X] \Bigr]
        & \leq \exp\left(-\frac{\epsilon^2 \cdot \E[X]}{3}\right).
    \end{aligned}
    \end{align}
\end{theorem}

We then introduce a read-$k$ family of random variables and
a read-$k$ analogue of the Chernoff bound due to \citet{gavinsky2015tail}.

\begin{definition}
    A family $X_1, \ldots, X_n$ of random variables is called
    a \emph{read-$k$ family}
    if there exist
    $m$ independent random variables $Y_1, \ldots, Y_m$,
    $n$ subsets $S_1, \ldots, S_n$ of $[m]$, and
    $n$ Boolean functions $f_1, \ldots, f_n$ such that
    \begin{itemize}
    \item each $X_i$ is represented as $X_i = f_i((Y_j)_{j \in S_i})$, and
    \item each $j$ of $[m]$ appears in at most $k$ of the $S_i$'s.
    \end{itemize}
\end{definition}

\begin{theorem}[Read-$k$ Chernoff bound \cite{gavinsky2015tail}]
\label{thm:read-k-Chernoff}
    Let $X_1, \ldots, X_n$ be a family of read-$k$ Bernoulli random variables, and
    $X \defeq \sum_{i \in [n]} X_i$.
    Then, for any real $\epsilon > 0$, it holds that
    \begin{align}
    \begin{aligned}
        \Pr\Bigl[ X \leq \E[X] - \epsilon n \Bigr]
        & \leq \exp\left(
            -\frac{2\epsilon \cdot n}{k}
        \right), \\
        \Pr\Bigl[ X \geq \E[X] + \epsilon n \Bigr]
        & \leq \exp\left(
            -\frac{2\epsilon \cdot n}{k}
        \right).
    \end{aligned}
    \end{align}
\end{theorem}

%% file: Cut-hard.tex
\section{$\PSPACE$-hardness of $\left(1-\Omega\left(\frac{1}{k}\right)\right)$-factor Approximation for \MMkCutReconf}
\label{sec:Cut-hard}

In this section, we prove that \MMkCutReconf is $\PSPACE$-hard to approximate
within a factor of $1 - \Omega\left(\frac{1}{k}\right)$ for every $k \geq 2$.

\begin{theorem}
\label{thm:Cut-hard}
    There exist universal constants $\delta_c,\delta_s \in (0,1)$ with $\delta_c < \delta_s$ such that
    for all sufficiently large $k \geq k_0 \defeq \kzero$,
    \prb{Gap$_{1-\frac{\delta_c}{k}, 1-\frac{\delta_s}{k}}$ \kCutReconf} is $\PSPACE$-hard.
    Moreover,
    there exists a universal constant $\delta_0 \in (0,1)$ such that
    \MMkCutReconf is
    $\PSPACE$-hard to approximate within a factor of $1 - \frac{\delta_0}{k}$
    for every $k \geq 2$.
    The same hardness result holds even if
    the maximum degree of the input graph is $\bigO(k^2)$.
\end{theorem}

\subsection{Outline of the Proof of \texorpdfstring{\cref{thm:Cut-hard}}{Theorem~\protect\ref{thm:Cut-hard}}}
Here, we present an outline of the proof of \cref{thm:Cut-hard}.
Our starting point is $\PSPACE$-hardness of approximating \MMtwoCutReconf,
whose proof is based on \cite{bonsma2009finding,hirahara2024probabilistically,ohsaka2023gap} and
deferred to \cref{app:Cut-hard:2Cut}.

\begin{proposition}[$*$]
\label{prp:Cut-hard:2Cut}
There exist universal constants $\epsilon_c,\epsilon_s \in (0,1)$ with $\epsilon_c < \epsilon_s$ such that
\prb{Gap$_{1-\epsilon_c,1-\epsilon_s}$ \twoCutReconf}
is $\PSPACE$-hard.
Moreover, the same hardness result holds even if
the maximum degree of input graphs is bounded by some constant $\Delta \in \bbN$.
\end{proposition}\noindent
We then construct the following two gap-preserving reductions from \MMtwoCutReconf to \MMkCutReconf,
the former for all sufficiently large $k$ and
the latter for finitely many $k$.

\begin{lemma}
\label{lem:Cut-hard:crazy}
    For every reals $\epsilon_c,\epsilon_s \in (0,1)$ with $\epsilon_c < \epsilon_s$,
    there exist reals $\delta_c, \delta_s \in (0,1)$ with $\delta_c < \delta_s$
    depending only on the values of $\epsilon_c$ and $\epsilon_s$ such that
    for all sufficiently large $k \geq k_0 \defeq \kzero$ and any integer $\Delta \in \bbN$,
    there exists a gap-preserving reduction from
    \prb{Gap$_{1-\epsilon_c,1-\epsilon_s}$ \twoCutReconf}
    on graphs of maximum degree $\Delta$
    to
    \prb{Gap$_{1-\frac{\delta_c}{k},1-\frac{\delta_s}{k}}$ \kCutReconf}
    on graphs of maximum degree $\bigO(\Delta \cdot k^2)$.
\end{lemma}

\begin{lemma}[$*$]
\label{lem:Cut-hard:quadratic}
For every integer $k \geq 3$,
every reals $\epsilon_c,\epsilon_s \in (0,1)$ with $\epsilon_c < \epsilon_s$, and
every integer $\Delta \in \bbN$,
there exist universal constants $\delta_c,\delta_s \in (0,1)$ with $\delta_c<\delta_s$ such that
there exists a gap-preserving reduction from
\prb{Gap$_{1-\epsilon_c, 1-\epsilon_s}$ \twoCutReconf}
on graphs of maximum degree $\Delta$
to
\prb{Gap$_{1-\delta_c, 1-\delta_s}$ \kCutReconf}
on graphs of maximum degree $\bigO(\Delta + \poly(k))$.
\end{lemma}

\begin{remark}
    The values of $\delta_c,\delta_s$ in \cref{lem:Cut-hard:quadratic}
    depend on $\epsilon_c,\epsilon_s,k$ and quadratically decrease in $k$; i.e.,
    $\delta_c,\delta_s = \Theta(k^{-2})$.
    We thus cannot use \cref{lem:Cut-hard:quadratic} to prove \cref{thm:Cut-hard}
    for large $k$.
\end{remark}\noindent
The proof of \cref{lem:Cut-hard:quadratic} is deferred to \cref{app:Cut-hard:quadratic}.
As a corollary of \cref{prp:Cut-hard:2Cut,lem:Cut-hard:crazy,lem:Cut-hard:quadratic},
we obtain \cref{thm:Cut-hard}.

\begin{proof}[Proof of \cref{thm:Cut-hard}]
By \cref{prp:Cut-hard:2Cut},
\prb{Gap$_{1-\epsilon_c,1-\epsilon_s}$ \twoCutReconf}
on graphs of maximum degree $\Delta$
is $\PSPACE$-hard
for some constants $\epsilon_c,\epsilon_s \in (0,1)$ with $\epsilon_c < \epsilon_s$ and
$\Delta \in \bbN$.
By \cref{lem:Cut-hard:crazy},
there exist universal constants $\delta_c,\delta_s \in (0,1)$ with $\delta_c < \delta_s$
such that
\prb{Gap$_{1-\frac{\delta_c}{k}, 1-\frac{\delta_s}{k}}$ \kCutReconf}
is $\PSPACE$-hard for every $k \geq k_0$.
The ratio between completeness and soundness is evaluated as follows:
\begin{align}
    \frac{1-\frac{\delta_s}{k}}{1-\frac{\delta_c}{k}}
    = \frac{1-\frac{\delta_c}{k} + \frac{\delta_c-\delta_s}{k}}{1-\frac{\delta_c}{k}}
    = 1 - \frac{\delta_s - \delta_c}{1-\frac{\delta_c}{3}} \cdot \frac{1}{k}
    \leq 1 - \frac{\delta_s-\delta_c}{k}.
\end{align}
Therefore, \MMkCutReconf is $\PSPACE$-hard to approximate
within a factor of $1 - \frac{\delta_s-\delta_c}{k}$ for every $k \geq k_0$.
By applying \cref{lem:Cut-hard:quadratic} to \cref{prp:Cut-hard:2Cut} for each $k < k_0$,
we obtain a universal constant $\delta' \in (0,1)$ such that
\MMkCutReconf is $\PSPACE$-hard to approximate within a factor of 
$1-\delta'$ for every $k < k_0$.
Both results imply the existence of a universal constant $\delta_0 \in (0,1)$ such that
\MMkCutReconf is $\PSPACE$-hard to approximate within a factor of 
$1-\frac{\delta_0}{k}$ for every $k \geq 2$,
accomplishing the proof.
\end{proof}

The remainder of this section is devoted to the proof of \cref{lem:Cut-hard:crazy}.

\subsection{Three Tests}
\label{sec:Cut-hard:tests}

In this subsection, we introduce the key ingredients in the proof of \cref{lem:Cut-hard:crazy}.
Consider a probabilistic verifier $\V$,
given oracle access to a $k$-coloring $\f \colon V \to [k]$,
that is allowed to
sample a pair $(v,w)$ of distinct vertices from $V$ (nonadaptively) and
accepts (resp.~rejects) if $\f(v) \neq \f(w)$ (resp.~$\f(v) = \f(w)$).
Observe easily that $\V$ can be emulated by a multigraph $G$ on vertex set $V$ in a sense that
the acceptance (resp.~rejection) probability of $\V$ is equal to
the fraction of the bichromatic (resp.~monochromatic) edges in $G$.
Our reduction in \cref{sec:Cut-hard:crazy} from \MMtwoCutReconf to \MMkCutReconf 
will be described in the language of such verifiers.

Suppose we are given an instance $(G,\f_\sss,\f_\ttt)$ of \MMtwoCutReconf.
We shall encode a $2$-coloring of each vertex $v$ of $G$
by using a $k$-coloring of $[k]^2$, denoted by $\f'(v) \colon [k]^2 \to [k]$,
whose motivation was described in \cref{sec:overview-Cut-hard:our}.
Specifically,
$\f'(v)$ is supposed to be ``horizontally striped'' if $v$'s color is $1$, and
$\f'(v)$ is supposed to be ``vertically striped'' if $v$'s color is $2$.
We would like to check if
these $k$-colorings $(\f'(v))_{v \in V}$ are an encoding of a \emph{proper} $2$-coloring of $G$.
For this purpose,
we will implement the following three auxiliary verifiers:
\begin{itemize}
    \item \textbf{Stripe verifier} $\Vstripe$, which checks
        if a $k$-coloring $\f$ of $[k]^2$ is close to a ``striped'' pattern.
    \item \textbf{Consistency verifier} $\Vcons$, which checks
        if a pair of $k$-coloring $\f, \g$ of $[k]^2$
        share the \emph{same} striped pattern
        (given that both $\f$ and $\g$ are close to striped patterns).
    \item \textbf{Edge verifier} $\Vedge$, which checks
        if a pair of $k$-coloring $\f, \g$ of $[k]^2$
        are \emph{closed to} the same striped pattern,
        by calling $\Vstripe$ and $\Vcons$ with a carefully designed probability.
\end{itemize}

We will say that a $k$-coloring $\f \colon [k]^2 \to [k]$ is
\emph{horizontally striped} if 
$\f(x,y) = \sigma(y)$ for all $(x,y) \in [k]^2$
for some permutation $\sigma \in \frakS_k$,
\emph{vertically striped} if
$\f(x,y) = \sigma(x)$ for all $(x,y) \in [k]^2$
for some permutation $\sigma \in \frakS_k$, and
\emph{striped} if
it is horizontally or vertically striped.
Throughout this subsection, we fix $k \geq k_0 \defeq \kzero$.

\subsubsection{Stripe Test}
\label{sec:Cut-hard:tests:stripe}

We first introduce the \emph{stripe verifier} $\Vstripe$,
which tests if a $k$-coloring $\f$ of $[k]^2$ is close to being striped.

\begin{itembox}[l]{\textbf{Stripe verifier $\Vstripe$.}}
\begin{algorithmic}[1]
    \item[\textbf{Oracle access:}]
        a $k$-coloring $\f \colon [k]^2 \to [k]$.
    \State select $(x_1,y_1) \in [k]^2$ and $(x_2,y_2) \in [k]^2$
        s.t.~$x_1 \neq x_2$ and $y_1 \neq y_2$ uniformly at random.
    \If{$\f(x_1,y_1) = \f(x_2,y_2)$}
        \State declare \Reject.
    \Else
        \State declare \Accept.
    \EndIf
\end{algorithmic}
\end{itembox}

Observe easily that $\Vstripe$ always accepts $f$ if and only if $f$ is striped.

\begin{lemma}
\label{lem:Cut-hard:stripe:striped}
    Let  $\f \colon [k]^2 \to [k]$ be any $k$-coloring.
    Then, $\Vstripe$ accepts $\f$ with probability $1$
    if and only if $\f$ is striped.
\end{lemma}
\begin{proof} 
Since the ``if'' direction is obvious, we show the ``only-if'' direction.
Consider a graph $G$ that emulates $\Vstripe$.
Let $\f \colon [k]^2 \to [k]$ be any $k$-coloring
accepted by $\Vstripe$ with probability $1$.
Denoting by $(I_1, \ldots, I_k)$ a partition of $[k]^2$ such that
$I_\alpha \defeq \{(x,y) \in [k]^2 \mid \f(x,y) = \alpha\}$,
we find each $I_\alpha$ an independent set of $G$.
Since $(x,x) \neq (y,y)$ do not belong to the same independent set,
we can assume $(\sigma(\alpha),\sigma(\alpha)) \in I_\alpha$ for some permutation $\sigma \colon [k] \to [k]$.
Observe that any maximal independent set is of the form either
$\{ (x,y) \in [k]^2 \mid y = \alpha, x \in [k] \}$ or 
$\{ (x,y) \in [k]^2 \mid x = \alpha, y \in [k] \}$
for some $\alpha \in [k]$,
implying that $\f$ must be either horizontally striped or vertically striped.
\end{proof}

Let $\hor$ denote the set of all horizontally striped $k$-colorings,
$\ver$ denote the set of all vertically striped $k$-colorings, and
$\str \defeq \hor \cup \ver$.
We say that a $k$-coloring $\f \colon [k]^2 \to [k]$ is
\emph{$\epsilon$-far from being striped} if $\rHam(\f, \str) > \epsilon$ and
\emph{$\epsilon$-close to being striped} if $\rHam(\f, \str) \leq \epsilon$.
We now demonstrate that if a $k$-coloring $\f \colon [k]^2 \to [k]$ is $\epsilon$-far from being striped,
$\Vstripe$ rejects $\f$ with probability $\Omega\left(\frac{\epsilon}{k}\right)$,
whose proof is rather complicated and deferred to \cref{sec:Cut-hard:stripe:far}.

\begin{lemma} 
\label{lem:Cut-hard:stripe:far}
    There exists a universal constant $\rho \defeq \rhozero$ such that
    for any $k$-coloring $\f \colon [k]^2 \to [k]$ that is $\epsilon$-far from being striped,
    $\Vstripe$ rejects $\f$ with probability more than
    \begin{align}
        \frac{\rho \cdot \epsilon}{k}.
    \end{align}
\end{lemma}

\subsubsection{Consistency Test}
\label{sec:Cut-hard:tests:cons}

We next proceed to the \emph{consistency verifier} $\Vcons$, which tests
if a pair of $k$-colorings $\f,\g$ of $[k]^2$ share the same striped pattern.
Specifically,
$\Vcons$ runs the following two tests with equal probability:
\begin{itemize}
    \item the \emph{row test},
    which accepts $\f \circ \g$ if
    they have the same horizontally striped pattern;
    \item the \emph{column test},
    which accepts $\f \circ \g$ if
    they have the same vertically striped pattern.
\end{itemize}

\begin{itembox}[l]{\textbf{Consistency verifier $\Vcons$.}}
\begin{algorithmic}[1]
    \item[\textbf{Oracle access:}]
        two $k$-colorings $\f, \g \colon [k]^2 \to [k]$.
    \State sample $r \sim [0,1]$.
    \If{$0 \leq r < \frac{1}{2}$} \Comment{perform the row test w.p.~$\frac{1}{2}$.}
        \State select $(x_1,y_1) \in [k]^2$ and $(x_2,y_2) \in [k]^2$
            s.t.~$y_1 \neq y_2$ uniformly at random.
    \Else \Comment{perform the column test w.p.~$\frac{1}{2}$.}
        \State select $(x_1,y_1) \in [k]^2$ and $(x_2,y_2) \in [k]^2$
            s.t.~$x_1 \neq x_2$ uniformly at random.
    \EndIf
    \If{$\f(x_1,y_1) = \g(x_2,y_2)$}
        \State declare \Reject.
    \Else
        \State declare \Accept.
    \EndIf
\end{algorithmic}
\end{itembox}

Let $\dec(\f)$ indicate whether $\f$ is closest to being
horizontally striped (denoted by $1$) or vertically striped (denoted by $2$); namely,
\begin{align}
    \dec(\f) \defeq
    \begin{cases}
        1 & \text{if } \rHam(\f, \hor) \leq \rHam(\f, \ver), \\
        2 & \text{if } \rHam(\f, \hor) > \rHam(\f, \ver). \\
    \end{cases}
\end{align}
A pair of $k$-colorings $\f,\g \colon [k]^2 \to [k]$ are said to be
\emph{consistent} if $\dec(\f) = \dec(\g)$
(i.e., both $\f$ and $\g$ are closest to
being horizontally striped or vertically striped), and
\emph{inconsistent} if $\dec(\f) \neq \dec(\g)$.

When $\f$ and $\g$ are striped,
$\Vcons$'s rejection probability can be calculated exactly as follows.

\begin{lemma}
\label{lem:Cut-hard:cons:striped}
    For any striped two $k$-colorings $\f,\g \colon [k]^2 \to [k]$,
    the following hold\textup{:}
\begin{itemize}
    \item if $\f = \g$ (in particular, $\f$ and $\g$ are consistent),
    $\Vcons$ rejects $\f \circ \g$
    with probability exactly $\frac{1}{2k}$\textup{;}
\item if $\f$ and $\g$ are inconsistent,
    $\Vcons$ rejects $\f \circ \g$
    with probability exactly $\frac{1}{k}$.
\end{itemize}
\end{lemma}
\begin{proof} 
To prove the first claim, assume $\f=\g$ is horizontally striped.
The other case can be shown similarly.
In the row test, it always holds that $y_1 \neq y_2$; i.e.,
we have $\f(x_1,y_1) = \g(x_2,y_2)$ with probability $0$.
In the column test, $y_1$ and $y_2$ are chosen \emph{independently} and uniformly at random;
thus, we have $\f(x_1,y_1) = \g(x_2,y_2)$ with probability exactly $\frac{1}{k}$.
Therefore, the consistency verifier rejects $\f \circ \g$ with probability
$\frac{1}{2}\left(0 + \frac{1}{k}\right) = \frac{1}{2k}$, as desired.

To prove the second claim,
assume $\f$ is horizontally striped and $\g$ is vertically striped.
The opposite case can be shown in the same way.
In the row test,
we have $\f(x_1,y_1) = \g(x_2,y_2)$ with probability $\frac{1}{k}$
because $x_2$ (and thus $\g(x_2,y_2)$) is uniformly distributed over $[k]$.
In the column test,
we have $\f(x_1,y_1) = \g(x_2,y_2)$ with probability $\frac{1}{k}$
because $y_1$ (and thus $\f(x_1,y_1)$) is uniformly distributed.
Therefore, $\Vcons$ rejects $\f \circ \g$ with probability
$\frac{1}{2}\left(\frac{1}{k}+\frac{1}{k}\right) = \frac{1}{k}$, as desired.
\end{proof}

Even when $\f$ and $\g$ are not striped,
$\Vcons$'s rejection probability can be bounded from below as follows.

\begin{lemma}
\label{lem:Cut-hard:cons:far}
    For any two $k$-colorings $\f,\g \colon [k]^2 \to [k]$ such that
    $\f$ is $\epsilon_{\f}$-close to being striped and 
    $\g$ is $\epsilon_{\g}$-close to being striped,
    the following hold\textup{:}
    \begin{itemize}
    \item if $\f$ and $\g$ are inconsistent,
        $\Vcons$ rejects $\f \circ \g$ with probability more than
        \begin{align}
            \Bigl( 1-2 \epsilon_{\f} - 2\epsilon_{\g} \Bigr) \cdot \frac{1}{k},
        \end{align}
    \item if $\f$ and $\g$ are consistent,
        $\Vcons$ rejects $\f \circ \g$ with probability more than
        \begin{align}
            \Bigl( 1-2 \epsilon_{\f} - 2\epsilon_{\g} \Bigr) \cdot \frac{1}{2k}.
        \end{align}
    \end{itemize}
\end{lemma}
\begin{proof} 
Let $\f^*, \g^* \colon [k]^2 \to [k]$ be two striped $k$-colorings
closest to $\f$ and $\g$ such that
$\dec(\f^*) = \dec(\f)$ and $\dec(\g^*) = \dec(\g)$,\footnote{
We need this condition for tie-breaking.
} respectively.
For each $i \in [k]$, let
    $\Delta_{\f}(*, i)$ and $\Delta_{\g}(*, i)$
    denote the number of $x$'s such that
    $\f(x,i) \neq \f^*(x,i)$ and $\g(x,i) \neq \g^*(x,i)$, respectively, and
    let
    $\Delta_{\f}(i,*)$ and $\Delta_{\g}(i,*)$
    denote the number of $y$'s such that
    $\f(i,y) \neq \f^*(i,y)$ and $\g(i,y) \neq \g^*(i,y)$, respectively.
By assumption, we have
\begin{align}
    & \sum_{x \in [k]} \Delta_{\f}(x,*) \leq \epsilon_{\f} \cdot k^2
    \;\text{ and }\;
    \sum_{x \in [k]} \Delta_{\g}(x,*) \leq \epsilon_{\g} \cdot k^2, \\
    & \sum_{y \in [k]} \Delta_{\f}(*, y) \leq \epsilon_{\f} \cdot k^2
    \;\text{ and }\;
    \sum_{y \in [k]} \Delta_{\g}(*, y) \leq \epsilon_{\g} \cdot k^2.
\end{align}

Suppose first $\dec(\f) \neq \dec(\g)$;
    we assume that
    $\dec(\f)=1$ and
    $\dec(\g)=2$ without loss of generality.
By symmetry of $\Vcons$,
    the rows and columns can be rearranged so that
    $\f^*(x,y) = y$ and 
    $\g^*(x,y) = x$ for all $(x,y) \in [k]^2$.
We first bound the rejection probability of the row test.
Let $Q$ denote the set of all quadruples examined by the row test; namely,
\begin{align}
    Q \defeq \Bigl\{
        (x_1,y_1,x_2,y_2) \in [k]^4 \Bigm| y_1 \neq y_2
    \Bigr\}.
\end{align}
Note that $|Q| = k^3(k-1)$.
Conditioned on the event that $y_1=x_2=i$ for some $i \in [k]$,
\begin{itemize}
    \item there are $(k-\Delta_{\f}(*,i))$ $x_1$'s such that $\f(x_1,i) = \f^*(x_1,i)=i$;
    \item there are $(k-1-\Delta_{\g}(i,*))$ $y_2$'s such that $\g(i,y_2) = \g^*(i,y_2)=i$ and $y_2 \neq y_1 = i$;
\end{itemize}
namely, there are (at least) $(k-\Delta_{\f}(*,i)) \cdot (k-1-\Delta_{\g}(i,*))$ pairs $(x_1,y_2)$ such that
$\f(x_1,y_1) = \g(x_2,y_2)$.
Taking the sum over all $i \in [k]$, we deduce that
    the number of quadruples $(x_1,y_1,x_2,y_2)$ in $Q$ such that
    $\f(x_1,y_1) = \g(x_2,y_2)$ is at least
\begin{align}
\begin{aligned}
    & \sum_{i \in [k]} (k-\Delta_{\f}(*,i)) \cdot (k-1-\Delta_{\g}(i,*)) \\
    & = \sum_{i \in [k]} k(k-1)
        - (k-1) \cdot \underbrace{\sum_{i \in [k]}\Delta_{\f}(*,i)}_{\leq \epsilon_{\f} \cdot k^2}
        - k \cdot \underbrace{\sum_{i \in [k]} \Delta_{\g}(i,*)}_{\leq \epsilon_{\g} \cdot k^2}
        + \underbrace{\sum_{i \in [k]} \Delta_{\f}(*,i) \cdot \Delta_{\g}(i,*)}_{\geq 0} \\
    & \geq k^2(k-1) - k^2(k-1)\cdot\epsilon_{\f} - k^3\cdot\epsilon_{\g}.
\end{aligned}
\end{align}
Since the row test draws a quadruple from $Q$ uniformly at random,
    its rejection probability is at least
\begin{align}
\begin{aligned}
    \frac{1}{|Q|} \cdot \Bigl(
        k^2(k-1) - k^2(k-1)\cdot \epsilon_{\f} - k^3 \cdot \epsilon_{\g}
    \Bigr)
    & = \frac{1}{k} - \frac{\epsilon_{\f}}{k} - \frac{\epsilon_{\g}}{k-1} \\
    & = \frac{1}{k} \left(1-\epsilon_{\f}-\epsilon_{\g}-\frac{\epsilon_{\g}}{k-1}\right) \\
    & \geq \frac{1}{k} \Bigl(1-\epsilon_{\f} -2 \epsilon_{\g}\Bigr).
\end{aligned}
\end{align}
Similarly, the column test rejects with probability at least
\begin{align}
    \frac{1}{k} \Bigl( 1-2\epsilon_{\f}-\epsilon_{\g} \Bigr).
\end{align}
Consequently, the rejection probability of $\Vcons$ is at least
\begin{align}
    \frac{1}{2} \left(\frac{1}{k} \Bigl( 1-\epsilon_{\f} -2 \epsilon_{\g}\Bigr)
        + \frac{1}{k} \Bigl(1-2\epsilon_{\f}-\epsilon_{\g} \Bigr)\right) 
    \geq \frac{1}{k} \Bigl( 1-2\epsilon_{\f}-2\epsilon_{\g} \Bigr),
\end{align}
as desired.

Suppose next $\dec(\f) = \dec(\g)$;
we assume $\dec(\f)=\dec(\g)=1$ without loss of generality.
We bound the rejection probability of the column test.
Conditioned on the event that
$\f^*(\cdot,y_1) = \g^*(\cdot,y_2) = i$ for some $i \in [k]$,
there are (at least) $(k-\Delta_{\f}(*,y_1))\cdot(k-1-\Delta_{\g}(*,y_2))$ pairs $(x_1,x_2)$ such that
$\f(x_1,y_1) = \f^*(x_1,y_1) = i = \g^*(x_2,y_2) = \g(x_2,y_2)$.
Taking the sum over all $i \in [k]$, we deduce that the number of quadruples $(x_1,y_1,x_2,y_2) \in Q$ such
that $\f(x_1,y_1) = \g(x_2,y_2)$ is at least
\begin{align}
\begin{aligned}
    \sum_{i \in [k]} (k-\Delta_{\f}(*,y_1))\cdot(k-1-\Delta_{\g}(*,y_2)) 
    \geq k^2(k-1) - k^2(k-1) \cdot \epsilon_{\f} - k^3 \cdot \epsilon_{\g}.
\end{aligned}
\end{align}
The rejection probability of the column test is at least 
\begin{align}
    \frac{1}{|Q|} \cdot \Bigl(
        k^2(k-1) - k^2(k-1) \cdot \epsilon_{\f} - k^3 \cdot \epsilon_{\g}
    \Bigr) \geq \frac{1}{k} \Bigl(1-2\epsilon_{\f} -2\epsilon_{\g} \Bigr).
\end{align}
Consequently, the rejection probability of $\Vcons$ is at least
\begin{align}
    \frac{1}{2}\left(
        \frac{1}{k} \Bigl(1-2\epsilon_{\f} -2\epsilon_{\g}\Bigr) + 0
    \right)
    = \frac{1}{2k} \Bigl(1-2\epsilon_{\f} -2\epsilon_{\g}\Bigr),
\end{align}
which completes the proof.
\end{proof}

\subsubsection{Edge Test}
\label{sec:Cut-hard:tests:edge}

We finally design the \emph{edge verifier} $\Vedge$,
which tests if a pair of $k$-colorings $\f,\g$ of $[k]^2$ are 
close to the same stripe pattern.
For this purpose, $\Vedge$ executes
$\Vstripe$ on $\f$ with probability $\frac{2}{\rho Z}$,
$\Vstripe$ on $\g$ with probability $\frac{2}{\rho Z}$, and
$\Vcons$ on $\f \circ \g$ with probability $\frac{1}{Z}$, where
$Z \defeq \frac{2}{\rho} + \frac{2}{\rho} + 1$ and
$\rho \defeq \rhozero$ is the rejection rate of $\Vstripe$.

\begin{itembox}[l]{\textbf{Edge verifier $\Vedge$.}}
\begin{algorithmic}[1]
    \item[\textbf{Oracle access:}]
        two $k$-colorings $\f,\g \colon [k]^2 \to [k]$.
    \State let $Z \defeq \frac{2}{\rho} + \frac{2}{\rho} + 1$.
    \State sample $r \sim [0,1]$.
    \If{$0 \leq r < \frac{2}{\rho Z}$} \Comment{with probability $\frac{2}{\rho Z}$}
        \State execute $\Vstripe$ on $\f$.
    \ElsIf{$\frac{2}{\rho Z} \leq r < \frac{2}{\rho Z} + \frac{2}{\rho Z}$} \Comment{with probability $\frac{2}{\rho Z}$}
        \State execute $\Vstripe$ on $\g$.
    \Else \Comment{with probability $\frac{1}{Z}$}
        \State execute $\Vcons$ on $\f \circ \g$.
    \EndIf
\end{algorithmic}
\end{itembox}

When $\f$ and $\g$ are striped,
$\Vedge$'s rejection probability is obtained immediately 
from \cref{lem:Cut-hard:stripe:striped,lem:Cut-hard:cons:striped} as follows.
\begin{lemma}
\label{lem:Cut-hard:edge:striped}
For any two striped $k$-colorings $\f,\g \colon [k]^2 \to [k]$,
the following hold\textup{:}
\begin{itemize}
    \item if $\f = \g$ (in particular, $\f$ and $\g$ are consistent),
        $\Vedge$ rejects $\f \circ \g$ with probability exactly $\frac{1}{2Z \cdot k}$\textup{;}
    \item if $\f$ and $\g$ are inconsistent,
        $\Vedge$ rejects $\f \circ \g$ with probability exactly $\frac{1}{Z \cdot k}$.
\end{itemize}
\end{lemma}

Whenever $\f$ and $\g$ are inconsistent,
$\Vedge$'s rejection probability is at least $\frac{1}{Z \cdot k}$
(regardless of the distance from being striped).

\begin{lemma}
\label{lem:Cut-hard:edge:mismatch}
    Let $\f,\g \colon [k]^2 \to [k]$ be any two inconsistent $k$-colorings.
    Then, $\Vedge$ rejects $\f \circ \g$ with probability at least
    $\frac{1}{Z \cdot k}$.
\end{lemma}
\begin{proof} 
Define
    $\epsilon_{\f} \defeq \rHam(\f, \str)$ and
    $\epsilon_{\g} \defeq \rHam(\g, \str)$.
Since $\dec(\f) \neq \dec(\g)$,
by \cref{lem:Cut-hard:stripe:far,lem:Cut-hard:cons:far},
$\Vedge$ rejects $\f \circ \g$ with probability at least
\begin{align}
\begin{aligned}
    & \frac{2}{\rho Z} \cdot \frac{\rho \cdot \epsilon_{\f}}{k}
        + \frac{2}{\rho Z} \cdot \frac{\rho \cdot \epsilon_{\g}}{k}
        + \frac{1}{Z} \cdot \max\left\{
            \Bigl(1-2\epsilon_{\f} -2\epsilon_{\g} \Bigr)\frac{1}{k}, 0
        \right\}
    \\
    & = \frac{1}{Z\cdot k} \left(  
        2\epsilon_{\f} + 2\epsilon_{\g}
        + \max\Bigl\{
            1-2\epsilon_{\f} -2\epsilon_{\g}, 0
        \Bigr\}
    \right).
\end{aligned}
\end{align}
If $1-2\epsilon_{\f} -2\epsilon_{\g} < 0$,
this value is at least
\begin{align}
    \frac{1}{Z\cdot k} \Bigl(  
        2\epsilon_{\f} + 2\epsilon_{\g}
    \Bigr)
    > \frac{1}{Z \cdot k}.
\end{align}
Otherwise, this value is at least
\begin{align}
    \frac{1}{Z\cdot k} \left(  
        2\epsilon_{\f} + 2\epsilon_{\g}
        + \Bigl(1-2\epsilon_{\f} -2\epsilon_{\g} \Bigr)
    \right)
    = \frac{1}{Z \cdot k},
\end{align}
as desired.
\end{proof}

Also, we give a lower bound $\frac{1}{2Z \cdot k}$ on $\Vedge$'s rejection probability for any two $k$-colorings.

\begin{lemma}
\label{lem:Cut-hard:edge:any}
    Let $\f,\g \colon [k]^2 \to [k]$ be any two $k$-colorings.
    Then, $\Vedge$ rejects $\f \circ \g$ with probability at least
    $\frac{1}{2Z \cdot k}$.
\end{lemma}
\begin{proof} 
Owing to \cref{lem:Cut-hard:edge:mismatch},
    it is sufficient to bound the rejection probability in the case of $\dec(\f) = \dec(\g)$.
By \cref{lem:Cut-hard:stripe:far,lem:Cut-hard:cons:far},
$\Vedge$ rejects $\f \circ \g$ with probability at least
\begin{align}
\begin{aligned}
    & \frac{2}{\rho Z} \cdot \frac{\rho \cdot \epsilon_{\f}}{k}
        + \frac{2}{\rho Z} \cdot \frac{\rho \cdot \epsilon_{\g}}{k}
        + \frac{1}{Z} \cdot \max\left\{
            \Bigl(1-2\epsilon_{\f} -2\epsilon_{\g}\Bigr)\frac{1}{2k}, 0
        \right\}
    \\
    & = \frac{1}{Z\cdot k} \left(  
        2\epsilon_{\f} + 2\epsilon_{\g}
        + \max\left\{
            \frac{1-2\epsilon_{\f} -2\epsilon_{\g}}{2}, 0
        \right\}
    \right)
\end{aligned}
\end{align}
If $1-2\epsilon_{\f} -2\epsilon_{\g} < 0$,
this value is at least
\begin{align}
    \frac{1}{Z\cdot k} \Bigl(  
        2\epsilon_{\f} + 2\epsilon_{\g}
    \Bigr)
    > \frac{1}{Z \cdot k}.
\end{align}
Otherwise, this value is at least
\begin{align}
    \frac{1}{Z\cdot k} \left(  
        2\epsilon_{\f} + 2\epsilon_{\g}
        + \frac{1-2\epsilon_{\f} -2\epsilon_{\g} }{2}
    \right)
    \geq \frac{1}{2Z \cdot k},
\end{align}
as desired.
\end{proof}

\subsection{Putting Them Together: Proof of \texorpdfstring{\cref{lem:Cut-hard:crazy}}{Lemma~\protect\ref{lem:Cut-hard:crazy}}}
\label{sec:Cut-hard:crazy}

\paragraph{Reduction.}
Our gap-preserving reduction from \MMtwoCutReconf to \MMkCutReconf is described below.
Fix $k \geq k_0$,
$\epsilon_c,\epsilon_s \in (0,1)$ with $\epsilon_c < \epsilon_s$, and
$\Delta \in \bbN$.
Let $(G,\f_\sss,\f_\ttt)$ be an instance of \prb{Gap$_{1-\epsilon_c, 1-\epsilon_s}$ \twoCutReconf}, where
$G=(V,E)$ is a graph of maximum degree $\Delta \in \bbN$, and
$\f_\sss,\f_\ttt \colon V \to [2]$ are a pair of $2$-colorings of $G$.
We construct an instance $(H,\f'_\sss,\f'_\ttt)$ of \MMkCutReconf as follows.
For each vertex $v$ of $G$,
we create a fresh copy of $[k]^2$, denoted $S_v$; namely,
\begin{align}
    S_v \defeq \Bigl\{
        (v,x,y) \Bigm| (x,y) \in [k]^2
    \Bigr\},
\end{align}
and we define
\begin{align}
    V(H) \defeq \bigcup_{v \in V} S_v = V \times [k]^2.
\end{align}
Since a $k$-coloring $\f' \colon V \times [k]^2 \to [k]$ of $V(H)$
consists of a collection of $|V|$ $k$-colorings of $[k]^2$,
we will think of it as $\f' \colon V \to ([k]^2 \to [k])$ such that
$\f'(v)$ gives a $k$-coloring of $S_v$.

Consider the following verifier $\V_G$,
given oracle access to a $k$-coloring $\f' \colon V \to ([k]^2 \to [k])$:
\begin{itembox}[l]{\textbf{Overall verifier $\V_G$.}}
\begin{algorithmic}[1]
    \item[\textbf{Input:}]
        a graph $G = (V,E)$.
    \item[\textbf{Oracle access:}]
        a $k$-coloring $\f' \colon V \to ([k]^2 \to [k])$.
    \State select an edge $(v,w)$ of $G$ uniformly at random.
    \State execute $\Vedge$ on $\f'(v) \circ \f'(w)^\top$,
    where $\f'(w)^\top$ is the \emph{transposition} of $\f'(w)$; i.e.,
    $\f'(w)^\top(x,y) = \f'(w)(y,x)$ for all $(x,y) \in [k]^2$.
\end{algorithmic}
\end{itembox}
Create the set $E(H)$ of parallel edges between $V(H)$ so as to emulate $\V_G$ in a sense that 
for any $k$-coloring $\f' \colon V \times [k]^2 \to [k]$,
\begin{align}
    \val_H(\f') = \Pr\Bigl[ \V_G \text{ accepts } \f' \Bigr].
\end{align}
Since a pair of vertices $(v,x_1,y_1)$ and $(w,x_2,y_2)$ of $H$
might be selected by $\V_G$ only if
$(v,w) \in E$, 
the maximum degree of $H$ can be bounded by $\bigO(\Delta \cdot k^2)$.
For a $2$-coloring $\f \colon V \to [2]$ of $G$,
consider a $k$-coloring $\f' \colon V \times [k]^2 \to [k]$ of $H$ such that
$\f'(v)$ is horizontally striped if $\f(v)=1$ and
vertically striped if $\f(v)=2$; namely,
\begin{align}
    \f'(v,x,y) \defeq
    \begin{cases}
        y & \text{if } \f(v) = 1 \\
        x & \text{if } \f(v) = 2
    \end{cases}
    \text{ for all } (v,x,y) \in V \times [k]^2.
\end{align}
Construct a pair of $k$-colorings $\f'_\sss, \f'_\ttt \colon V \times [k]^2 \to [k]$ of $H$
from $\f_\sss, \f_\ttt$ according to the above procedure, respectively.
This completes the description of the reduction.

\paragraph{Correctness.}
We first show the following completeness.

\begin{lemma}
\label{lem:Cut-hard:complete}
    The following holds\textup{:}
    \begin{align}
        \opt_G\bigl(\f_\sss \reco \f_\ttt\bigr) \geq 1-\epsilon_c
        \implies \opt_H\bigl(\f'_\sss \reco \f'_\ttt\bigr) \geq 1-\frac{1+\epsilon_c}{2Z \cdot k} - \frac{\Delta}{|E|}.
    \end{align}
\end{lemma}
\begin{proof} 
Suppose $\opt_G(\f_\sss \reco \f_\ttt) = 1-\epsilon_c$.
It is sufficient to consider the case that
    $\f_\sss$ and $\f_\ttt$ differ in a single vertex, say $v^\star$.
Without loss of generality, we assume that
    $\f_\sss(v^\star) = 1$ and $\f_\ttt(v^\star) = 2$.
Note that $\f'_\sss$ and $\f'_\ttt$ differ only in vertices of $S_{v^\star}$.
Consider an irredundant reconfiguration sequence $\sqcol'$
from $\f'_\sss$ to $\f'_\ttt$,
which is obtained
by recoloring (some) vertices of $S_{v^\star}$.
For any intermediate $k$-coloring $\f'$ in $\sqcol'$, the following hold:
\begin{itemize}
    \item For at most $(1-\epsilon_c)$-fraction of edges $(v,w)$ of $G$,
        $\f'(v)$ and $\f'(w)$ are striped and
        $\dec(\f'(v)) \neq \dec(\f'(w))$.
        The edge verifier $\Vedge$ rejects $\f'(v) \circ \f'(w)^\top$ for such $(v,w)$ with probability $\frac{1}{2Z\cdot k}$
        due to \cref{lem:Cut-hard:edge:striped}.
    \item For at most $\epsilon_c$-fraction of edges $(v,w)$ of $G$,
        $\f'(v)$ and $\f'(w)$ are striped and
        $\dec(\f'(v)) = \dec(\f'(w))$.
        The edge verifier $\Vedge$ rejects $\f'(v) \circ \f'(w)^\top$ for such $(v,w)$
        with probability $\frac{1}{Z \cdot k}$
        due to \cref{lem:Cut-hard:edge:striped}.
    \item Since $\f'(v^\star)$, which is \emph{in transition}, may be far from being striped,
        for at most $\frac{\Delta}{m}$-fraction of edges $(v^\star,w)$ of $G$,
        the edge verifier $\Vedge$ may reject $\f'(v^\star) \circ \f'(w)^\top$ with probability at most $1$.
\end{itemize}
Consequently, $\V_G$ rejects $\f'$ with probability at most
\begin{align}
\begin{aligned}
    & (1-\epsilon_c) \cdot \frac{1}{2Z \cdot k}
    + \epsilon_c \cdot \frac{1}{Z \cdot k}
    + \frac{\Delta}{|E|} \cdot 1
    = (1+\epsilon_c)\cdot \frac{1}{2Z \cdot k} + \frac{\Delta}{|E|}, \\
    & \implies 
    \opt_H\bigl(\f_\sss \reco \f_\ttt\bigr) \geq 1-\frac{1+\epsilon_c}{2Z \cdot k} - \frac{\Delta}{|E|},
\end{aligned}
\end{align}
which completes the proof.
\end{proof}

We then show the following soundness.

\begin{lemma}
\label{lem:Cut-hard:sound}
    The following holds\textup{:}
    \begin{align}
        \opt_G\bigl(\f_\sss \reco \f_\ttt\bigr) < 1-\epsilon_s
        \implies \opt_H\bigl(\f'_\sss \reco \f'_\ttt\bigr) < 1 - \frac{1+\epsilon_s}{2Z \cdot k}.
    \end{align}
\end{lemma}
\begin{proof} 
Suppose $\opt_G(\f_\sss \reco \f_\ttt) < 1-\epsilon_s$.
Let $\sqcol' = (\f'^{(1)}, \ldots, \f'^{(T)})$
    be any reconfiguration sequence from $\f'_\sss$ to $\f'_\ttt$ such that
$\val_H(\sqcol') = \opt_H(\f'_\sss \reco \f'_\ttt)$.
Construct then
a new reconfiguration sequence
$\sqcol=(\f^{(1)}, \ldots, \f^{(T)})$
over $2$-colorings of $G$
from $\f_\sss$ to $\f_\ttt$ such that
$\f^{(t)}(v) \defeq \dec(\f'^{(t)}(v))$
for all $v \in V$.
Since $\sqcol$ is a valid reconfiguration sequence,
there exists $\f^{(t^\star)}$ in $\sqcol$ that makes more than $\epsilon_s$-fraction of edges of $G$ monochromatic,
denoted $M \subset E$.
Conditioned on the event that any edge of $M$ is sampled,
$\V_G$ rejects $\f'^{(t^\star)}$ with probability $\frac{1}{Z\cdot k}$ by \cref{lem:Cut-hard:edge:mismatch}.
On the other hand,
regardless of the sampled edge,
$\V_G$ rejects $\f'^{(t^\star)}$ with probability $\frac{1}{2Z \cdot k}$ by \cref{lem:Cut-hard:edge:any}.
Consequently, $\V_G$ rejects $\f'^{(t^\star)}$ with probability more than 
\begin{align}
\begin{aligned}
    & (1-\epsilon_s)\cdot \frac{1}{2Z \cdot k} + \epsilon_s \cdot \frac{1}{Z\cdot k}
    = \frac{1+\epsilon_s}{2Z \cdot k}, \\
    & \implies \opt_H\bigl(\f'_\sss \reco \f'_\ttt\bigr) < 1-\frac{1+\epsilon_s}{2Z \cdot k},
\end{aligned}
\end{align}
which completes the proof.
\end{proof}

We are now ready to prove \cref{lem:Cut-hard:crazy}.
\begin{proof}[Proof of \cref{lem:Cut-hard:crazy}]
Given
a graph $G$ of maximum degree $\Delta$ and
a pair of its $2$-colorings $\f_\sss,\f_\ttt \colon V(G) \to [2]$
as an instance of \prb{Gap$_{1-\epsilon_c,1-\epsilon_s}$ \twoCutReconf},
we create a multigraph $H$ of maximum degree $\bigO(\Delta \cdot k^2)$ and
a pair of its $k$-colorings $\f'_\sss,\f'_\ttt \colon V(H) \to [k]$
as an instance of \MMkCutReconf
according to the above reduction.
By \cref{lem:Cut-hard:complete,lem:Cut-hard:sound}, it holds that
\begin{align}
\begin{aligned}
    \opt_G\bigl(\f_\sss \reco \f_\ttt\bigr) \geq 1-\epsilon_c
    & \implies \opt_H\bigl(\f'_\sss \reco \f'_\ttt\bigr) \geq 1 - \frac{1+\epsilon_c}{2Z \cdot k} - \frac{\Delta}{|E|}, \\
    \opt_G\bigl(\f_\sss \reco \f_\ttt\bigr) < 1-\epsilon_s
    & \implies \opt_H\bigl(\f'_\sss \reco \f'_\ttt\bigr) < 1 - \frac{1+\epsilon_s}{2Z \cdot k}.
\end{aligned}
\end{align}
Without loss of generality, we can assume that
$|E|$ is sufficiently large so that
\begin{align}
    \frac{1+\epsilon_c}{2Z \cdot k} + \frac{\Delta}{|E|}
    < \frac{1+\frac{\epsilon_c+\epsilon_s}{2}}{2Z \cdot k}.
\end{align}
In fact, the above inequality holds when
\begin{align}
    |E| > \frac{4\Delta Z \cdot k}{\epsilon_s - \epsilon_c}.
\end{align}
Consequently, we obtain the following:
\begin{align}
\begin{aligned}
    \opt_G\bigl(\f_\sss \reco \f_\ttt\bigr) \geq 1-\epsilon_c
    & \implies \opt_H\bigl(\f'_\sss \reco \f'_\ttt\bigr) \geq 1 - \frac{\delta_c}{k}, \\
    \opt_G\bigl(\f_\sss \reco \f_\ttt\bigr) < 1-\epsilon_s
    & \implies \opt_H\bigl(\f'_\sss \reco \f'_\ttt\bigr) < 1 - \frac{\delta_s}{k},
\end{aligned}
\end{align}
where $\delta_c$ and $\delta_s$ are defined as 
\begin{align}
    \delta_c \defeq \frac{1+\frac{\epsilon_c+\epsilon_s}{2}}{2Z} \;\text{ and }\;
    \delta_s \defeq \frac{1+\epsilon_s}{2Z}.
\end{align}
Note that $\delta_s$ and $\delta_c$ do not depend on $k$, and $\delta_c < \delta_s$, as desired.
\end{proof}

Our construction of $H$ can also be used to derive the following gap-preserving reduction
from \prb{Max 2-Cut} to \prb{Max $k$-Cut},
which reproves the $\NP$-hardness of approximating \prb{Max $k$-Cut}
within a factor of $1 - \Omega\left(\frac{1}{k}\right)$ \cite{kann1997hardness,guruswami2013improved}:
\begin{lemma}
\label{lem:Cut-hard:KKLP97}
    For a graph $G$ and a multigraph $H$ generated by the above reduction,
    the following hold\textup{:}
    \begin{align}
    \begin{aligned}
        \exists \f \colon V(G) \to [2], \; \val_G(\f) \geq 1-\epsilon_c
        & \implies
        \exists \f' \colon V(H) \to [k], \; \val_H(\f') \geq 1 - \frac{1+\epsilon_c}{2Z \cdot k}, \\
        \forall \f \colon V(G) \to [2], \; \val_G(\f) < 1-\epsilon_s
        & \implies
        \forall \f' \colon V(H) \to [k], \; \val_H(\f') < 1 - \frac{1+\epsilon_s}{2Z \cdot k}.
    \end{aligned}
    \end{align}
    Therefore,
    for every reals $\epsilon_c,\epsilon_s \in (0,1)$ with $\epsilon_c < \epsilon_s$,
    there exist reals $\delta_c, \delta_s \in (0,1)$ with $\delta_c < \delta_s$ such that
    for all sufficiently large $k \geq k_0 \defeq \kzero$,
    there exists a gap-preserving reduction from
    \prb{Gap$_{1-\epsilon_c,1-\epsilon_s}$ 2-Cut}
    to
    \prb{Gap$_{1-\frac{\delta_c}{k},1-\frac{\delta_s}{k}}$ $k$-Cut}.
\end{lemma}
\begin{proof}
    See the proofs of \cref{lem:Cut-hard:complete,lem:Cut-hard:sound}.
\end{proof}

\subsection{Rejection Rate of the Stripe Test: Proof of \texorpdfstring{\cref{lem:Cut-hard:stripe:far}}{Lemma~\protect\ref{lem:Cut-hard:stripe:far}}}
\label{sec:Cut-hard:stripe:far}

This subsection is devoted to the proof of \cref{lem:Cut-hard:stripe:far}.
Some notations and definitions are introduced below.
Fix $k \geq k_0 = \kzero$.
Let $\f \colon [k]^2 \to [k]$ be a $k$-coloring of $[k]^2$ such that
$\rHam(\f,\str) = \epsilon$
for some $\epsilon < \epsilon_0 \defeq \epsilonzero$.
Each $(x,y)$ in $[k]^2$ will be referred to as a \emph{point}.
Hereafter, let $X_1,Y_1,X_2,Y_2$ denote independent random variables uniformly chosen from $[k]$.
The stripe verifier $\Vstripe$ rejects $\f$ with probability
\begin{align}
    \Pr\Bigl[ \Vstripe \text{ rejects } \f \Bigr] \defeq
    \Pr\Bigl[
        \f(X_1,Y_1) = \f(X_2,Y_2) \Bigm| X_1 \neq X_2 \text{ and } Y_1 \neq Y_2
    \Bigr].
\end{align}
We say that \emph{$\Vstripe$ rejects $\f$ by color $\alpha \in [k]$} when
    $\Vstripe$ draws $(X_1,Y_1,X_2,Y_2)$ such that
    $\f(X_1,Y_1) = \f(X_2,Y_2) = \alpha$.
Such an event occurs with probability
\begin{align}
    \Pr\Bigl[ \Vstripe \text{ rejects } \f \text{ by } \alpha \Bigr]
    \defeq \Pr\Bigl[
        \f(X_1,Y_1) = \f(X_2,Y_2) = \alpha \Bigm| X_1 \neq X_2 \text{ and } Y_1 \neq Y_2
    \Bigr].
\end{align}
Note that
\begin{align}
    \Pr\Bigl[ \Vstripe \text{ rejects } \f \Bigr]
    = \sum_{\alpha \in [k]} \Pr\Bigl[ \Vstripe \text{ rejects } \f \text{ by } \alpha \Bigr].
\end{align}
For each color $\alpha \in [k]$, 
we use $\f^{-1}(\alpha)$ to denote the set of $(x,y)$'s such that $\f(x,y) = \alpha$; namely,
\begin{align}
    \f^{-1}(\alpha) \defeq \Bigl\{
        (x,y) \in [k]^2 \Bigm| \f(x,y) = \alpha
    \Bigr\}.
\end{align}
For each $x,y,\alpha \in [k]$, let
$R_{y,\alpha}$ denote the number of $x$'s such that $\f(x,y) = \alpha$ and
$C_{x,\alpha}$ denote the number of $y$'s such that $\f(x,y) = \alpha$; namely,
\begin{align}
\begin{aligned}
    R_{y,\alpha} & \defeq 
        \left|\Bigl\{ x \in [k] \Bigm| \f(x,y) = \alpha \Bigr\}\right|, \\
    C_{x,\alpha} & \defeq
        \left|\Bigl\{ y \in [k] \Bigm| \f(x,y) = \alpha \Bigr\}\right|.
\end{aligned}
\end{align}

Let $\f^*$ be a striped $k$-coloring of $[k]^2$ that is closest to $\f$.
Without loss of generality, we can assume that
$\f^*$ is horizontally striped, and that
the rows of $\f$ and $\f^*$ are rearranged so that
$\f^*(x,y) = y$ for all $(x,y) \in [k]^2$.
For each $y \in [k]$, let $D_y$ denote the set of $(x,y)$'s at which
$\f$ and $\f^*$ disagree, and
let $D$ denote the union of $D_y$'s for all $y \in [k]$; namely,
\begin{align}
\begin{aligned}
    D_y & \defeq \Bigl\{
        (x,y) \in [k]^2 \Bigm| x \in [k], \f(x,y) \neq \f^*(x,y)
    \Bigr\}, \\
    D & \defeq \bigcup_{y \in [k]} D_y
    = \Bigl\{
    (x,y) \in [k]^2 \Bigm| \f(x,y) \neq \f^*(x,y)
    \Bigr\}.
\end{aligned}
\end{align}
Note that $|D| = \epsilon k^2$.

Define further
    $\Good$ and $\Bad$ as the set of $y$'s such that
    $|D_y|$ is at most $0.99k$ and greater than $0.99k$, respectively; namely,
\begin{align}
\begin{aligned}
    \Good & \defeq \Bigl\{
        y \in [k] \Bigm| |D_y| \leq 0.99 k
    \Bigr\}, \\
    \Bad & \defeq \Bigl\{
        y \in [k] \Bigm| |D_y| > 0.99 k
    \Bigr\}.
\end{aligned}
\end{align}
Observe that $|\Bad| < 1.02 \epsilon k$ because
\begin{align}
    \Pr_{Y \sim [k]}\Bigl[Y \in \Bad\Bigr]
    = \Pr_{Y \sim [k]}\Bigl[|D_Y| > 0.99 k\Bigr]
    < \frac{\displaystyle \E_{Y \sim [k]}\Bigl[|D_Y|\Bigr]}{0.99 k}
    < 1.02 \epsilon.
\end{align}
For each color $\alpha \in [k]$, let $N_\alpha$ denote
the number of $(x,y)$'s in $D$ such that $\f(x,y) = \alpha$; namely,
\begin{align}
    N_\alpha \defeq \left|\Bigl\{
        (x,y) \in D \Bigm| \f(x,y) = \alpha
    \Bigr\}\right|.
\end{align}
For a set $S \subseteq [k]$ of colors,
we define $N(S)$ as the sum of $N_\alpha$ over $\alpha \in S$; namely,
\begin{align}
    N(S) \defeq \sum_{\alpha \in S} N_\alpha.
\end{align}
Denote
$N_\Good \defeq N(\Good)$ and $N_\Bad \defeq N(\Bad)$;
note that $N([k]) = N_\Good + N_\Bad = |D|$.

Lastly, we show the probability of $\Vstripe$ rejecting by color $\alpha$,
depending on the number of occurrences of $\alpha$ per row and column,
which will be used several times.

\begin{lemma}
\label{clm:Cut-hard:stripe:far:useful}
    For a $k$-coloring $\f \colon [k]^2 \to [k]$ such that
    color $\alpha$ appears at least $m \geq 100$ times,
    $R_{y,\alpha} \leq \theta k$ for all $y \in [k]$, and
    $C_{x,\alpha} \leq \theta k$ for all $x \in [k]$,
    it holds that
    \begin{align}
        \Pr\Bigl[ \Vstripe \text{ rejects } \f \text{ by } \alpha \Bigr]
        \geq \frac{m \cdot (m-\theta k)}{10^2 \cdot k^4}.
    \end{align}
\end{lemma}
\begin{proof}
Consider the following case analysis:
(1) $R_{y^*,\alpha} \geq \frac{m}{16}$ for some $y^*$,
(2) $C_{x^*,\alpha} \geq \frac{m}{16}$ for some $x^*$, and
(3) $R_{y,\alpha} < \frac{m}{16}$ for all $y$ and
    $C_{x,\alpha} < \frac{m}{16}$ for all $x$.

Suppose first $R_{y^*,\alpha} \geq \frac{m}{16}$ for some $y^*$.
Since $R_{y^*,\alpha} \leq \theta k$ by assumption, we have
\begin{align}
    \sum_{y \neq y^*} R_{y,\alpha} \geq m-\theta k,
\end{align}
and thus, $\Vstripe$ rejects $\f$ by $\alpha$ with the following probability:
\begin{align}
\begin{aligned}
    & \Pr_{\substack{X_1 \neq X_2 \\ Y_1 \neq Y_2}}\Bigl[ \f(X_1,Y_1) = \f(X_2,Y_2) = \alpha \Bigr] \\
    & \geq \Pr_{\substack{X_1 \neq X_2 \\ Y_1 \neq Y_2}}\Bigl[
        \f(X_1,Y_1) = \f(X_2,Y_2) = \alpha \text{ and }
        Y_1 = y^* \text{ and }
        Y_2 \neq y^*
    \Bigr] \\
    & = \Pr_{\substack{X_1 \neq X_2 \\ Y_1 \neq Y_2}}\Bigl[
        \f(X_1,Y_1) = \f(X_2,Y_2) = \alpha \Bigm|
        Y_1 = y^* \text{ and } Y_2 \neq y^*
    \Bigr] \cdot \Pr_{\substack{X_1 \neq X_2 \\ Y_1 \neq Y_2}}\Bigl[ Y_1 = y^* \text{ and } Y_2 \neq y^* \Bigr] \\
    & = \frac{1}{k} \cdot \Pr_{\substack{X_1 \neq X_2 \\ Y_2 \neq y^*}}\Bigl[
        \f(X_1,y^*) = \f(X_2,Y_2) = \alpha
    \Bigr] \\
    & = \frac{1}{k} \cdot \underbrace{\Pr_{\substack{X_1 \neq X_2 \\ Y_2 \neq y^*}}\Bigl[
        \f(X_1, y^*) = \alpha \Bigm| \f(X_2, Y_2) = \alpha
    \Bigr]}_{\geq \frac{\frac{m}{16}-1}{k-1}} \cdot 
    \underbrace{\Pr_{\substack{X_1 \neq X_2 \\ Y_2 \neq y^*}}\Bigl[
        \f(X_2,Y_2) = \alpha
    \Bigr]}_{\frac{m-\theta k}{k(k-1)}} \\
    & \geq \frac{(\frac{m}{16}-1)\cdot (m-\theta k)}{k^2(k-1)^2}
    \underbrace{\geq}_{m \geq 100} \frac{m \cdot (m-\theta k)}{20 k^4}.
\end{aligned}
\end{align}
Suppose next $C_{x^*,\alpha} \geq \frac{m}{16}$ for some $x^*$.
Similarly to the first case,
$\Vstripe$ rejects $\f$ by $\alpha$ with probability at least $\frac{m \cdot (m-\theta k)}{20 k^4}$.

Suppose then $R_{y,\alpha} < \frac{m}{16}$ for all $y$ and 
$C_{x,\alpha} < \frac{m}{16}$ for all $x$.
Let $A_1, \ldots, A_{2k}$ be $2k$ independent random variables uniformly chosen from $[2]$.
For each $(i,j) \in [2]^2$ and $(x,y) \in \f^{-1}(\alpha)$,
we define $B^{(i,j)}_{x,y}$ as
\begin{align}
    B^{(i,j)}_{x,y} \defeq \Bigl\llbracket A_x = i \text{ and } A_{y+k} = j \Bigr\rrbracket.
\end{align}
For each $(i,j) \in [2]^2$,
we define $Z^{(i,j)}$ as
\begin{align}
    Z^{(i,j)} \defeq \sum_{(x,y) \in \f^{-1}(\alpha)} B^{(i,j)}_{x,y}.
\end{align}
Observe that for each $(i,j) \in [2]^2$,
the collection of $B^{(i,j)}_{x,y}$'s is a read-$\frac{m}{16}$ family.
By the read-$k$ Chernoff bound (\cref{thm:read-k-Chernoff}), it holds that for any $\delta > 0$,
\begin{align}
    \Pr\left[ Z^{(i,j)} \leq \E\left[Z^{(i,j)}\right] - \delta m \right]
    \leq \exp\left( - \frac{2\delta m}{\frac{m}{16}} \right).
\end{align}
Since $\E[Z^{(i,j)}] = \frac{m}{4}$,
we let $\delta \defeq \frac{1}{8}$ to obtain
\begin{align}
    \Pr\left[ Z^{(i,j)} \leq \tfrac{m}{8} \right] \leq 
    \exp\left( -\frac{2 \cdot \frac{m}{8}}{\frac{m}{16}} \right)
    = \exp(-4)
    < 0.02.
\end{align}
Taking a union bound, we derive
\begin{align}
    \Pr\Bigl[ \exists (i,j) \in [2]^2 \text{ s.t.~} Z^{(i,j)} \leq \tfrac{m}{8} \Bigr] \leq 4 \cdot 0.02 < 1.
\end{align}
Therefore, there exist two partitions $(P_1,P_2)$ and $(Q_1,Q_2)$ of $[k]$ such that
\begin{align}
    \Bigl| \f^{-1}(\alpha) \cap (P_1 \times Q_1) \Bigr| \geq \frac{m}{8} \text{ and }
    \Bigl| \f^{-1}(\alpha) \cap (P_2 \times Q_2) \Bigr| \geq \frac{m}{8}.
\end{align}
Letting
$S_{1,1} \defeq \f^{-1}(\alpha) \cap (P_1 \times Q_1)$ and
$S_{2,2} \defeq \f^{-1}(\alpha) \cap (P_2 \times Q_2)$, we derive
\begin{align}
\begin{aligned}
    & \Pr\Bigl[ \f(X_1,Y_1) = \f(X_2,Y_2) = \alpha \Bigm|
        X_1 \neq X_2 \text{ and } Y_1 \neq Y_2 \Bigr] \\
    & \geq \Pr\Bigl[ (X_1,Y_1) \in S_{1,1} \text{ and } (X_2,Y_2) \in S_{2,2} \Bigm|
        X_1 \neq X_2 \text{ and } Y_1 \neq Y_2
    \Bigr] \\
    & \geq \Pr\Bigl[ (X_1,Y_1) \in S_{1,1} \text{ and } (X_2,Y_2) \in S_{2,2} \Bigr] \\
    & \geq \frac{m}{8k^2} \cdot \frac{m}{8k^2} = \frac{m^2}{64k^4}.
\end{aligned}
\end{align}
Consequently, we get
\begin{align}
    \Pr\Bigl[ \Vstripe \text{ rejects } \f \text{ by } \alpha \Bigr]
    \geq \min\left\{ \frac{m \cdot (m-\theta k)}{20 k^4}, \frac{m^2}{64k^4} \right\}
    \geq \frac{m\cdot(m-\theta k)}{10^2 \cdot k^4},
\end{align}
which completes the proof.
\end{proof}

Hereafter, we present the proof of \cref{lem:Cut-hard:stripe:far} by cases.
We first divide into two cases according to $N_\Good$.

\paragraph{\fbox{(Case 1) $N_\Good \geq 0.01 \epsilon k^2$.}}
We show that $\Vstripe$'s rejection probability is $\Omega\left(\frac{N_\Good}{k^3}\right)$.
See \cref{fig:Cut-hard:stripe:far:1} for illustration of its proof.

\begin{figure}[t]
    \centering
    \includegraphics[scale=0.3]{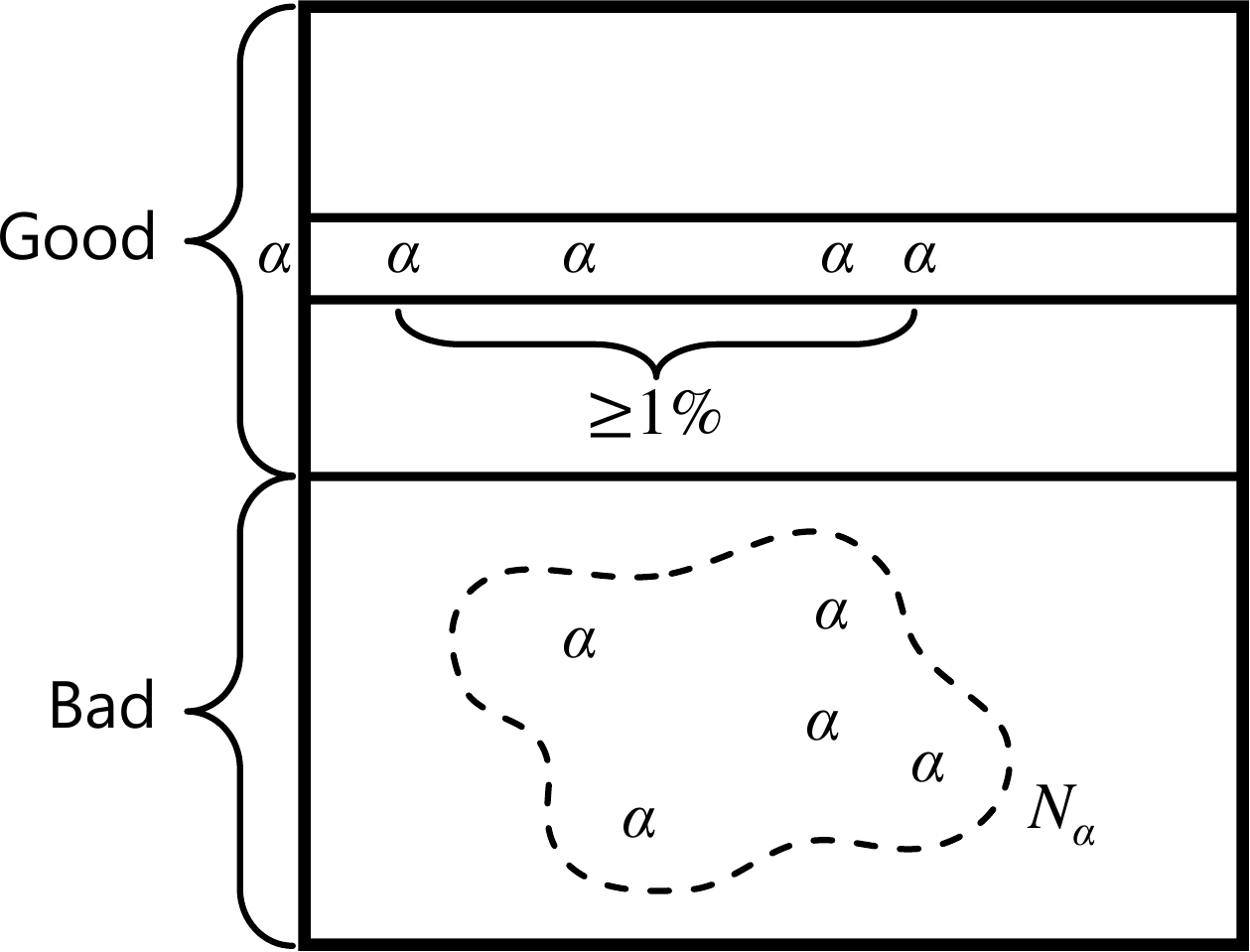}
    \caption{
        Illustration of the proof of \cref{clm:Cut-hard:stripe:far:1}.
        Since each color $\alpha$ of $\Good$ appears $0.01k$ times on the \nth{$\alpha$} row and
        $N_\Good = \Omega(\epsilon k^2)$,
        we can bound the rejection probability from below.
    }
    \label{fig:Cut-hard:stripe:far:1}
\end{figure}

\begin{claim}
\label{clm:Cut-hard:stripe:far:1}
    It holds that
    \begin{align}
        \Pr\Bigl[ \Vstripe \text{ rejects } \f \Bigr]
        \geq \frac{10^{-3}}{k^3} \cdot N_\Good
        \geq \frac{10^{-5}\cdot \epsilon}{k}.
    \end{align}
\end{claim}
\begin{proof}
Observe that $\Vstripe$'s rejection probability is
\begin{align}
\begin{aligned}
    & \Pr\Bigl[
        \f(X_1,Y_1) = \f(X_2,Y_2) \Bigm| X_1 \neq X_2 \text{ and } Y_1 \neq Y_2
    \Bigr] \\
    & = \frac{1}{k^2} \cdot
    \sum_{(x_1,y_1) \in [k]^2}
    \Pr_{\substack{
        X_2 \neq x_1 \\ Y_2 \neq y_1
    }}\Bigl[
        \f(x_1,y_1) = \f(X_2,Y_2)
    \Bigr] \\
    & \geq \frac{1}{k^2} \cdot
    \sum_{\substack{
        (x_1,y_1) \in D \\ \f(x_1,y_1) \in \Good
    }}
    \underbrace{
    \Pr_{\substack{
        X_2 \neq x_1 \\ Y_2 \neq y_1
    }}\Bigl[
        \f(x_1,y_1) = \f(X_2,Y_2) \text{ and }
        Y_2 = \f(x_1,y_1)
    \Bigr]}_{\bigstar \defeq} \\
\end{aligned}
\end{align}
For each point $(x_1,y_1) \in D$ such that $\f(x_1,y_1) \in \Good$,
we bound $\bigstar$ as follows.
\begin{align}
\begin{aligned}
    \bigstar 
    & = \Pr_{\substack{
        X_2 \neq x_1 \\ Y_2 \neq y_1
    }}\Bigl[
        \f(x_1,y_1) = \f(X_2,Y_2) \Bigm| Y_2 = \f(x_1,y_1)
    \Bigr] \cdot
    \Pr_{\substack{Y_2 \neq y_1}}\Bigl[
        Y_2 = \f(x_1,y_1)
    \Bigr] \\
    & = \Pr_{\substack{X_2 \neq x_1}}\Bigl[
        \f(x_1,y_1) = \f(X_2, \f(x_1,y_1))
    \Bigr] \cdot \frac{1}{k-1} \\
    & = \Pr_{\substack{X_2 \neq x_1}}\Bigl[
        \f^*(X_2, \f(x_1,y_1)) = \f(X_2, \f(x_1,y_1))
    \Bigr] \cdot \frac{1}{k-1} \\
    & = \left(\frac{k-1-|D_{\f(x_1,y_1)}|}{k-1} \right) \cdot \frac{1}{k-1} \\
    & \geq \left(\frac{k-1-0.99k}{k}\right) \cdot \frac{1}{k-1} \\
    & \underbrace{\geq}_{k \geq k_0} \frac{0.009}{k-1}
    \geq \frac{10^{-3}}{k}.
\end{aligned}
\end{align}
Consequently, $\Vstripe$'s rejection probability is at least
\begin{align}
    \frac{1}{k^2} \cdot \sum_{\substack{
        (x,y) \in D \\ \f(x,y) \in \Good
    }} \frac{10^{-3}}{k}
    = \frac{10^{-3}}{k^3} \cdot N_\Good
    \geq \frac{10^{-3}}{k^3} \cdot 0.01 \epsilon k^2
    > \frac{10^{-5}\cdot \epsilon}{k},
\end{align}
which completes the proof.
\end{proof}

\paragraph{\fbox{(Case 2) $N_\Good < 0.01 \epsilon k^2$.}}
Note that $N_\Bad > 0.99 \epsilon k^2$ by assumption.
We partition $\Bad$ into $\Badgtr$ and $\Badlss$ as follows:
\begin{align}
\begin{aligned}
    \Badgtr & \defeq
        \Bigl\{ \alpha \in \Bad \Bigm| N_\alpha \geq 1.01k \Bigr\}, \\
    \Badlss & \defeq
        \Bigl\{ \alpha \in \Bad \Bigm| N_\alpha < 1.01k \Bigr\}.
\end{aligned}
\end{align}

We will divide into two cases according to the size of $\Badgtr$.

\paragraph{\fbox{(Case 2-1) $|\Badgtr| \geq 0.01 \epsilon k$.}}

We show that $\Vstripe$'s rejection probability is $\Omega\left(\frac{|\Badgtr|}{k^2}\right)$.

\begin{claim}
\label{clm:Cut-hard:stripe:far:21}
It holds that
    \begin{align}
        \Pr\Bigl[ \Vstripe \text{ rejects } \f \Bigr]
        \geq \frac{10^{-4}}{k^2} \cdot |\Badgtr|
        \geq \frac{10^{-6} \cdot \epsilon}{k}.
    \end{align}
\end{claim}
\begin{proof}
By applying \cref{clm:Cut-hard:stripe:far:useful} with $\theta = 1$
to each color $\alpha$ of $\Badgtr$,
we have
\begin{align}
    \Pr\Bigl[ \Vstripe \text{ rejects } \f \text{ at } \alpha \Bigr]
    \geq \frac{N_\alpha \cdot (N_\alpha - k)}{10^2 \cdot k^4}
    \geq \frac{1.01k \cdot 0.01k}{10^2 \cdot k^4}
    > \frac{10^{-4}}{k^2}.
\end{align}
Consequently, $\Vstripe$'s rejection probability is at least
\begin{align}
\begin{aligned}
    \sum_{\alpha \in \Badgtr} \Pr\Bigl[ \Vstripe \text{ rejects } \f \text{ by } \alpha \Bigr]
    \geq \sum_{\alpha \in \Badgtr} \frac{10^{-4}}{k^2}
    = \frac{10^{-4}}{k^2} \cdot |\Badgtr|
    \geq \frac{10^{-6} \cdot \epsilon}{k},
\end{aligned}
\end{align}
completing the proof.
\end{proof}

\paragraph{\fbox{(Case 2-2) $|\Badgtr| < 0.01 \epsilon k$.}}

We will divide into two cases according to $N(\Badgtr)$.

\paragraph{\fbox{(Case 2-2-1) $N(\Badgtr) \geq 0.02\epsilon k^2$.}}

We show that $\Vstripe$'s rejection probability is
$\Omega\left(\frac{N(\Badgtr)}{k^3}\right)$ for very small $|\Badgtr|$.

\begin{claim}
\label{clm:Cut-hard:stripe:far:221}
It holds that
\begin{align}
    \Pr\Bigl[ \Vstripe \text{ rejects } \f \Bigr]
    \geq \frac{10^{-2}}{k^3} \cdot \Bigl(N(\Badgtr) - k \cdot |\Badgtr| \Bigr)
    \geq \frac{10^{-4} \cdot \epsilon}{k}.
\end{align}
\end{claim}
\begin{proof}
By applying \cref{clm:Cut-hard:stripe:far:useful} with $\theta = 1$
to each color $\alpha$ of $\Badgtr$,
we have
\begin{align}
\begin{aligned}
    \Pr\Bigl[ \Vstripe \text{ rejects } \f \text{ by } \alpha \Bigr]
    \geq \frac{N_\alpha \cdot (N_\alpha - k)}{10^2 \cdot k^4}
    \geq \frac{1.01k \cdot (N_\alpha - k)}{10^2 \cdot k^4}
    > \frac{N_\alpha - k}{10^2 \cdot k^3}.
\end{aligned}
\end{align}
Consequently, $\Vstripe$'s rejection probability is at least
\begin{align}
\begin{aligned}
    \sum_{\alpha \in \Badgtr} \Pr\Bigl[ \Vstripe \text{ rejects } \f \text{ by } \alpha \Bigr]
    & \geq \sum_{\alpha \in \Badgtr} \frac{N_\alpha - k}{10^2 \cdot k^3} \\
    & \geq \frac{10^{-2}}{k^3} \cdot \Bigl(N(\Badgtr) - k \cdot |\Badgtr| \Bigr) \\
    & > \frac{10^{-2}}{k^3} \cdot \Bigl( 0.02\epsilon k^2 - k \cdot 0.01\epsilon k \Bigr) \\
    & > \frac{10^{-4} \cdot \epsilon}{k},
\end{aligned}
\end{align}
which completes the proof.
\end{proof}

\paragraph{\fbox{(Case 2-2-2) $N(\Badgtr) < 0.02\epsilon k^2$.}}

By assumption,
$N(\Badlss)$ and $|\Badlss|$ can be bounded from below as follows:
\begin{align}
\begin{aligned}
    & 0.99\epsilon k^2
    < N_\Bad
    = N(\Badgtr) + N(\Badlss)
    < 0.02 \epsilon k^2 + N(\Badlss), \\
    & \implies N(\Badlss) > 0.97 \epsilon k^2.
\end{aligned}
\end{align}
\begin{align}
\begin{aligned}
    & 0.97 \epsilon k^2
    < N(\Badlss) = \sum_{\alpha \in \Badlss} N_\alpha
    < 1.01k \cdot |\Badlss|, \\
    & \implies |\Badlss| > 0.96 \epsilon k.
\end{aligned}
\end{align}
For each color $\alpha \in [k]$,
let $x_\alpha$ denote the column that includes the largest number of $\alpha$'s; namely,
\begin{align}
    x_\alpha \defeq \argmax_{x \in [k]} \Bigl\{ C_{x,\alpha} \Bigr\}.
\end{align}
Define $\SQ$ as a subset of $[k]^2$ obtained by excluding
the \nth{$x_\alpha$} column for every $\alpha \in \Bad$ and
the rows specified by $\Good$; namely,
\begin{align}
\begin{aligned}
    \SQ & \defeq
        \Bigl( [k] \setminus \Bigl\{x_\alpha \Bigm| \alpha \in \Bad \Bigr\} \Bigr) \times
        \Bigl( [k] \setminus \Good \Bigr).
\end{aligned}
\end{align}
See \cref{fig:Cut-hard:stripe:far:SQ} for illustration.
Note that the size of $\SQ$ is 
\begin{align}
\begin{aligned}
    |\SQ| & \geq (k-|\Bad|) \cdot |\Bad| \\
    & \geq (k-|\Bad|) \cdot |\Badlss| \\
    & \underbrace{>}_{|\Bad| < 1.02\epsilon k \text{ \& } |\Badlss| > 0.96\epsilon k}
        (k-1.02\epsilon k) \cdot 0.96 \epsilon k \\
    & \underbrace{>}_{\epsilon < \epsilon_0} 0.95 \cdot \epsilon k^2.
\end{aligned}
\end{align}

\begin{figure}[t]
    \centering
    \includegraphics[scale=0.3]{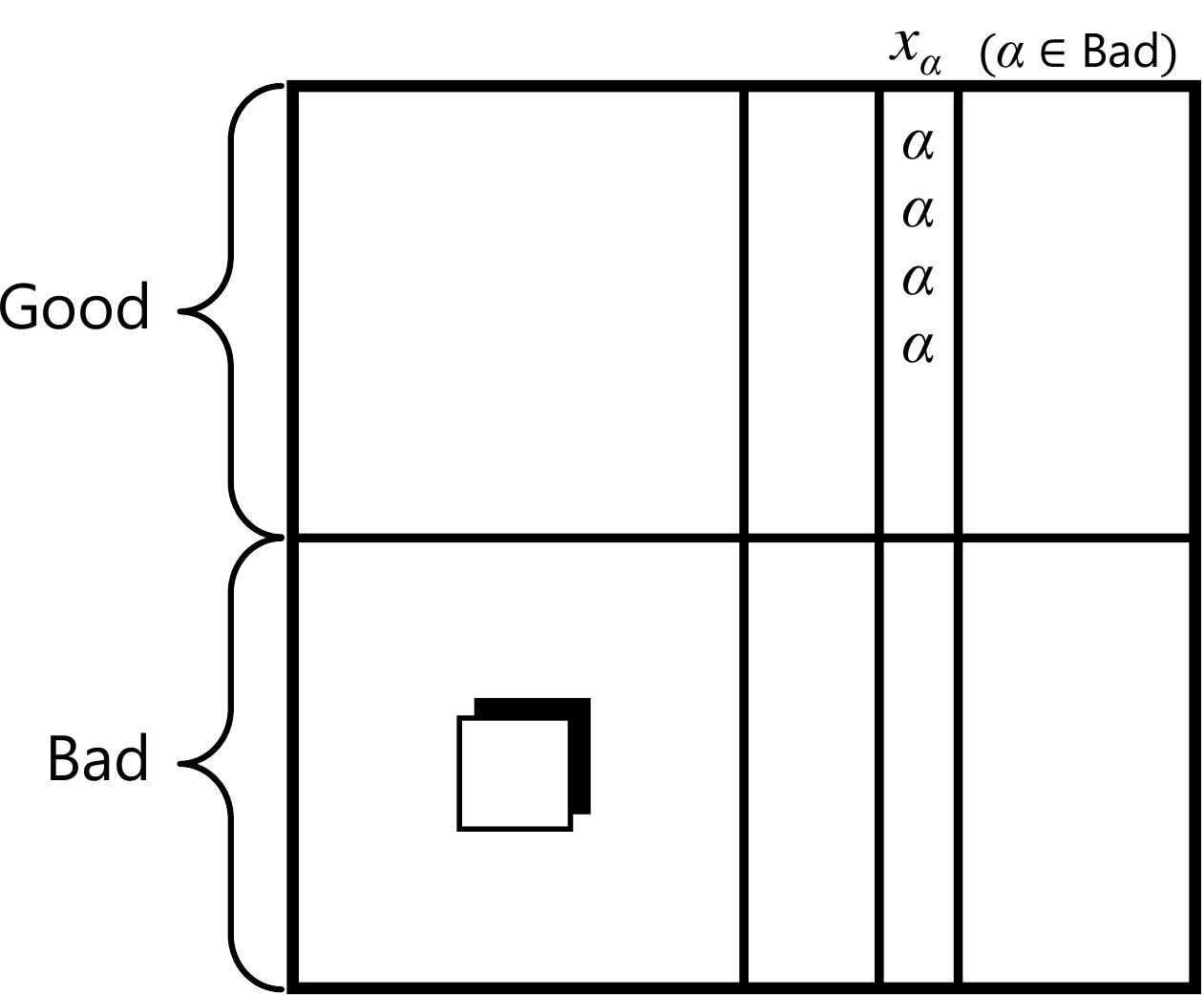}
    \caption{
        Illustration of $\SQ$, which exclude the \nth{$x_\alpha$} column for every $\alpha \in \Bad$.
    }
    \label{fig:Cut-hard:stripe:far:SQ}
\end{figure}

We show that most of the points of $\SQ$ are colored in $\Badlss$.

\begin{claim}
\label{clm:Cut-hard:stripe:far:222-aux}
It holds that
\begin{align}
    \Bigl| \f^{-1}(\Badlss) \cap \SQ \Bigr| > 0.91 \cdot \epsilon k^2,
\end{align}
namely, more than $0.91 \cdot \epsilon k^2$ points of $\SQ$ are colored in $\Badlss$.
\end{claim}
\begin{proof}
The number of points of $\SQ$ \emph{not} colored in $\Badlss$ can be bounded as follows:
\begin{itemize}
    \item Since $\SQ$ does not include \nth{$\alpha$} row for any $\alpha \in \Good$,
    the number of points colored in $\Good$ is
    $N_\Good < 0.01 \cdot \epsilon k^2$.
    \item The number of points colored in $\Badgtr$ that \emph{disagree} with $\f^*$ is
    $N(\Badgtr) < 0.02 \cdot \epsilon k^2$.
    \item The number of points colored in $\Badgtr$ that \emph{agree} with $\f^*$ is
    \begin{align}
        \sum_{\alpha \in \Badgtr} (k-|D_\alpha|)
        \underbrace{<}_{|D_\alpha| > 0.99k} 0.01k \cdot |\Badgtr|
        \underbrace{<}_{|\Badgtr| < 0.01\epsilon k} 10^{-4} \cdot \epsilon k^2.
    \end{align}
\end{itemize}
Consequently, the number of points of $\SQ$ colored in $\Badlss$ is
\begin{align}
\begin{aligned}
    \Bigl| \f^{-1}(\Badlss) \cap \SQ \Bigr| &
    = |\SQ| - \left(N_\Good + N(\Badgtr) + \sum_{\alpha \in \Badgtr} (k-|D_\alpha|) \right) \\
    & > 0.95 \cdot \epsilon k^2
        - \Bigl(0.01 \cdot \epsilon k^2 + 0.02 \cdot \epsilon k^2 + 10^{-4} \cdot \epsilon k^2\Bigr)
    > 0.91 \cdot \epsilon k^2,
\end{aligned}
\end{align}
as desired.
\end{proof}

We further partition $\Badlss$ into  $\BadlssL$ and $\BadlssS$ defined as
\begin{align}
\begin{aligned}
    \BadlssL & \defeq \Bigl\{ \alpha \in \Badlss \Bigm| C_{x_\alpha, \alpha} \geq 0.01k \Bigr\}, \\
    \BadlssS & \defeq \Bigl\{ \alpha \in \Badlss \Bigm| C_{x_\alpha, \alpha} < 0.01k \Bigr\}.
\end{aligned}
\end{align}
Below, we will divide into two cases according to the size of $\BadlssL$.

\paragraph{\fbox{(Case 2-2-2-1) $|\BadlssL| \geq 0.2 \cdot |\Badlss|$.}}
Note that $\BadlssS \leq 0.8 \cdot |\Badlss|$ by assumption.
We first show that a certain fraction of points of $\SQ$ are colored in $\BadlssL$.

\begin{claim}
\label{clm:Cut-hard:stripe:far:2221-aux}
It holds that
\begin{align}
    \Bigl| \f^{-1}(\BadlssL) \cap \SQ \Bigr| > 0.07 \cdot \epsilon k^2,
\end{align}
namely, more than $0.07 \cdot \epsilon k^2$ points of $\SQ$ are colored in $\BadlssL$.
\end{claim}
\begin{proof}
The number of points colored in $\BadlssS$ can be bounded as follows:
\begin{itemize}
\item The number of points colored in $\BadlssS$ that \emph{disagree} with $\f^*$ is
\begin{align}
\begin{aligned}
    N(\BadlssS)
    & = \sum_{\alpha \in \BadlssS} N_\alpha \\
    & \underbrace{<}_{N_\alpha < 1.01k}
        1.01k \cdot |\BadlssS| \\
    & \underbrace{\leq}_{|\BadlssS| \leq 0.8 \cdot |\Bad|}
        1.01k \cdot 0.8 \cdot |\Bad| \\
    & \underbrace{<}_{|\Bad| < 1.02 \epsilon k} 1.01k \cdot 0.8 \cdot 1.02 \epsilon k
    < 0.83\epsilon k^2.
\end{aligned}
\end{align}

\item The number of points colored in $\BadlssS$ that \emph{agree} with $\f^*$ is
\begin{align}
    \sum_{\alpha \in \BadlssS} (k-|D_\alpha|)
    \underbrace{<}_{|D_\alpha| > 0.99k} 0.01k \cdot |\BadlssS|
    \underbrace{<}_{|\Bad| < 1.02 \epsilon k} 0.01k \cdot 0.8 \cdot 1.02 \epsilon k
    < 0.01 \cdot \epsilon k^2.
\end{align}
\end{itemize}

Consequently,
by \cref{clm:Cut-hard:stripe:far:222-aux},
the number of points of $\SQ$ colored in $\BadlssL$ is
\begin{align}
\begin{aligned}
    \Bigl| \f^{-1}(\BadlssL) \cap \SQ \Bigr|
    & = \Bigl| \f^{-1}(\Badlss) \cap \SQ \Bigr|
        - \left(N(\BadlssS) + \sum_{\alpha \in \BadlssS} (k-|D_\alpha|)\right) \\
    & > 0.91 \cdot \epsilon k^2 - (0.83 \cdot \epsilon k^2 + 0.01 \cdot \epsilon k^2)
    = 0.07 \cdot \epsilon k^2,
\end{aligned}
\end{align}
as desired.
\end{proof}

We show that $\Vstripe$'s rejection probability is
$\Omega\left(\frac{|\f^{-1}(\BadlssL) \cap \SQ|}{k^3}\right)$.
See \cref{fig:Cut-hard:stripe:far:2221} for illustration of its proof.

\begin{figure}
    \centering
    \includegraphics[scale=0.3]{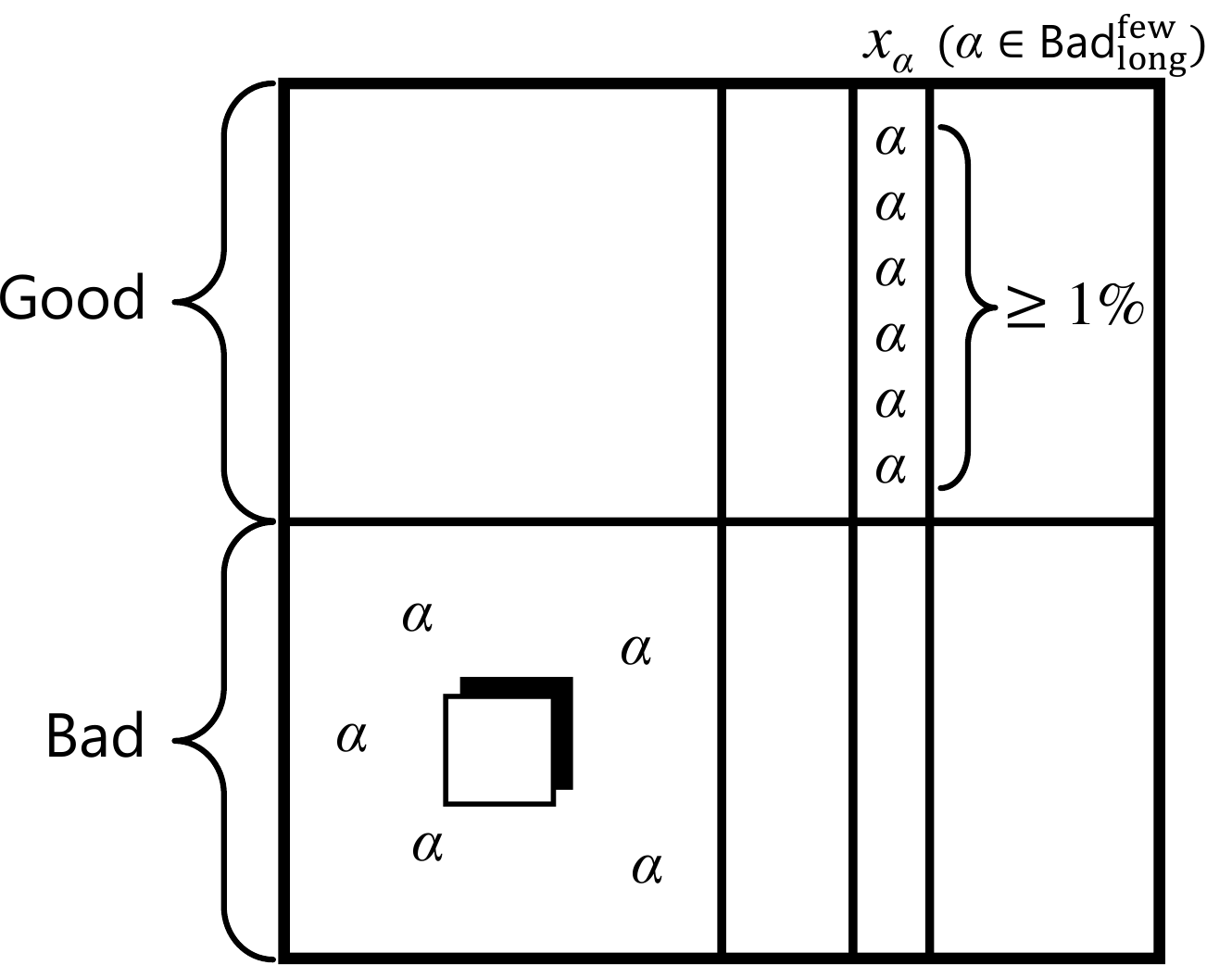}
    \caption{
        Illustration of the proof of \cref{clm:Cut-hard:stripe:far:2221}.
        Since each color $\alpha$ of $\BadlssL$ appears at least $0.01k$ times on the \nth{$x_\alpha$} column and
        $\BadlssL$ appears $\Omega(\epsilon k^2)$ times in $\SQ$
        due to \cref{clm:Cut-hard:stripe:far:2221-aux},
        we can bound the rejection probability from below.
    }
    \label{fig:Cut-hard:stripe:far:2221}
\end{figure}

\begin{claim}
\label{clm:Cut-hard:stripe:far:2221}
It holds that
\begin{align}
    \Pr\Bigl[ \Vstripe \text{ rejects } \f \Bigr]
    \geq \frac{10^{-3}}{k^3} \cdot \Bigl|\f^{-1}(\BadlssL) \cap \SQ\Bigr|
    \geq \frac{10^{-5} \cdot \epsilon}{k}.
\end{align}
\end{claim}
\begin{proof}
For each color $\alpha \in \BadlssL$,
$\Vstripe$ rejects $\f$ by $\alpha$ with probability
\begin{align}
\begin{aligned}
    & \Pr\Bigl[
        \f(X_1,Y_1) = \f(X_2,Y_2) = \alpha \Bigm| X_1 \neq X_2 \text{ and } Y_1 \neq Y_2
    \Bigr] \\
    & \geq \Pr_{\substack{X_1 \neq X_2 \\ Y_1 \neq Y_2}}\Bigl[
        \f(X_1,Y_1) = \f(X_2,Y_2) = \alpha \text{ and }
        X_1 = x_\alpha \text{ and }
        (X_2,Y_2) \in \f^{-1}(\alpha) \cap \SQ
    \Bigr] \\
    & = \underbrace{\Pr_{\substack{X_1 \neq X_2 \\ Y_1 \neq Y_2}}\Bigl[
        \f(X_1,Y_1) = \f(X_2,Y_2) = \alpha \Bigm|
        X_1 = x_\alpha \text{ and }
        (X_2,Y_2) \in \f^{-1}(\alpha) \cap \SQ
    \Bigr]}_{\text{(first term)}} \\
    & \qquad \cdot \underbrace{\Pr_{\substack{X_1 \neq X_2 \\ Y_1 \neq Y_2}}\Bigl[
        X_1 = x_\alpha \text{ and }
        (X_2,Y_2) \in \f^{-1}(\alpha) \cap \SQ
    \Bigr]}_{\text{(second term)}}.
\end{aligned}
\end{align}

Since $C_{x_\alpha,\alpha} \geq 0.01k$ by definition of $\BadlssL$
(i.e., there are $0.01k$ $y$'s such that $\f(x_\alpha, y) = \alpha$)
and $\SQ$ does not contain \nth{$x_\alpha$} column,
the first term can be bounded as follows:
\begin{align}
\begin{aligned}
    & \Pr_{\substack{X_1 \neq X_2 \\ Y_1 \neq Y_2}}\Bigl[
        \f(X_1,Y_1) = \f(X_2,Y_2) = \alpha \Bigm|
        X_1 = x_\alpha \text{ and }
        (X_2,Y_2) \in \f^{-1}(\alpha) \cap \SQ
    \Bigr] \\
    & = \Pr_{\substack{X_2 \neq x_\alpha \\ Y_1 \neq Y_2}}\Bigl[
        \f(x_\alpha,Y_1) = \alpha \Bigm|
        (X_2,Y_2) \in \f^{-1}(\alpha) \cap \SQ
    \Bigr] \\
    & \geq \frac{0.01k-1}{k-1}
    \underbrace{>}_{k \geq k_0} 10^{-3}.
\end{aligned}
\end{align}

The second term can be bounded as follows:
\begin{align}
\begin{aligned}
    & \Pr_{\substack{X_1 \neq X_2 \\ Y_1 \neq Y_2}}\Bigl[
        X_1 = x_\alpha \text{ and } (X_2,Y_2) \in \f^{-1}(\alpha) \cap \SQ
    \Bigr] \\
    & = \Pr_{\substack{X_1 \neq X_2 \\ Y_1 \neq Y_2}}\Bigl[X_1 = x_\alpha \Bigm| (X_2,Y_2) \in \f^{-1}(\alpha) \cap \SQ \Bigr]
    \cdot \Pr_{\substack{X_1 \neq X_2 \\ Y_1 \neq Y_2}}\Bigl[
        (X_2,Y_2) \in \f^{-1}(\alpha) \cap \SQ
    \Bigr] \\
    & \geq \frac{1}{k-1} \cdot \frac{|\f^{-1}(\alpha) \cap \SQ|}{k^2}
    \geq \frac{|\f^{-1}(\alpha) \cap \SQ|}{k^3}.
\end{aligned}
\end{align}

Therefore, $\Vstripe$'s rejection probability is at least
\begin{align}
\begin{aligned}
    \sum_{\alpha \in \BadlssL}
        \Pr\Bigl[ \Vstripe \text{ rejects } \f \text{ by } \alpha \Bigr]
    & \geq \sum_{\alpha \in \BadlssL}
        \frac{10^{-3}}{k^3} \cdot \Bigl|\f^{-1}(\alpha) \cap \SQ\Bigr| \\
    & = \frac{10^{-3}}{k^3} \cdot \Bigl|\f^{-1}(\BadlssL) \cap \SQ\Bigr| \\
    & \underbrace{>}_{\text{\cref{clm:Cut-hard:stripe:far:2221-aux}}} \frac{10^{-3}}{k^3} \cdot 0.07 \cdot \epsilon k^2
    > \frac{10^{-5} \cdot \epsilon}{k},
\end{aligned}
\end{align}
which completes the proof.
\end{proof}

\paragraph{\fbox{(Case 2-2-2-2) $|\BadlssL| < 0.2 \cdot |\Badlss|$.}}

We first show that a large majority of the points of $\SQ$ are colored in $\BadlssS$.

\begin{claim}
\label{clm:Cut-hard:stripe:far:2222-aux}
The number of points of $\SQ$ colored in $\BadlssS$ that disagree with $\f^*$ is
more than $0.68 \cdot \epsilon k^2$.
\end{claim}
\begin{proof}
Observe the following:
\begin{itemize}
\item The number of points colored in $\BadlssL$ that disagree with $\f^*$ is
\begin{align}
    N(\BadlssL)
    \underbrace{<}_{N_\alpha < 1.01k}
        1.01k \cdot |\BadlssL|
    \underbrace{<}_{|\BadlssL| < 0.2 |\Badlss|}
        1.01k \cdot 0.2 \cdot 1.02 \epsilon k
    < 0.21 \cdot \epsilon k^2.
\end{align}

\item The number of points colored in $\Badlss$ that \emph{agree} with $\f^*$ is
\begin{align}
    \sum_{\alpha \in \Badlss} (k-|D_\alpha|)
    \underbrace{<}_{|D_\alpha| > 0.99k}
        0.01k \cdot |\Badlss|
    \underbrace{<}_{|\Badlss| \leq |\Bad|} 0.01k \cdot 1.02 \epsilon k
    < 0.02 \epsilon k^2.
\end{align}
\end{itemize}
Therefore, by \cref{clm:Cut-hard:stripe:far:222-aux},
the number of points of $\SQ$ colored in $\BadlssS$ that disagree with $\f^*$ is
\begin{align}
\begin{aligned}
    & \Bigl| \f^{-1}(\Badlss) \cap \SQ \Bigr| - \left(N(\BadlssL) + \sum_{\alpha \in \Badlss} (k-|D_\alpha|) \right) \\
    & > 0.91 \cdot \epsilon k^2 - (0.21 \cdot \epsilon k^2 + 0.02 \cdot \epsilon k^2)
    > 0.68 \cdot \epsilon k^2,
\end{aligned}
\end{align}
as desired.
\end{proof}

\cref{clm:Cut-hard:stripe:far:2222-aux} implies $N(\BadlssS) > 0.68 \cdot \epsilon k^2$.
Define 
\begin{align}
    \tilde{\BadlssS} \defeq \Bigl\{
        \alpha \in \BadlssS \Bigm| N_\alpha \geq 0.01k
    \Bigr\}.
\end{align}
Observe that
\begin{align}
\begin{aligned}
    N\Bigl(\tilde{\BadlssS}\Bigr)
    & \geq N(\BadlssS) -
        \sum_{\alpha \in \BadlssS} N_\alpha \cdot \Bigl\llbracket N_\alpha < 0.01k \Bigr\rrbracket \\
    & \geq N(\BadlssS) - 0.01k \cdot |\BadlssS| \\
    & \geq N(\BadlssS) - 0.01k \cdot |\Bad| \\
    & \geq 0.68\epsilon k^2 - 0.01k \cdot 1.02\epsilon k \\
    & > 0.66 \epsilon k^2.
\end{aligned}
\end{align}

We will show that $\Vstripe$'s rejection probability is
$\Omega\left(\frac{1}{k^3}\cdot N\Bigl(\tilde{\BadlssS}\Bigr)\right)$.
See \cref{fig:Cut-hard:stripe:far:2222} for illustration of its proof.

\begin{figure}[t]
    \centering
    \includegraphics[scale=0.3]{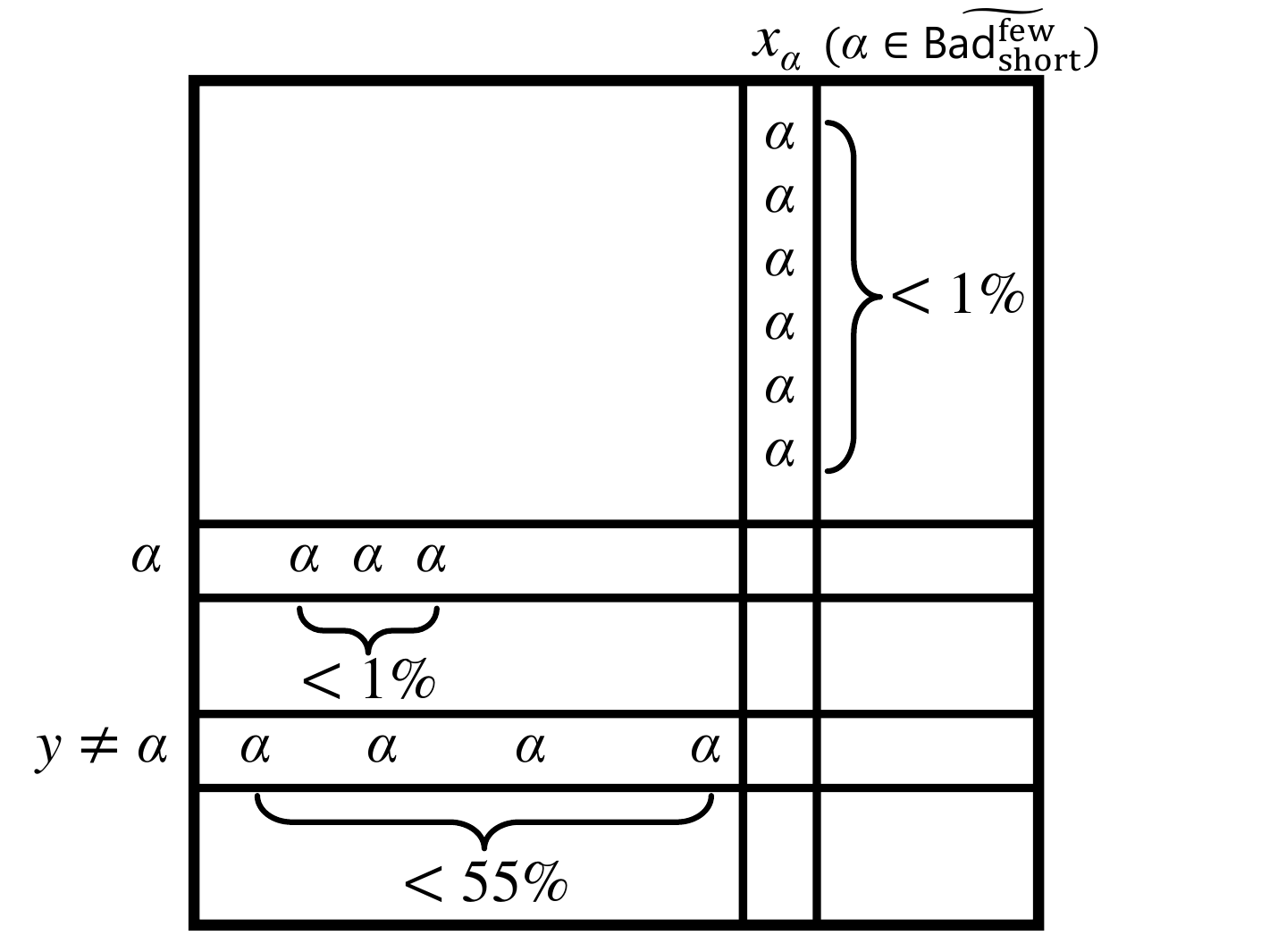}
    \caption{
        Illustration of the proof of \cref{clm:Cut-hard:stripe:far:2222}.
        Since $\alpha \in \tilde{\BadlssS}$ appears at most $0.55k$ times for any column and row
        and $N\Bigl(\tilde{\BadlssS}\Bigr) = \Omega(\epsilon k^2)$
        due to \cref{clm:Cut-hard:stripe:far:2222-aux},
        we can apply \cref{clm:Cut-hard:stripe:far:useful} to bound
        the rejection probability from below.
    }
    \label{fig:Cut-hard:stripe:far:2222}
\end{figure}

\begin{claim}
\label{clm:Cut-hard:stripe:far:2222}
    \begin{align}
        \Pr\Bigl[ \Vstripe \text{ rejects } \f \Bigr]
        \geq \frac{10^{-5}}{k^3} \cdot \left(
            N\Bigl(\tilde{\BadlssS}\Bigr) - 0.6k \cdot \Bigl|\tilde{\BadlssS}\Bigr|
        \right)
        \geq \frac{10^{-6} \cdot \epsilon}{k}.
    \end{align}
\end{claim}
\begin{proof}
For each color $\alpha \in \tilde{\BadlssS}$,
we claim that $R_{y,\alpha} < 0.55k$ for all $y \in [k]$.
Suppose for contradiction that
$R_{y^*,\alpha} \geq 0.55k$ for some $y^* \in [k]$.
Note that $y^* \neq \alpha$ as
$R_{\alpha, \alpha} = k-|D_\alpha| < 0.01k$ by definition of $\Bad$.
Since $\f^*$ is closest to $\f$, we must have
\begin{align}
\begin{aligned}
    & R_{\alpha,\alpha} + R_{y^*,y^*} \geq R_{\alpha,y^*} + R_{y^*,\alpha} \\
    & \implies R_{y^*,y^*}
    \geq R_{\alpha,y^*} + R_{y^*,\alpha} - R_{\alpha,\alpha}
    \geq R_{y^*,\alpha} - R_{\alpha,\alpha}
    > 0.54k \\
    & \implies R_{y^*,\alpha} + R_{y^*,y^*} > k,
\end{aligned}
\end{align}
which is a contradiction.

Since $C_{x,\alpha} < 0.01k < 0.55k$ for all $x \in [k]$ by definition of $\BadlssS$,
we apply \cref{clm:Cut-hard:stripe:far:useful} with $\theta = 0.55$
to every color $\alpha \in \tilde{\BadlssS}$ and derive
\begin{align}
    \Pr\Bigl[ \Vstripe \text{ rejects } \f \text{ by } \alpha \Bigr]
    \geq \frac{N_\alpha \cdot (N_\alpha - 0.55k)}{10^2 \cdot k^3}
    \geq \frac{0.01k \cdot (N_\alpha - 0.55k)}{10^2 \cdot k^3}
    = \frac{N_\alpha - 0.55k}{10^4 \cdot k^3}.
\end{align}
Therefore, $\Vstripe$'s rejection probability is at least
\begin{align}
\begin{aligned}
    \Pr\Bigl[ \Vstripe \text{ rejects } \f \Bigr]
    & \geq \sum_{\alpha \in \tilde{\BadlssS}} \frac{N_\alpha - 0.55k}{10^4 \cdot k^3} \\
    & \geq \frac{10^{-4}}{k^3} \cdot \left(
        N\Bigl(\tilde{\BadlssS}\Bigr) - 0.55k \cdot \Bigl|\tilde{\BadlssS}\Bigr|
    \right) \\
    & \geq \frac{10^{-4}}{k^3} \cdot \left(
        N\Bigl(\tilde{\BadlssS}\Bigr) - 0.55k \cdot |\Bad|
    \right) \\
    & \geq \frac{10^{-4}}{k^3} \cdot \Bigl(
        0.66 \cdot \epsilon k^2 - 0.55k \cdot 1.02\epsilon k
    \Bigr)
    > \frac{10^{-6} \cdot \epsilon}{k},
\end{aligned}
\end{align}
which completes the proof.
\end{proof}

Using the claims shown so far, we eventually prove \cref{lem:Cut-hard:stripe:far}.

\begin{proof}[Proof of \cref{lem:Cut-hard:stripe:far}]
By \cref{clm:Cut-hard:stripe:far:1,clm:Cut-hard:stripe:far:21,clm:Cut-hard:stripe:far:221,clm:Cut-hard:stripe:far:2221,clm:Cut-hard:stripe:far:2222},
the following hold:
\begin{itemize}
    \item If $\f$ is $\epsilon$-far from being striped for any $\epsilon < \epsilon_0 = \epsilonzero$,
    then $\Vstripe$ rejects with probability more than $\frac{10^{-6} \cdot \epsilon}{k}$.
    \item If $\f$ is $\epsilonzero$-far from being striped,
    then $\Vstripe$ rejects with probability more than $\frac{10^{-6} \cdot \epsilonzero}{k}$.
\end{itemize}
Setting the rejection rate $\rho \defeq \rhozero$ completes the proof.
\end{proof}

%% file: Cut-alg.tex
\section{Deterministic $\left(1-\frac{2}{k}\right)$-factor Approximation Algorithm for \MMkCutReconf}
\label{sec:Cut-alg}

In this section, we develop
a deterministic $\left(1-\frac{2}{k}\right)$-factor approximation algorithm for 
\MMkCutReconf for every $k \geq 2$.

\begin{theorem}
\label{thm:Cut-alg}
For every integer $k \geq 2$ and every real $\epsilon > 0$,
there exists a deterministic polynomial-time algorithm that
given a simple graph $G$ and a pair of its $k$-colorings $\f_\sss,\f_\ttt$,
returns a reconfiguration sequence $\sqcol$ from $\f_\sss$ to $\f_\ttt$ such that
\begin{align}
    \val_G(\sqcol) \geq
    \left(1-\frac{1}{k}-\epsilon\right)^2
    \cdot \min\Big\{ \val_G(\f_\sss), \val_G(\f_\ttt) \Bigr\}.
\end{align}
In particular,
letting $\epsilon \defeq \frac{1}{k^3}$,
this algorithm approximates
\MMkCutReconf on simple graphs within a factor of
$\left(1-\frac{1}{k}-\frac{1}{k^3}\right)^2 \geq 1-\frac{2}{k}$.
\end{theorem}

\subsection{Outline of the Proof of \texorpdfstring{\cref{thm:Cut-alg}}{Theorem~\protect\ref{thm:Cut-alg}}}
Our proof of \cref{thm:Cut-alg} can be divided into the following three steps:
\begin{itemize}
    \item In \cref{subsec:Cut-alg:low-value},
        we deal with the case that 
        $\f_\sss$ or $\f_\ttt$ has a low value, say $o(1)$.
        We show how to safely transform such a $k$-coloring into a $\frac{1}{2}$-value $k$-coloring.

    \item In \cref{subsec:Cut-alg:low-degree},
        for a graph consisting only of ``low-degree'' vertices,
        we demonstrate that \emph{a random reconfiguration sequence via a random $k$-coloring}
        makes $\approx \left(1-\frac{1}{k}\right)^2$-fraction of edges bichromatic with high probability,
        whose proof is based on the read-$k$ Chernoff bound (\cref{thm:read-k-Chernoff}).

    \item In \cref{subsec:Cut-alg:high-degree},
        we handle ``high-degree'' vertices.
        Intuitively,
        $k$ colors are distributed almost evenly over
        each high-degree vertex's neighbors with high probability.
\end{itemize}

\subsection{Low-value Case}
\label{subsec:Cut-alg:low-value}
Here, we introduce the following lemma, which enables us to assume that
$\val_G(\f_\sss)$ and $\val_G(\f_\ttt)$ are at least $\frac{1}{2}$ without loss of generality.

\begin{lemma}
\label{lem:Cut-alg:low-value}
For a graph $G=(V,E)$ and
a $k$-coloring $\f_\sss \colon V \to [k]$ of $G$ such that
$\val_G(\f_\sss) < 1 - \frac{1}{k}$,
there exists a reconfiguration sequence $\sqcol$
from $\f_\sss$ to another $k$-coloring $\f'_\sss \colon V \to [k]$ such that
$\val_G(\f'_\sss) \geq 1-\frac{1}{k}$ and $\val_G(\sqcol) = \val_G(\f_\sss)$.
Such $\sqcol$ can be found in polynomial time.
\end{lemma}
\begin{proof} 
We first claim the following:
\begin{claim}
For a $k$-coloring $\f \colon V \to [k]$ with
$\val_G(\f) < 1 - \frac{1}{k}$,
there is another $k$-coloring $\f' \colon V \to [k]$ such that
$\val_G(\f') > \val_G(\f)$ and
$\f$ and $\f'$ differ in a single vertex.
\end{claim}
\begin{proof}
We prove the contrapositive.
Suppose that recoloring any single vertex does not decrease the number of monochromatic edges.
Then, for every $v \in V$, $\f(v)$ appears in at most $\frac{1}{k}$-fraction of $v$'s neighbors; i.e.,
at least $\left(1-\frac{1}{k}\right)$-fraction of the incident edges to $v$ must be bichromatic.
This implies $\val_G(\f) \geq 1-\frac{1}{k}$.
\end{proof}
By the above claim,
until $\val_G(\f) \geq 1-\frac{1}{k}$,
one can find a pair of vertex $v \in V$ and color $\alpha \in [k]$ such that
recoloring $v$ to $\alpha$ strictly increases the number of bichromatic edges, as desired.
\end{proof}

\subsection{Low-degree Case}
\label{subsec:Cut-alg:low-degree}

We then show that
if $G$ contains only ``low-degree'' vertices,
\emph{a random reconfiguration sequence to a random $k$-coloring}
makes $\approx \left(1-\frac{1}{k}\right)^2$-fraction of edges bichromatic with
with high probability.
For a graph $G = (V,E)$,
we define $\Delta \defeq |E|^{\frac{2}{3}}$, and
we say that
a vertex of $G$ is \emph{low degree} if its degree is less than $\Delta$, and
\emph{high degree} otherwise.

\begin{lemma}
\label{lem:Cut-alg:low-degree}
Let $G=(V,E)$ be a graph of maximum degree at most $\Delta = |E|^{\frac{2}{3}}$ and
$\f_\sss \colon V \to [k]$ be a proper $k$-coloring of $G$.
Consider a uniformly random $k$-coloring $\frnd \colon V \to[k]$ and
a random irredundant reconfiguration sequence $\sqcol$
uniformly chosen from $\stsqcol(\f_\sss \reco \frnd)$.
Then, it holds that
\begin{align}
    \Pr_{\sqcol}\left[
        \val_G(\sqcol) < \left(1-\tfrac{1}{k}\right)^2 - |E|^{-\frac{1}{4}}
    \right]
    < \exp\left(-2 \cdot |E|^\frac{1}{12}\right).
\end{align}
\end{lemma}

To prove \cref{lem:Cut-alg:low-degree},
we first show that 
each edge consistently remains bichromatic through $\sqcol$
with probability $\left(1-\frac{1}{k}\right)^2$.

\begin{lemma}
\label{lem:Cut-alg:outward}
Let $e = (v,w)$ be an edge and
$\f_\sss \colon \{v,w\} \to [k]$ be a proper $k$-coloring of $e$.
Consider a uniformly random $k$-coloring $\frnd \colon \{v,w\} \to [k]$ and
a random irredundant reconfiguration sequence $\sqcol$
uniformly chosen from $\stsqcol(\f_\sss \reco \frnd)$.
Then, $\sqcol$ keeps $e$ bichromatic with probability at least
$\left(1-\frac{1}{k}\right)^2$\textup{;}
namely,
\begin{align}
    \Pr\Bigl[
        \forall \f \in \sqcol, \;
        \f(v) \neq \f(w)
    \Bigr]
    \geq \left(1-\tfrac{1}{k}\right)^2.
\end{align}
\end{lemma}
\begin{proof} 
Denote
$\alpha_v \defeq \frnd(v)$,
$\alpha_w \defeq \frnd(w)$, and
$L \defeq \{\f_\sss(v), \f_\sss(w)\}$.
Consider the following case analysis on $\alpha_v$ and $\alpha_w$:

\begin{description}
    \item[(Case 1)]
        If $\alpha_v \neq \alpha_w$ and
            $\{\alpha_v, \alpha_w\} \cap L = \emptyset$: \\
        There are $(k-2)(k-3)$ such colorings in total.
        Observe easily that $\sqcol$ succeeds with probability $1$.
    \item[(Case 2)]
        If $\alpha_v = \f_\sss(v)$ and $\alpha_w \notin L$: \\
        There are $(k-2)$ such colorings.
        $\sqcol$ succeeds with probability $1$.
    \item[(Case 3)]
        If $\alpha_w = \f_\sss(w)$ and $\alpha_v \notin L$: \\
        There are $(k-2)$ such colorings.
        Similarly to (Case 3),
        $\sqcol$ succeeds with probability $1$.
    \item[(Case 4)]
        If $\alpha_v = \f_\sss(v)$ and $\alpha_w = \f_\sss(w)$: \\
        There is only one such coloring;
        $\sqcol$ always succeeds.
    \item[(Case 5)]
        If $\alpha_v = \f_\sss(w)$ and $\alpha_w \notin L$: \\
        There are $(k-2)$ such colorings.
        If $w$ is recolored at first, $\sqcol$ fails; i.e.,
        the success probability is $\frac{1}{2}$.
    \item[(Case 6)]
        If $\alpha_w = \f_\sss(v)$ and $\alpha_v \notin L$: \\
        There are $(k-2)$ such colorings.
        Similarly to the previous case,
        $\sqcol$ succeeds with probability $\frac{1}{2}$.
    \item[(Case 7)]
        If $\alpha_v = \f_\sss(w)$ and $\alpha_w = \f_\sss(v)$: \\
        There is only one such coloring.
        $\sqcol$ cannot succeed anyway.
    \item[(Case 8)]
        Otherwise (i.e., $\alpha_v = \alpha_w$): \\
        There are $k$ such colorings;
        $\sqcol$ never succeeds.
\end{description}
Summing the success probability over the preceding eight cases,
we derive
\begin{align}
\begin{aligned}
    & k^{-2} \cdot \Bigl(
    (k-2)(k-3) \cdot 1
    + (k-2) \cdot 1 + (k-2) \cdot 1 + 1 \cdot 1
    + (k-2) \cdot \tfrac{1}{2} + (k-2) \cdot \tfrac{1}{2}
    \Bigr)
    \\
    & = \left(1-\frac{1}{k}\right)^2,
\end{aligned}
\end{align}
completing the proof.
\end{proof}

We now apply the read-$k$ Chernoff bound to \cref{lem:Cut-alg:outward}
and prove \cref{lem:Cut-alg:low-degree}.

\begin{proof}[Proof of \cref{lem:Cut-alg:low-degree}]
Define $n \defeq |V|$ and $m \defeq |E|$.
Let
$\vec{I} = (I_v)_{v \in V}$
    be a random integer sequence distributed uniformly over $[k]^V$ and
$\f_\vec{I} \colon V \to [k]$
    denote a $k$-coloring of $G$ such that
    $\f_\vec{I}(v) \defeq I_v$ for all $v \in V$.
Let 
$\vec{R} = (R_v)_{v \in V}$
    be a random real sequence distributed uniformly over $(0,1)^V$ and
$\sigma_\vec{R} \colon [n] \to V$ denote an ordering of $V$ such that
    $R_{\sigma(1)} > R_{\sigma(2)} > \cdots > R_{\sigma(n)}$.
Let $\sqcol(\vec{I}, \vec{R}; \f_\sss)$ be
a random irredundant reconfiguration sequence from $\f_\sss$ to $\f_\vec{I}$
obtained by
recoloring vertex $\sigma(i)$ from 
$\f_\sss(\sigma(i))$ to $\f_\vec{I}(\sigma(i))$
(if $\f_\sss(\sigma(i)) \neq \f_\vec{I}(\sigma(i))$)
in the order of $\sigma(1), \ldots, \sigma(n)$.
Observe easily that $\sqcol(\vec{I}, \vec{R}; \f_\sss)$
is uniformly distributed over $\stsqcol(\f_\sss \reco \frnd)$.

For each edge $e = (v,w)$ of $G$,
let $Y_e$ be a random variable that takes $1$
if $e$ is bichromatic throughout $\sqcol(\vec{I}, \vec{R}; \f_\sss)$ and
takes $0$ otherwise; namely,
\begin{align}
    Y_e \defeq \Bigl\llbracket
        \forall \f \in \sqcol(\vec{I}, \vec{R}; \f_\sss), \;
        \f(v) \neq \f(w)
    \Bigr\rrbracket.
\end{align}
Note that each $Y_e$ is Boolean and depends only on $I_v$, $I_w$, $R_v$, and $R_w$;
thus, $Y_e$'s are a read-$\Delta$ family.
Let $Y$ be the sum of $Y_e$ over all edges $e$; namely,
\begin{align}
    Y \defeq \sum_{e \in E} Y_e.
\end{align}
Since $\sqcol(\vec{I}, \vec{R}; \f_\sss)$ is distributed uniformly over $\stsqcol(\f_\sss \reco \frnd)$,
by \cref{lem:Cut-alg:outward},
it holds that
\begin{align}
    \E[Y] \geq \left(1-\tfrac{1}{k}\right)^2 \cdot m.
\end{align}
By applying the read-$k$ Chernoff bound (\cref{thm:read-k-Chernoff}) to $Y_e$'s,
    we derive
\begin{align}
    \Pr\Bigl[Y \leq \E[Y] - m^{\frac{3}{4}}\Bigr]
    \leq \exp\left(-\frac{2\cdot m^{\frac{3}{4}}}{\Delta}\right)
    \underbrace{<}_{\Delta = m^\frac{2}{3}} \exp\left(-2\cdot m^\frac{1}{12}\right),
\end{align}
which completes the proof.
\end{proof}

\subsection{Handling High-degree Vertices}
\label{subsec:Cut-alg:high-degree}

We now handle high-degree vertices and show the following using \cref{lem:Cut-alg:low-degree}.

\begin{proposition}
\label{prp:Cut-alg:high-prob}
    Let $G=(V,E)$ be a simple graph such that $|E| \geq 10^6$, and
    $\f_\sss \colon V \to [k]$ be a $k$-coloring of $G$ such that
    $\val_G(\f_\sss) \geq \frac{1}{2}$.
    Let
    $\Vl$ and $\Vg$ be the set of low-degree and high-degree vertices, respectively,
    where $\Delta = |E|^\frac{2}{3}$.

    Consider a uniformly random $k$-coloring $\frnd \colon V \to [k]$ and
    a random reconfiguration sequence $\sqcol$ from $\f_\sss$ to $\frnd$
    uniformly chosen from $\stsqcol(\f_\sss \reco \breve{\f} \reco \frnd)$, where
    $\breve{\f}$ agrees
    with $\f_\sss$ on $\Vg$ and
    with $\frnd$ on $\Vl$; namely,
    \begin{align}
        \breve{\f}(v) \defeq
        \begin{cases}
            \f_\sss(v) & \text{if } v \in \Vg, \\
            \frnd(v) & \text{if } v \in \Vl.
        \end{cases}
    \end{align}
    Then, it holds that
    \begin{align}
        \Pr_{\sqcol}\left[
            \val_G(\sqcol)
            < \left(1-\tfrac{1}{k}\right)^2 \cdot \val_G(\f_\sss) - 5 \cdot |E|^{-\frac{1}{4}}
        \right]
        < \exp\left(- \Omega\left(k^{-5} \cdot |E|^{\frac{1}{24}}\right)\right).
    \end{align}
\end{proposition}

Define
$n \defeq |V|$,
$m \defeq |E|$, and
\begin{align}
\label{eq:Cut-alg:Vl-Vg}
\begin{aligned}
    \Vl \defeq \Bigl\{ v \in V \Bigm| d_G(v) \leq \Delta \Bigr\}, \\
    \Vg \defeq \Bigl\{ v \in V \Bigm| d_G(v) > \Delta \Bigr\}.
\end{aligned}
\end{align}
Partition $G$ into the following three subgraphs:
\begin{align}
\begin{aligned}
    G_1 & \defeq G[\Vl],
        & m_1 & \defeq |E(G_1)|, \\
    G_2 & \defeq G - (E(G[\Vl]) \cup E(G[\Vg])),
        & m_2 & \defeq |E(G_2)|, \\
    G_3 & \defeq G[\Vg],
        & m_3 & \defeq |E(G_3)|.
\end{aligned}
\end{align}
Roughly speaking,
$G_1$ is the subgraph of $G$ induced by $\Vl$,
$G_3$ is the subgraph of $G$ induced by $\Vg$, and
$G_2$ is the subgraph of $G$ obtained by leaving only the edges connecting between $G_1$ and $G_2$.
Note that the union of $E(G_1)$, $E(G_2)$, and $E(G_3)$ is equal to $E$.
Since $|\Vg| \cdot \Delta < 2m$,
it holds that $|\Vg| < 2\cdot m^{\frac{1}{3}}$; thus,
$m_3 \leq |\Vg|^2 < 4 \cdot m^{\frac{2}{3}}$.

Let $m_1'$, $m_2'$, and $m_3'$ denote 
the number of bichromatic edges in $G_1$, $G_2$, and $G_3$,
with respect to $\f_\sss$, respectively; namely,
\begin{align}
\begin{aligned}
    m_1' & \defeq m_1 \cdot \val_{G_1}(\f_\sss), \\
    m_2' & \defeq m_2 \cdot \val_{G_2}(\f_\sss), \\
    m_3' & \defeq m_3 \cdot \val_{G_3}(\f_\sss).
\end{aligned}
\end{align}
Note that $m_1'+m_2'+m_3' = m \cdot \val_G(\f_\sss) \geq \frac{m}{2}$ by assumption.

We first demonstrate that
the number of edges of $G_2$ that are bichromatic throughout $\sqcol$,
i.e., $m_2 \cdot \val_{G_2}(\sqcol)$,
is at least $\left(1-\frac{1}{k}\right)^2 \cdot m_2'$
with high probability.

\begin{lemma}
\label{lem:Cut-alg:high-degree}
It holds that
\begin{align}
    \Pr\Bigl[
        m_2 \cdot \val_{G_2}(\sqcol)
        < \left(1-\tfrac{1}{k}\right)^2 \cdot m_2'
    \Bigr]
    < 4 \cdot m^{\frac{2}{3}} \cdot \exp\left(
        -\frac{m_2'}{48 k^5 \cdot m^\frac{1}{3}}
    \right).
\end{align}
\end{lemma}
\begin{proof} 
Fix a high-degree vertex $v \in \Vg$.
Let $X_v$ denote a random variable for
the maximum number of monochromatic edges
between $v$ and $\Vl$, where
the maximum is taken over all $k$-colorings of $\sqcol$; namely,
\begin{align}
    X_v \defeq \max_{\f \in \sqcol} \left\{
        \sum_{w \in \nei_G(v) \cap \Vl}
        \Bigl\llbracket \f(v) = \f(w) \Bigr\rrbracket
    \right\}.
\end{align}
Observe easily that
\begin{align}
    m_2 \cdot \val_{G_2}(\sqcol)
    \geq \sum_{v \in \Vg} \Bigl(
        |\nei_G(v) \cap \Vl| - X_v
    \Bigr).
\end{align}
Let $M_v$ and $B_v$ be the set of vertices $w \in \nei_G(v) \cap \Vl$
such that $(v,w)$ is monochromatic and bichromatic on $\f_\sss$,
respectively; namely,
\begin{align}
    M_v & \defeq \Bigl\{ w \in \nei_G(v) \cap \Vl \Bigm| \f_\sss(v) = \f_\sss(w) \Bigr\}, \\
    B_v & \defeq \Bigl\{ w \in \nei_G(v) \cap \Vl \Bigm| \f_\sss(v) \neq \f_\sss(w) \Bigr\}.
\end{align}
Note that $M_v$ and $B_v$ form a partition of $\nei_G(v) \cap \Vl$, and
\begin{align}
    \sum_{v \in \Vl} |B_v| = m_2'.
\end{align}
Since $\f_\sss$ makes all edges of $B_v$ bichromatic, and
$\sqcol$ would recolor $v$ after recoloring all vertices of $B_v$,
we derive
\begin{align}
\begin{aligned}
    X_v & \leq |M_v| + \max\left\{
        \sum_{w \in B_v} \Bigl\llbracket \f_\sss(v) = \frnd(w) \Bigr\rrbracket,
        \sum_{w \in B_v} \Bigl\llbracket \frnd(v) = \frnd(w) \Bigr\rrbracket
    \right\} \\
    & \leq |M_v| +
        \max_{\alpha \in [k]} \left\{
            \sum_{w \in B_v} \Bigl\llbracket \alpha = \frnd(w) \Bigr\rrbracket
        \right\} \\
    & \leq |M_v| +
        \max_{\alpha \in [k]} \left|\Bigl\{
            w \in B_v \Bigm| \frnd(w) = \alpha
        \Bigr\}\right| \\
    & = |M_v| +
        \max_{\alpha \in [k]} \left\{Y_\alpha\right\},
\end{aligned}
\end{align}
where $Y_1, \ldots, Y_k$ are random variables that follow the multinomial distribution
with $|B_v|$ trials and event probabilities
$\underbrace{\tfrac{1}{k}, \ldots, \tfrac{1}{k}}_{k \text{ times}}$.
Here, the maximum of $Y_i$'s may exceed
$\left(\frac{1}{k}+\epsilon\right)n$ with small probability,
which will be proved later.

\begin{claim}
\label{clm:Cut-alg:high-degree:multinomial}
For $k$ random variables, $Y_1, \ldots, Y_k$,
that follow the multinomial distribution with $n$ trials and
event probabilities $\underbrace{\tfrac{1}{k}, \ldots, \tfrac{1}{k}}_{k \text{ times}}$,
    it holds that for any $\epsilon > 0$,
    \begin{align}
        \Pr\left[
            \max_{1 \leq i \leq k} \{Y_i\}
            > \left(\tfrac{1}{k}+\epsilon\right) \cdot n
        \right]
        < n \cdot \exp\left(-\frac{\epsilon^2 k n}{3}\right).
    \end{align}
\end{claim}

By taking a union bound of \cref{clm:Cut-alg:high-degree:multinomial} 
over all $v \in \Vg$ \emph{such that}
$|B_v| \geq \frac{\epsilon \cdot m_2'}{|\Vg|}$,
with probability at least
$1-|\Vg|^2 \cdot \exp\left(-\frac{\epsilon^2 k}{3} \frac{\epsilon \cdot m_2'}{|\Vg|}\right)$,
it holds that
\begin{align}
    X_v \leq |M_v| + \left(\tfrac{1}{k}+\epsilon\right) \cdot |B_v|
    \text{ for all } v \in \Vg 
    \text{ such that } |B_v| \geq \frac{\epsilon \cdot m_2'}{|\Vg|}.
\end{align}
In such a case, we derive
\begin{align}
\begin{aligned}
    m_2 \cdot \val_{G_2}(\sqcol)
    & \geq \sum_{v \in \Vg} \Bigl( |\nei_G(v) \cap \Vl| - X_v \Bigr) \\
    & \geq \sum_{v \in \Vg: |B_v| \geq \frac{\epsilon \cdot m_2'}{|\Vg|}} \Bigl( |M_v| + |B_v| - X_v \Bigr) \\
    & \geq \sum_{v \in \Vg: |B_v| \geq \frac{\epsilon \cdot m_2'}{|\Vg|}}
        \Bigl( |M_v| + |B_v| - \Bigl(|M_v| + \left(\tfrac{1}{k}+\epsilon\right) \cdot |B_v|\Bigr) \Bigr) \\
    & = \left(1-\tfrac{1}{k}-\epsilon\right) \cdot
        \sum_{v \in \Vg: |B_v| \geq \frac{\epsilon \cdot m_2'}{|\Vg|}} |B_v| \\
    & = \left(1-\tfrac{1}{k}-\epsilon\right) \cdot \left(
        \sum_{v \in \Vg} |B_v| -
        \sum_{v \in \Vg: |B_v| < \frac{\epsilon \cdot m_2'}{|\Vg|}} |B_v|
        \right) \\
    & \geq \left(1-\tfrac{1}{k}-\epsilon\right) \cdot (1-\epsilon) \cdot m_2'.
\end{aligned}
\end{align}
Setting finally $\epsilon \defeq \frac{1}{2k^2}$,
we have
\begin{align}
    m_2 \cdot \val_{G_2}(\sqcol) \geq 
    \left(1-\tfrac{1}{k}-\epsilon\right) \cdot (1-\epsilon) \cdot m_2'
    \geq \left(1-\tfrac{1}{k}\right)^2 \cdot m_2'
\end{align}
with probability at least
\begin{align}
    1-|\Vg|^2 \cdot \exp\left(-\frac{\epsilon^2 k}{3} \frac{\epsilon \cdot m_2'}{|\Vg|}\right)
    \geq 1 - 4 \cdot m^{\frac{2}{3}} \cdot \exp\left(
        -\frac{m_2'}{48 k^5 \cdot m^\frac{1}{3}}
    \right),
\end{align}
where we used the fact that $|\Vg| < 2\cdot m^\frac{1}{3}$,
as desired.
\end{proof}

\begin{proof}[Proof of \cref{clm:Cut-alg:high-degree:multinomial}]
It is sufficient to bound
$\Pr[Y_i > (\frac{1}{k}+\epsilon)n]$.
Since $Y_i$ follows a binomial distribution with $n$ trials and event probability $\frac{1}{k}$,
which has mean $\E[Y_i] = \frac{n}{k}$,
we apply the Chernoff bound (\cref{thm:Chernoff}) to obtain
\begin{align}
    \Pr\Bigl[Y_i \geq (1+k\epsilon) \E[Y_i] \Bigr]
    \leq \exp\left(-\frac{k^2\epsilon^2 \cdot \E[Y_i]}{3}\right)
    = \exp\left(-\frac{\epsilon^2 k n}{3}\right).
\end{align}
Taking a union bound over $Y_i$'s accomplishes the proof.
\end{proof}

We are now ready to prove \cref{prp:Cut-alg:high-prob}.
\begin{proof}[Proof of \cref{prp:Cut-alg:high-prob}]
Since $\sqcol$ is uniformly distributed over $\stsqcol(\f_\sss \reco \breve{\f})$,
we apply \cref{lem:Cut-alg:low-degree} to
the subgraph of $G_1$ induced by bichromatic edges with respect to $\f_\sss$
and obtain
\begin{align}
\label{eq:Cut-alg:high-prob:low-degree}
\begin{aligned}
    \Pr\left[
        m_1 \cdot \val_{G_1}(\sqcol)
        < \left(1-\tfrac{1}{k}\right)^2 \cdot m_1' - m_1'^{\frac{3}{4}}
    \right]
    < \exp\left(-2\cdot m_1'^{\frac{1}{12}}\right).
\end{aligned}
\end{align}
By \cref{lem:Cut-alg:high-degree}, we have
\begin{align}
\label{eq:Cut-alg:high-prob:high-degree}
\Pr\left[
    m_2 \cdot \val_{G_2}(\sqcol)
    < \left(1-\tfrac{1}{k}\right)^2 \cdot m_2'
\right]
< 4 \cdot m^{\frac{2}{3}} \cdot \exp\left(
    -\frac{m_2'}{48 k^5 \cdot m^\frac{1}{3}}
\right).
\end{align}

We proceed by a case analysis on $m_1'$ and $m_2'$:
\begin{description}
    \item[(Case 1)] if $m_1' \leq m^\frac{1}{2}$: \\
        Since $m_1' + m_3' \leq m^\frac{1}{2} + 4 \cdot m^\frac{2}{3} \leq 5 \cdot m^\frac{2}{3}$,
        we have
        \begin{align}
        \begin{aligned}
            m_2'
            & = (m_1'+m_2'+m_3') - (m_1'+m_3') \\
            & \geq m\cdot \val_G(f) - 5 \cdot m^\frac{2}{3} \\
            & \geq \tfrac{1}{2} \cdot m - 5 \cdot m^\frac{2}{3}
            \underbrace{\geq}_{m \geq 10^6} \tfrac{1}{4} \cdot m.
        \end{aligned}
        \end{align}
        Therefore, with probability at least
        \begin{align}
            1 - \text{\cref{eq:Cut-alg:high-prob:high-degree}}
            = 1 - 4 \cdot m^{\frac{2}{3}} \cdot \exp\left(
                -\frac{m_2'}{48 k^5 \cdot m^\frac{1}{3}}
            \right)
            \geq 1 - 4 \cdot m^{\frac{2}{3}} \cdot \exp\left(
                -\frac{m^\frac{2}{3}}{192 k^5}
            \right),
        \end{align}
        we derive
        \begin{align}
        &
        \begin{aligned}
            m \cdot \val_G(\sqcol)
            & \geq m_2 \cdot \val_{G_2}(\sqcol) \\
            & \geq m_2' \cdot \left(1-\tfrac{1}{k}\right)^2 \\
            & = (m_1'+m_2'+m_3') \cdot \left(1-\tfrac{1}{k}\right)^2
                - (m_1'+m_3') \cdot \left(1-\tfrac{1}{k}\right)^2 \\
            & \geq m \cdot \val_G(\f_\sss) \cdot \left(1-\tfrac{1}{k}\right)^2 - 5 \cdot m^\frac{2}{3}
        \end{aligned} \\
        & \implies \val_G(\sqcol)
            \geq \left(1-\tfrac{1}{k}\right)^2 \cdot \val_G(\f_\sss) - 5 \cdot m^{-\frac{1}{3}}.
        \end{align}
    \item[(Case 2)] if $m_2' \leq m^\frac{1}{2}$: \\
        Since $m_2' + m_3' \leq m^\frac{1}{2} + 4 \cdot m^\frac{2}{3} \leq 5 \cdot m^\frac{2}{3}$,
        we have
        \begin{align}
        \begin{aligned}
            m_1'
            & = (m_1'+m_2'+m_3') - (m_2'+m_3') \\
            & \geq m \cdot \val_G(f) - 5 \cdot m^\frac{2}{3} \\
            & \geq \tfrac{1}{2} \cdot m - 5 \cdot m^\frac{2}{3}
            \underbrace{\geq}_{m \geq 10^6} \frac{m}{4}.
        \end{aligned}
        \end{align}
        Thus, with probability at least
        \begin{align}
            1 - \text{\cref{eq:Cut-alg:high-prob:low-degree}}
            = 1 - \exp\left(-2\cdot m_1'^{\frac{1}{12}}\right)
            \geq 1 - \exp\left(- m^\frac{1}{12} \right),
        \end{align}
        we get
        \begin{align}
        &
        \begin{aligned}
            m \cdot \val_G(\sqcol)
            & \geq m_1 \cdot \val_{G_1}(\sqcol) \\
            & \geq m_1' \cdot \left(1-\tfrac{1}{k}\right)^2 - m_1'^{\frac{3}{4}} \\
            & = (m_1'+m_2'+m_3') \cdot \left(1-\tfrac{1}{k}\right)^2
                - (m_2'+m_3') \cdot \left(1-\tfrac{1}{k}\right)^2
                - m_1'^{\frac{3}{4}} \\
            & \geq m \cdot \val_G(\f_\sss) \cdot \left(1-\tfrac{1}{k}\right)^2
                - 6 \cdot m^\frac{3}{4}
        \end{aligned} \\
        & \implies \val_G(\sqcol)
            \geq \left(1-\tfrac{1}{k}\right)^2 \cdot \val_G(\f_\sss) - 6 \cdot m^{-\frac{1}{4}}.
        \end{align}
    \item[(Case 3)] if $m_1' > m^\frac{1}{2}$ and $m_2' > m^\frac{1}{2}$: \\
        Taking a union bound, with probability at least
        \begin{align}
        \begin{aligned}
            1
            - \text{\cref{eq:Cut-alg:high-prob:low-degree}}
            - \text{\cref{eq:Cut-alg:high-prob:high-degree}}
            & \geq 1
            - 4 \cdot m^{\frac{2}{3}} \cdot \exp\left(
                -\frac{m_2'}{48 k^5 \cdot m^\frac{1}{3}}
            \right)
            - \exp\left(-2\cdot m_1'^{\frac{1}{12}}\right) \\
            & \geq 1
            - 4 \cdot m^{\frac{2}{3}} \cdot \exp\left(
                -\frac{m^\frac{1}{6}}{48 k^5}
            \right)
            - \exp\left(-2\cdot m^{\frac{1}{24}}\right),
        \end{aligned}
        \end{align}
        we have
        \begin{align}
        &
        \begin{aligned}
            m \cdot \val_G(\sqcol)
            & \geq m_1' \cdot \val_{G_1}(\sqcol) + m_2' \cdot \val_{G_2}(\sqcol) \\
            & \geq m_1' \cdot \left(1-\tfrac{1}{k}\right)^2 - m_1'^{\frac{3}{4}}
            + m_2' \cdot \left(1-\tfrac{1}{k}\right)^2 \\
            & \geq (m_1'+m_2'+m_3') \cdot \left(1-\tfrac{1}{k}\right)^2
                - m_1'^{\frac{3}{4}} - m_3' \cdot \left(1-\tfrac{1}{k}\right)^2 \\
            & \underbrace{\geq}_{m_3' \leq 4\cdot m^\frac{2}{3}}
                m \cdot \val_G(\f_\sss) \cdot \left(1-\tfrac{1}{k}\right)^2
                - m^\frac{3}{4} -
                4 \cdot m^\frac{2}{3} \\
            & \geq m \cdot \left(1-\tfrac{1}{k}\right)^2 \cdot \val_G(\f_\sss) - 5 \cdot m^\frac{3}{4}
        \end{aligned} \\
        & \implies \val_G(\sqcol)
            \geq \left(1-\tfrac{1}{k}\right)^2 \cdot \val_G(\f_\sss) - 5 \cdot m^{-\frac{1}{4}}.
        \end{align}
\end{description}
Consequently, in either case, it holds that
\begin{align}
    \Pr\Bigl[
        \val_G(\sqcol)
        \geq \left(1-\tfrac{1}{k}\right)^2 \cdot \val_G(\f) - 6 \cdot m^{-\frac{1}{4}}
    \Bigr]
    \geq 1 - \exp\left( - \Omega\left(k^{-5} \cdot m^{\frac{1}{24}}\right) \right),
\end{align}
which completes the proof.
\end{proof}

\subsection{Putting Them Together: Proof of \texorpdfstring{\cref{thm:Cut-alg}}{Theorem~\protect\ref{thm:Cut-alg}}}

We eventually conclude the proof of \cref{thm:Cut-alg} using \cref{prp:Cut-alg:high-prob}.

\begin{proof}[Proof of \cref{thm:Cut-alg}]
Let $G=(V,E)$ be a simple graph on $m$ edges and
    $\f_\sss,\f_\ttt \colon V \to [k]$ be a pair of its $k$-colorings.
Let $\epsilon > 0$ be any small real.
Let $m_0 \colon \bbN \to \bbN$ be a function such that
    $m_0(k) = \Theta\left(k^{240}\right)$.
If $m < m_0(k)$,
any optimal reconfiguration sequence from $\f_\sss$ to $\f_\sss$
can be found by running a brute-force search,
which completes in $m^{\Theta(k)} = k^{\Theta(k)}$ time.

Hereafter, we assume $m > m_0(k)$.
If $\val_G(\f_\sss) < \frac{1}{2}$,
we use \cref{lem:Cut-alg:low-value} to replace $\f_\sss$ by $\f'_\sss$
such that $\val_G(\f'_\sss) \geq \frac{1}{2}$ and
$\opt_G(\f_\sss \reco \f_\ttt) \geq \val_G(\f_\sss)$.
A similar proprocessing can be applied to $\f_\ttt$ whenever $\val_G(\f_\ttt) < \frac{1}{2}$.
So, we can safely assume that
$\val_G(\f_\sss) \geq \frac{1}{2}$ and
$\val_G(\f_\ttt) \geq \frac{1}{2}$.

Define $\Delta \defeq m^\frac{2}{3}$, and
define $\Vl$ and $\Vg$ by \cref{eq:Cut-alg:Vl-Vg}.
Consider a uniformly random $k$-coloring $\frnd \colon V \to [k]$ of $G$ and
a random reconfiguration sequence
$\sqcol \defeq \sqcol_1 \circ \sqcol_2$
from $\f_\sss$ to $\f_\ttt$
obtained by concatenating
two reconfiguration sequences
$\sqcol_1 \sim \stsqcol(\f_\sss \reco \breve{\f}_\sss \reco \frnd)$ and
$\sqcol_2 \sim \stsqcol(\frnd \reco \breve{\f}_\ttt \reco \f_\ttt)$.
Here, $\breve{\f}_\sss$ agrees
with $\f_\sss$ on $\Vg$ and
with $\frnd$ on $\Vl$, and
$\breve{\f}_\ttt$ agrees 
with $\f_\ttt$ on $\Vg$ and
with $\frnd$ on $\Vl$; namely,
\begin{align}
    \breve{\f}_\sss(v) & \defeq
    \begin{cases}
        \f_\sss(v) & \text{if } v \in \Vg, \\
        \frnd(v) & \text{if } v \in \Vl,
    \end{cases} \\
    \breve{\f}_\ttt(v) & \defeq
    \begin{cases}
        \f_\ttt(v) & \text{if } v \in \Vg, \\
        \frnd(v) & \text{if } v \in \Vl.
    \end{cases}
\end{align}
Such $\sqcol$ can be generated by the following randomized procedure:

\begin{itembox}[l]{\textbf{Generating a random reconfiguration sequence $\sqcol$ from $\f_\sss$ to $\f_\ttt$.}}
\begin{algorithmic}[1]
    \item[\textbf{Input:}]
        a simple graph $G = (V,E)$ and
        two $k$-colorings $\f_\sss, \f_\ttt \colon V \to [k]$ of $G$.
    \State let $\Delta \defeq |E|^\frac{2}{3}$,
        $\Vl \defeq \{ v \in V \mid d_G(v) \leq \Delta \}$, and
        $\Vg \defeq \{ v \in V \mid d_G(v) > \Delta \}$.
    \State sample a $k$-coloring $\frnd \colon V \to [k]$ uniformly at random.
    \State sample an ordering $\sigma_{\leq \Delta}$ over $\Vl$ uniformly at random.
    \State sample an ordering $\sigma_{> \Delta}$ over $\Vg$ uniformly at random.
    \For{\textbf{each} vertex $v \in \Vl$ in order of $\sigma_{\leq \Delta}$} \Comment{reconfiguration from $\f_\sss$ to $\breve{\f}_\sss$}
        \If{$\f_\sss(v) \neq \frnd(v)$}
            \State recolor $v$ from $\f_\sss(v)$ to $\frnd(v)$.
        \EndIf
    \EndFor
    \For{\textbf{each} vertex $v \in \Vg$ in order of $\sigma_{> \Delta}$} \Comment{reconfiguration from $\breve{\f}_\sss$ to $\frnd$}
        \If{$\f_\sss(v) \neq \frnd(v)$}
            \State recolor $v$ from $\f_\sss(v)$ to $\frnd(v)$.
        \EndIf
    \EndFor
    \For{\textbf{each} vertex $v \in \Vg$ in order of $\sigma_{> \Delta}$} \Comment{reconfiguration from $\frnd$ to $\breve{\f}_\ttt$}
        \If{$\frnd(v) \neq \f_\ttt(v)$}
            \State recolor $v$ from $\frnd(v)$ to $\f_\ttt(v)$.
        \EndIf
    \EndFor
    \For{\textbf{each} vertex $v \in \Vl$ in order of $\sigma_{\leq \Delta}$} \Comment{reconfiguration from $\breve{\f}_\ttt$ to $\f_\ttt$}
        \If{$\frnd(v) \neq \f_\ttt(v)$}
            \State recolor $v$ from $\frnd(v)$ to $\f_\ttt(v)$.
        \EndIf
    \EndFor
\end{algorithmic}
\end{itembox}

By applying \cref{prp:Cut-alg:high-prob} on
$\sqcol_1$ and $\sqcol_2$ and
taking a union bound,
we have
\begin{align}
\begin{aligned}
    & \Pr\Bigl[
        \val_G(\sqcol)
        < \left(1 - \tfrac{1}{k}\right)^2 \cdot \min\Bigl\{\val_G(\f_\sss), \val_G(\f_\ttt)\Bigr\}
        - 5\cdot m^{-\frac{1}{4}}
    \Bigr] \\
    & \leq \Pr\Bigl[
        \val_G(\sqcol_1)
        < \left(1 - \tfrac{1}{k}\right)^2 \cdot \val_G(\f_\sss)
        - 5\cdot m^{-\frac{1}{4}}
    \Bigr] \\
    & + \Pr\Bigl[
        \val_G(\sqcol_2)
        < \left(1 - \tfrac{1}{k}\right)^2 \cdot \val_G(\f_\ttt)
        - 5\cdot m^{-\frac{1}{4}}
    \Bigr] \\
    & < 2 \cdot \exp\left(- \Omega\left(k^{-5} \cdot m^{\frac{1}{24}}\right)\right)
    < \exp\left( -\Omega\left( m^{\frac{1}{48}} \right) \right)
        & \text{since } k^{-5} = \Omega\left(m^{-\frac{1}{48}}\right).
\end{aligned}
\end{align}
In particular, we have
\begin{align}
\begin{aligned}
    \E\Bigl[
        \val_G(\sqcol)
    \Bigr]
    & \geq \left(
            \left(1 - \tfrac{1}{k}\right)^2 \cdot \min\Bigl\{\val_G(\f_\sss), \val_G(\f_\ttt)\Bigr\}
            - 5\cdot m^{-\frac{1}{4}}
        \right)
    \cdot \left(1 - \exp\left( -\Omega\left( m^{\frac{1}{48}} \right) \right)\right) \\
    & \geq \left(1 - \tfrac{1}{k}\right)^2 \cdot \min\Bigl\{\val_G(\f_\sss), \val_G(\f_\ttt)\Bigr\}
        - \bigO\left(m^{-\frac{1}{4}}\right) \\
    & \geq \left(1 - \tfrac{1}{k} -\epsilon\right)^2 \cdot \min\Bigl\{\val_G(\f_\sss), \val_G(\f_\ttt)\Bigr\}.
\end{aligned}
\end{align}
where the last inequality holds for all sufficiently large $m$.
Hence, we can apply the method of conditional expectations \cite{alon2016probabilistic}
to the aforementioned randomized procedure
to construct a reconfiguration sequence $\sqcol^*$ such that
\begin{align}
    \val_G(\sqcol^*)
    \geq \left(1-\tfrac{1}{k}-\epsilon\right)^2
    \cdot \min\Bigl\{\val_G(\f_\sss), \val_G(\f_\ttt)\Bigr\}
\end{align}
in deterministic polynomial time,
which accomplishes the proof.
\end{proof}

\subsection{A Simple \texorpdfstring{$\left(1-\frac{9}{k}\right)$}{(1-9/k)}-factor Approximation Algorithm}
For the sake of completeness,
we give a simple $\left(1-\frac{9}{k}\right)$-factor approximation algorithm for \MMkCutReconf.

\begin{observation}
\label{obs:Cut-alg:simple}
    Let $G = (V,E)$ be a graph and
    $\f_\sss, \f_\ttt \colon V \to [k]$ be a pair of its proper $k$-colorings.
    Consider a uniformly random $k$-coloring $\frnd \colon V \to [k]$ and
    a random irredundant reconfiguration sequence $\sqcol$
    uniformly chosen from $\stsqcol(\f_\sss \reco \frnd \reco \f_\ttt)$.
    Then, it holds that
    \begin{align}
        \E\Bigl[ \val_G(\sqcol) \Bigr] \geq 1 - \tfrac{9}{k}.
    \end{align}
    In particular,
    $\sqcol$ is a $\left(1-\frac{9}{k}\right)$-factor approximation for
    \MMkCutReconf in expectation.
\end{observation}
\begin{proof} 
It is sufficient to show that 
for each edge $e = (v,w)$ of $G$,
\begin{align}
     \Pr\Bigl[
        \forall \f \in \sqcol, \;
        \f(v) \neq \f(w)
    \Bigr]
    \geq 1-\tfrac{9}{k}.
\end{align}\noindent
Define
\begin{align}
    L & \defeq \Bigl\{ \f_\sss(v), \f_\sss(w), \f_\ttt(v), \f_\ttt(w) \Bigr\}.
\end{align}
Conditioned on the event that
\begin{itemize}
    \item $\{\frnd(v), \frnd(w)\} \cap L = \emptyset$, and
    \item $\frnd(v) \neq \frnd(w)$,
\end{itemize}
$\sqcol$ always keeps $e$ bichromatic.
Since there are $k^2$ possible $k$-colorings of $(v,w)$,
the desired event occurs with probability at least
\begin{align}
    \frac{(k-|L|)^2 - (k-|L|)}{k^2}
    \geq 1 - \frac{2|L|+1}{k}
    \geq 1 - \tfrac{9}{k},
\end{align}
which completes the proof.
\end{proof}

%% file: app.tex
\section{Omitted Proofs in \texorpdfstring{\cref{sec:Cut-hard}}{Section~\protect\ref{sec:Cut-hard}}}
\label{app:Cut-hard}

\subsection{Proof of \texorpdfstring{\cref{prp:Cut-hard:2Cut}}{Proposition~\protect\ref{prp:Cut-hard:2Cut}}}
\label{app:Cut-hard:2Cut}
In this subsection, we prove \cref{prp:Cut-hard:2Cut}, i.e.,
$\PSPACE$-hardness of approximating \MMtwoCutReconf.
Let us begin with $\PSPACE$-hardness of approximating \MMsixCutReconf, which is immediate from
\cite{hirahara2024probabilistically,ohsaka2023gap,bonsma2009finding} and \cref{lem:Cut-hard:quadratic}.

\begin{lemma}
\label{lem:Cut-hard:6Cut}
    There exist a universal constant $\epsilon_0 \in (0,1)$ such that
    \prb{Gap$_{1,1-\epsilon_0}$ \sixCutReconf} is $\PSPACE$-hard.
    Moreover, this same hardness result holds even if 
    the maximum degree of input graphs is bounded by some constant $\Delta \in \bbN$.
\end{lemma}
\begin{proof} 
By the PCRP theorem of \citet{hirahara2024probabilistically} and
a series of gap-preserving reductions of \citet{ohsaka2023gap},
\prb{Gap$_{1,1-\epsilon}$ Nondeterministic Constraint Logic}
is $\PSPACE$-hard for some constant $\epsilon \in (0,1)$.
Since a polynomial-time reduction from
\prb{Nondeterministic Constraint Logic} to
\prb{4-Coloring Reconfiguration} due to \citet{bonsma2009finding}
is indeed gap-preserving,
\prb{Gap$_{1,1-\Omega(\epsilon)}$ 4-Cut Reconfiguration}
on graphs of maximum degree $5$ is $\PSPACE$-hard.
Lastly, by \cref{lem:Cut-hard:quadratic},
\prb{Gap$_{1,1-\Omega(\epsilon)}$ \sixCutReconf}
on graphs of maximum degree $\bigO(1)$ is $\PSPACE$-hard,
completing the proof.
\end{proof}

Hereafter, we present a gap-preserving reduction from
\MMsixCutReconf to \MMtwoCutReconf,
which along with \cref{lem:Cut-hard:6Cut} implies \cref{prp:Cut-hard:2Cut}.

\begin{lemma}
\label{lem:Cut-hard:6-to-2}
    For every real $\epsilon \in (0,1)$ and
    every integer $\Delta \in \bbN$,
    there exists a gap-preserving reduction from
    \prb{Gap$_{1,1-\epsilon}$ \sixCutReconf}
    on graphs of maximum degree $\Delta$
    to
    \prb{Gap$_{1-\delta_c,1-\delta_s}$ \twoCutReconf}
    on graphs of maximum degree $\bigO(\Delta)$, where
    \begin{align}
        \delta_c \defeq \frac{19+\frac{\epsilon}{2}}{54} \text{ and }
        \delta_s \defeq \frac{19+\epsilon}{54}.
    \end{align}
\end{lemma}

\paragraph{Reduction.}
Our reduction from \prb{Gap$_{1,1-\epsilon}$ \sixCutReconf} to \prb{Gap$_{1-\delta_c,1-\delta_s}$ \twoCutReconf}
is described below.
Fix $\epsilon \in (0,1)$ and $\Delta \in \bbN$.
Let $(G,\f_\sss,\f_\ttt)$ be an instance of
\prb{Gap$_{1,1-\epsilon}$ \sixCutReconf}, where 
$G=(V,E)$ is a graph of maximum degree $\Delta$, and
$\f_\sss,\f_\ttt \colon V \to [6]$ are a pair of its (proper) $6$-colorings.
We construct an instance $(H,\f'_\sss,\f'_\ttt)$ of \MMtwoCutReconf as follows.
For each vertex $v$ of $G$, create a set of four fresh vertices,
denoted $Z_v \defeq \{z_{v,1}, z_{v,2}, z_{v,3}, z_{v,4}\}$.
Define
\begin{align}
    V(H) \defeq \bigcup_{v \in V} Z_v.
\end{align}
Consider the following verifier $\Vsix$,
given oracle access to a $2$-coloring $\f' \colon V(H) \to [2]$:
\begin{itembox}[l]{\textbf{Verifier $\Vsix$.}}
\begin{algorithmic}[1]
    \item[\textbf{Input:}]
        a graph $G=(V,E)$.
    \item[\textbf{Oracle access:}]
        a $2$-coloring $\f' \colon V(H) \to [2]$.
    \State select an edge $(v,w)$ from $E$ uniformly at random.
    \State select $r \sim [0,1]$.
    \If{$0 \leq r < \frac{4}{9}$} \Comment{with probability $\frac{4}{9}$}
        \State select a pair $z_{v,i} \neq z_{v,j}$ from $Z_v$ uniformly at random.
        \State let $\alpha \defeq \f'(z_{v,i})$ and $\beta \defeq \f'(z_{v,j})$.
    \ElsIf{$\frac{4}{9} \leq r < \frac{8}{9}$} \Comment{with probability $\frac{4}{9}$}
        \State select a pair $z_{w,i} \neq z_{w,j}$ from $Z_w$ uniformly at random.
        \State let $\alpha \defeq \f'(z_{w,i})$ and $\beta \defeq \f'(z_{w,j})$.
    \Else \Comment{with probability $\frac{1}{9}$}
        \State select $i$ from $[4]$ uniformly at random.
        \State let $\alpha \defeq \f'(z_{v,i})$ and $\beta \defeq \f'(z_{w,i})$.
    \EndIf
    \If{$\alpha = \beta$}
        \State declare \Reject.
    \Else
        \State declare \Accept.
    \EndIf
\end{algorithmic}
\end{itembox}

Create the set $E(H)$ of parallel edges so as to emulate $\Vsix$
in a sense that 
for any $2$-coloring $\f' \colon V(H) \to 2$ of $H$,
\begin{align}
    \val_H(\f') = \Pr\Bigl[ \Vsix \text{ accepts } \f' \Bigr].
\end{align}
Note that the maximum degree of $H$ can be bounded by $\bigO(\Delta)$.

We define an encoding function $\enc \colon [6] \to [2]^4$ such that
for any color $\alpha \in [6]$,
\begin{align}
    \enc(\alpha) \defeq
    \begin{cases}
        (1,1,2,2) & \text{if } \alpha = 1, \\
        (1,2,1,2) & \text{if } \alpha = 2, \\
        (1,2,2,1) & \text{if } \alpha = 3, \\
        (2,1,1,2) & \text{if } \alpha = 4, \\
        (2,1,2,1) & \text{if } \alpha = 5, \\
        (2,2,1,1) & \text{if } \alpha = 6. \\
    \end{cases}
\end{align}
For any $6$-coloring $\f \colon V \to [6]$ of $G$,
consider a $2$-coloring $\f' \colon V(H) \to [2]$ of $H$ such that
$\f'(z_{v,i}) \defeq \enc(\f(v))_i$ for all $z_{v,i} \in V(H)$.
Construct finally two $2$-colorings $\f'_\sss,\f'_\ttt$ of $H$ from $\f_\sss,\f_\ttt$ by this procedure, respectively,
which completes the description of the reduction.

\paragraph{Correctness.}
We first analyze the acceptance probability of $\Vsix$.
We define a decoding function $\dec \colon [2]^4 \to [6] \cup \{\bot\}$ such that
for any $2$-color vector $\bm{\alpha} \in [2]^4$,
\begin{align}
    \dec(\bm{\alpha}) \defeq
    \begin{cases}
        1 & \bm{\alpha}=(1,1,2,2), \\
        2 & \bm{\alpha}=(1,2,1,2), \\
        3 & \bm{\alpha}=(1,2,2,1), \\
        4 & \bm{\alpha}=(2,1,1,2), \\
        5 & \bm{\alpha}=(2,1,2,1), \\
        6 & \bm{\alpha}=(2,2,1,1), \\
        \bot & \text{otherwise},
    \end{cases}
\end{align}
For any $2$-coloring $\f'$ of $H$,
let $\f'(Z_v) \defeq (\f'(z_{v,1}), \f'(z_{v,2}), \f'(z_{v,3}), \f'(z_{v,4}))$
for each vertex $v \in V$.

\begin{lemma}
\label{lem:Cut-hard:6-to-2:verifier}
    Conditioned on the edge $(v,w) \in E$ selected by $\Vsix$,
    the following hold\textup{:}
    \begin{itemize}
        \item if
        $\dec(\f'(Z_v)) \neq \bot$,
        $\dec(\f'(Z_w)) \neq \bot$, and
        $\dec(\f'(Z_v)) \neq \dec(\f'(Z_w))$,
        then $\Vsix$ accepts with probability at least $\frac{35}{54}$\textup{;}
        \item otherwise, 
        $\Vsix$ accepts with probability at most $\frac{34}{54}$.
    \end{itemize}
\end{lemma}
\begin{proof} 
    Conditioned on the selected edge $(v,w)$,
    the acceptance probability of $\Vsix$ is equal to
    \begin{align}
    \label{eq:Cut-hard:6-to-2:verifier}
        \frac{4}{9} \cdot \Pr_{i \neq j}\Bigl[ \f'(z_{v,i}) \neq \f'(z_{v,j}) \Bigr] +
        \frac{4}{9} \cdot \Pr_{i \neq j}\Bigl[ \f'(z_{w,i}) \neq \f'(z_{w,j}) \Bigr] +
        \frac{1}{9} \cdot \Pr_{i}\Bigl[ \f'(z_{v,i}) \neq \f'(z_{w,i}) \Bigr].
    \end{align}

\begin{itemize}
    \item 
    Suppose first
        $\dec(\f'(Z_v)) \neq \bot$,
        $\dec(\f'(Z_w)) \neq \bot$, and
        $\dec(\f'(Z_v)) \neq \dec(\f'(Z_w))$.
    Since
    $4$ of $6$ pairs $(z_{v,i}, z_{v,j})$ in ${Z_v \choose 2}$ of are bichromatic,
    $4$ of $6$ pairs $(z_{w,i}, z_{w,j})$ in ${Z_w \choose 2}$ are bichromatic, and
    at least $2$ of $4$ pairs $(z_{v,i}, z_{w,i})$ are bichromatic,
    we have
    \begin{align}
        \text{\cref{eq:Cut-hard:6-to-2:verifier}} =
        \frac{4}{9} \cdot \frac{4}{6} + \frac{4}{9} \cdot \frac{4}{6} + \frac{1}{9} \cdot \frac{2}{4}
        = \frac{35}{54}.
    \end{align}

    \item
    Suppose next
    $\dec(\f'(Z_v)) \neq \bot$,
    $\dec(\f'(Z_w)) \neq \bot$, but
    $\dec(\f'(Z_v)) = \dec(\f'(Z_w))$.
    Since none of $4$ pairs $(z_{v,i}, z_{w,i})$ are bichromatic,
    we have
    \begin{align}
        \text{\cref{eq:Cut-hard:6-to-2:verifier}} =
        \frac{4}{9} \cdot \frac{4}{6} + \frac{4}{9} \cdot \frac{4}{6} + \frac{1}{9} \cdot \frac{0}{4}
        = \frac{32}{54}.
    \end{align}

    \item
    Suppose finally $\dec(\f'(Z_v))$ or $\dec(\f'(Z_w))$ is equal to $\bot$.
    Without loss of generality, we can assume $\dec(\f'(Z_v)) = \bot$.
    Since
    at most $3$ of $6$ pairs $(z_{v,i}, z_{v,j})$ could be bichromatic by definition of $\dec$,
    at most $4$ of $6$ pairs $(z_{w,i}, z_{w,j})$ are bichromatic, and
    at most $4$ of $4$ pairs $(z_{v,i}, z_{w,i})$ are bichromatic,
    we have
    \begin{align}
        \text{\cref{eq:Cut-hard:6-to-2:verifier}} =
        \frac{4}{9} \cdot \frac{3}{6} + \frac{4}{9} \cdot \frac{4}{6} + \frac{1}{9} \cdot \frac{3}{4}
        = \frac{34}{54}.
    \end{align}
\end{itemize}
The above case analysis completes the proof.
\end{proof}

We are now ready to prove \cref{lem:Cut-hard:6-to-2}.
\begin{proof}[Proof of \cref{lem:Cut-hard:6-to-2}]
We first prove the completeness; i.e., 
\begin{align}
    \opt_G\bigl(\f_\sss \reco \f_\ttt\bigr) = 1
    \implies
    \opt_H\bigl(\f'_\sss \reco \f'_\ttt\bigr) \geq 1-\delta_c.
\end{align}
It is sufficient to consider the case that $\f_\sss$ and $\f_\ttt$ differ in 
a single vertex, say $v^\star$.
Without loss of generality,
we can assume that
$|E|$ is sufficiently large so that $\tfrac{\Delta}{|E|} < \frac{\epsilon}{70}$.
Consider a reconfiguration sequence $\sqcol'$ from $\f'_\sss$ to $\f'_\ttt$
obtained by recoloring $z_{v^\star,i}$
from $\f'_\sss(z_{v^\star,i})$ to $\f'_\ttt(z_{v^\star,i})$
for each $i \in [4]$.
Conditioned on the selected edge not being incident to $v^\star$,
$\Vsix$ accepts any intermediate coloring of $\sqcol'$ with probability at least $\frac{35}{54}$
by \cref{lem:Cut-hard:6-to-2:verifier}.
Since at most $\Delta$ edges are incident to $v^\star$,
$\Vsix$ accepts any intermediate coloring of $\sqcol'$ with probability at least
\begin{align}
    \frac{35}{54} \cdot \left(1 - \frac{\Delta}{|E|}\right)
    > \frac{35}{54} \cdot \left( 1 -\frac{\epsilon}{70} \right)
    = 1-\frac{19+\frac{\epsilon}{2}}{54}
    = 1-\delta_c.
\end{align}

We then prove the soundness; i.e.,
\begin{align}
    \opt_G\bigl(\f_\sss \reco \f_\ttt\bigr) < 1-\epsilon
    \implies
    \opt_H\bigl(\f'_\sss \reco \f'_\ttt\bigr) < 1-\delta_s.
\end{align}
Let $\sqcol' = (\f'^{(1)}, \ldots, \f'^{(T)})$ be any reconfiguration sequence 
from $\f'_\sss$ to $\f'_\ttt$ such that
$\val_H(\sqcol') = \opt_H(\f'_\sss \reco \f'_\ttt)$.
Construct then a sequence $\sqcol = (\f^{(1)}, \ldots, \f^{(T)})$ from $\f_\sss$ to $\f_\ttt$
such that $\f^{(t)}$ is a $6$-coloring of $G$ defined as follows:
\begin{align}
    \f^{(t)}(v) \defeq
    \begin{cases}
        \dec(\f'^{(t)}(Z_v)) & \text{if it is not } \bot \\
        1 & \text{otherwise}
    \end{cases}
    \text{ for all } v \in V.
\end{align}
Since $\sqcol$ is a valid reconfiguration sequence,
it includes $\f^{(t)}$ such that
$\val_G(\f^{(t)}) < 1-\epsilon$; i.e.,
$\f^{(t)}$ makes more than $\epsilon$-fraction of edges of $G$ monochromatic.
For each of such monochromatic edges $(v,w)$,
it must hold that
either
$\dec(\f'^{(t)}(Z_v)) = \bot$,
$\dec(\f'^{(t)}(Z_w)) = \bot$, or
$\dec(\f'^{(t)}(Z_v)) = \dec(\f'^{(t)}(Z_w))$.
Consequently,
by \cref{lem:Cut-hard:6-to-2:verifier},
$\Vsix$ accepts $\f'^{(t)}$ with probability less than
\begin{align}
    \frac{35}{54}\cdot(1-\epsilon) + \frac{34}{54}\cdot \epsilon
    = \frac{35}{54} \cdot \left(1-\frac{\epsilon}{35}\right)
    = 1-\frac{19+\epsilon}{54}
    = 1-\delta_s,
\end{align}
which completes the proof.
\end{proof}

\subsection{Proof of \texorpdfstring{\cref{lem:Cut-hard:quadratic}}{Lemma~\protect\ref{lem:Cut-hard:quadratic}}}
\label{app:Cut-hard:quadratic}

In this subsection, we prove \cref{lem:Cut-hard:quadratic}; i.e.,
there is a gap-preserving reduction from
\MMtwoCutReconf to \MMkCutReconf for every $k \geq 3$.

\paragraph{Reduction.}
Our reduction from
\prb{Gap$_{1-\epsilon_c, 1-\epsilon_s}$ \twoCutReconf} to
\prb{Gap$_{1-\delta_c, 1-\delta_s}$ \kCutReconf}
is described below.
Fix $k \geq 3$, $\epsilon_c,\epsilon_s \in (0,1)$ with $\epsilon_c < \epsilon_s$, and $\Delta \in \bbN$.
Let $(G,\f_\sss,\f_\ttt)$ be an instance of
\prb{Gap$_{1-\epsilon_c, 1-\epsilon_s}$ \twoCutReconf},
where $G=(V,E)$ is a graph of maximum degree $\Delta \in \bbN$, and
$\f_\sss,\f_\ttt \colon V \to [2]$ are a pair of its $2$-colorings.
We construct an instance $(H,\f'_\sss,\f'_\ttt)$
of \MMkCutReconf as follows.
Create a copy of $V$, and a set of fresh $k$ vertices for each vertex $v \in V$,
denoted by $Z_v \defeq \{z_{v,1}, \ldots, z_{v,k}\}$.
Define
\begin{align}
    V(H) \defeq V \cup \bigcup_{v \in V} Z_v.
\end{align}
Generate a \emph{$3$-regular expander graph} $X$ on $V$ whose edge expansion is a positive real $h > 0$; i.e.,
\begin{align}
\begin{aligned}
    & \min_{\emptyset \subsetneq S \subsetneq V}
        \frac{|\partial_X(S)|}{\min\{|S|, |V \setminus S|\}} \geq h, \\
    & \text{where }  \partial_X(S) \defeq \Bigl\{ (v,w) \in E(X) \Bigm| v \in S, w \notin S \Bigr\}.
\end{aligned}
\end{align}
Such an expander graph can be constructed in polynomial time; see, e.g., \cite{gabber1981explicit,reingold2002entropy,hoory2006expander}.
Consider the following verifier $\Vquad$,
given oracle access to a $k$-coloring $\f' \colon V(H) \to [k]$,
parameterized by $p_1$ and $p_2$ with $p_1+p_2=1$,
whose values depend only on $k,\epsilon_c,\epsilon_s$ and will be determined later:

\begin{itembox}[l]{\textbf{Verifier $\Vquad$.}}
\begin{algorithmic}[1]
    \item[\textbf{Input:}]
        a graph $G = (V,E)$,
        a $3$-regular expander graph $X$,
        parameters $p_1,p_2 \in (0,1)$ with $p_1 + p_2 = 1$.
    \item[\textbf{Oracle access:}]
        a $k$-coloring $\f' \colon V(H) \to [k]$.
    \If{with probability $p_1$} \LComment{first test}
        \State select an edge $(v,w)$ of $X$ uniformly at random.
        \State select a pair $i \neq j$ from $[k]$ uniformly at random.
        \State let $\alpha \defeq \f'(z_{v,i})$ and $\beta \defeq \f'(z_{w,j})$.
    \Else \Comment{with probability $p_2$} \LComment{second test}
        \State select an edge $(v,w)$ of $G$ uniformly at random.
        \State select $r \sim [0,1]$.
        \If{$0 \leq r < \frac{1}{2k-1}$} \Comment{with conditional probability $\frac{1}{2k-1}$}
            \State let $\alpha \defeq \f'(v)$ and $\beta \defeq \f'(w)$.
        \ElsIf{$\frac{1}{2k-1} \leq r < \frac{k-1}{2k-1}$} \Comment{with conditional probability $\frac{k-2}{2k-1}$}
            \State select $i$ from $\{3,\ldots,k\}$ uniformly at random.
            \State let $\alpha \defeq \f'(v)$ and $\beta \defeq \f'(z_{v,i})$.
        \Else \Comment{with conditional probability $\frac{k-2}{2k-1}$}
            \State select $i$ from $\{3,\ldots,k\}$ uniformly at random.
            \State let $\alpha \defeq \f'(w)$ and $\beta \defeq \f'(z_{w,i})$.
        \EndIf
    \EndIf
    \If{$\alpha = \beta$}
        \State declare \Reject.
    \Else
        \State declare \Accept.
    \EndIf
\end{algorithmic}
\end{itembox}
Create the set $E(H)$ of parallel edges between $V(H)$ so as to emulate $\Vquad$
in a sense that for any $k$-coloring $\f'$ of $H$,
\begin{align}
    \val_H(\f') = \Pr\Bigl[ \Vquad \text{ accepts } \f' \Bigr].
\end{align}
Note that the maximum degree of $H$ can be bounded by $\bigO(\Delta + \poly(k))$.
Construct finally two $k$-colorings $\f'_\sss,\f'_\ttt$ of $H$ such that
$\f'_\sss(v) = \f_\sss(v)$ and $\f'_\ttt(v) = \f_\ttt(v)$
for all $v \in V$, and
$\f'_\sss(z_{v,i}) = \f'_\ttt(z_{v,i}) = i$
for all $v \in V$ and $i \in [k]$.
This completes the description of the reduction.

\paragraph{Correctness.}
We first investigate the (conditional) rejection probability of the first test.
We say that $Z_v$ is \emph{good} regarding a $k$-coloring $\f'$ of $H$ if
$\f'(z_{v,i}) = i$ for all $i \in [k]$, and \emph{bad} otherwise.

\begin{lemma}
\label{lem:Cut-hard:quadratic:first}
Suppose that more than $\delta$-fraction and less than $\frac{1}{2}$-fraction of $Z_v$'s are bad for
$\delta \in \left(0,\frac{1}{2}\right)$.
Conditioned on the first test executed,
$\Vquad$ rejects with probability more than $\frac{2h \cdot \delta}{3k(k-1)}$.
\end{lemma}
\begin{proof} 
Since for any good $Z_v$ and bad $Z_w$,
there must be a pair $i \neq j$ such that
$z_{v,i} = z_{w,j}$,
$\Vquad$'s (conditional) rejection probability is 
\begin{align}
\begin{aligned}
    & \Pr_{\substack{(v,w) \sim E(X) \\ (i,j) \sim {[k] \choose 2}}}\Bigl[
        \f'(z_{v,i}) = \f'(z_{w,j})
    \Bigr] \\
    & \geq \Pr_{\substack{(v,w) \sim E(X) \\ (i,j) \sim {[k] \choose 2}}}\Bigl[
        \f'(z_{v,i}) = \f'(z_{w,j}) \Bigm| Z_v \text{ is good and } Z_w \text{ is bad}
    \Bigr]
    \cdot \Pr_{(v,w) \sim E(X)}\Bigl[
        Z_v \text{ is good and } Z_w \text{ is bad}
    \Bigr] \\
    & \geq \frac{1}{k(k-1)} \cdot
    \Pr_{(v,w) \sim E(X)}\Bigl[
        Z_v \text{ is good and } Z_w \text{ is bad}
    \Bigr]
\end{aligned}
\end{align}
Letting $S$ be the set of vertices $v \in V$ such that $Z_v$ is bad,
we have $\delta |V| < |S| < \frac{1}{2}|V| $, implying that
\begin{align}
\begin{aligned}
    \Pr_{(v,w) \sim E(X)}\Bigl[
        Z_v \text{ is good and } Z_w \text{ is bad}
    \Bigr]
    = \frac{|\partial_X(S)|}{|E(X)|}
    > \frac{h \delta \cdot |V|}{|E(X)|}
    = \frac{2h}{3} \cdot \delta.
\end{aligned}
\end{align}
Consequently, $\Vquad$'s rejection probability is more than $\frac{2h \cdot \delta}{3k(k-1)}$,
as desired.
\end{proof}

We then examine the (conditional) rejection probability of the second test.
We say that edge $(v,w)$ is \emph{legal} regarding a $k$-coloring $\f'$ of $H$ if
($\f'(v) \in [2]$, $\f'(w) \in [2]$, and $\f'(v) \neq \f'(w)$), and
\emph{illegal} otherwise.

\begin{lemma}
\label{lem:Cut-hard:quadratic:second}
Conditioned on the event that
the second test is executed and
both $Z_v$ and $Z_w$ are good for the selected edge $(v,w) \in E$,
the following hold\textup{:}
    \begin{itemize}
    \item if $(v,w)$ is legal, $\Vquad$ rejects with probability $0$\textup{;}
    
    \item if $(v,w)$ is illegal, $\Vquad$ rejects with probability at least $\frac{1}{2k-3}$\textup{;}
    
    \item if $\f'(v) \in [2]$, $\f'(w) \in [2]$, and $\f'(v) = \f'(w)$,
        then $\Vquad$ rejects with probability $\frac{1}{2k-3}$.
    \end{itemize}
\end{lemma}
\begin{proof} 
Conditioned on the selected edge $(v,w)$,
    the second test rejects with probability
\begin{align}
\label{eq:Cut-hard:quadratic:second}
\begin{aligned}
          \frac{k-2}{2k-3} \cdot \Pr_{i \sim \{3,\ldots,k\}}\Bigl[ \f'(v) = \f'(z_{v,i}) \Bigr]
        + \frac{k-2}{2k-3} \cdot \Pr_{i \sim \{3,\ldots,k\}}\Bigl[ \f'(w) = \f'(z_{w,i}) \Bigr]
        + \frac{1}{2k-3} \cdot \Bigl\llbracket \f'(v) = \f'(w) \Bigr\rrbracket
\end{aligned}
\end{align}
Observe easily that \cref{eq:Cut-hard:quadratic:second} is equal to $0$ whenever $(v,w)$ is legal.
Suppose then $(v,w)$ is illegal; at least one of the following must hold:
\begin{description}
    \item[(Case 1)] $\f'(v) \notin [2]$; thus,
        $\Pr_{i \sim \{3,\ldots,k\}}[ \f'(v) = \f'(x_{v,i}) ] = \frac{1}{k-2}$.
    \item[(Case 2)] $\f'(w) \notin [2]$; thus,
        $\Pr_{i \sim \{3,\ldots,k\}}[ \f'(w) = \f'(x_{w,i}) ] = \frac{1}{k-2}$.
    \item[(Case 3)] $\f'(v) = \f'(w)$; thus,
        $\llbracket \f'(v) = \f'(w) \rrbracket = 1$.
\end{description}
\cref{eq:Cut-hard:quadratic:second} is thus at least
\begin{align}
    \min\left\{
        \frac{k-2}{2k-3} \cdot \frac{1}{k-2},
        \frac{k-2}{2k-3} \cdot \frac{1}{k-2},
        \frac{1}{2k-3} \cdot 1
    \right\} = \frac{1}{2k-3}.
\end{align}
Obviously,
\cref{eq:Cut-hard:quadratic:second} is equal to $\frac{1}{2k-3}$
whenever $\f'(v) \in [2]$, $\f'(w) \in [2]$, and $\f'(v) = \f'(w)$,
as desired.
\end{proof}

We are now ready to prove \cref{lem:Cut-hard:quadratic}.

\begin{proof}[Proof of \cref{lem:Cut-hard:quadratic}]
We first prove the completeness; i.e.,
\begin{align}
\begin{aligned}
    & \opt_G\bigl(\f_\sss \reco \f_\ttt\bigr) \geq 1-\epsilon_c \implies
    \opt_H\bigl(\f'_\sss\ \reco \f'_\ttt\bigr) \geq 1 - \delta_c, \\
    & \text{where } \delta_c \defeq p_2 \cdot \frac{\epsilon_c}{2k-3}.
\end{aligned}
\end{align}
It is sufficient to consider the case that
$\f_\sss$ and $\f_\ttt$ differ in a single vertex, say $v^\star$.
Note that $\sqcol' = (\f'_\sss, \f'_\ttt)$ is
a reconfiguration sequence from $\f'_\sss$ to $\f'_\ttt$.
Let $\f'$ be either of $\f'_\sss$ or $\f'_\ttt$.
Since $Z_v$ is good regarding $\f'$ for all $v \in V$,
the first test never rejects $\f'$, and
the second test's rejection probability is
$0$ if the selected edge is bichromatic and
$\frac{1}{2k-3}$ if the selected edge is monochromatic by \cref{lem:Cut-hard:quadratic:second}.
Therefore, $\Vquad$ rejects $\f'$ with probability at most
\begin{align}
    p_1 \cdot 0 +
    p_2 \cdot \left( 0 \cdot (1-\epsilon_c) + \frac{1}{2k-3} \cdot \epsilon_c \right)
    = p_2 \cdot \frac{\epsilon_c}{2k-3}
    = \delta_c,
\end{align}
as desired.

We then prove the soundness; i.e.,
\begin{align}
\begin{aligned}
    & \opt_G\bigl(\f_\sss \reco \f_\ttt\bigr) < 1-\epsilon_s \implies
    \opt_H\bigl(\f'_\sss \reco \f'_\ttt\bigr) < 1 - \delta_s, \\
    & \text{where } \delta_s \defeq p_2 \cdot \frac{\frac{\epsilon_s + \epsilon_c}{2}}{2k-3}.
\end{aligned}
\end{align}
Let $\sqcol' = (\f'^{(1)}, \ldots, \f'^{(T)})$ be any reconfiguration sequence
from $\f'_\sss$ to $\f'_\ttt$ such that
$\val_H(\sqcol') = \opt_H(\f'_\sss \reco \f'_\ttt)$.
Define $\bar{\epsilon}$ as
\begin{align}
    \bar{\epsilon} \defeq \frac{\epsilon_s - \epsilon_c}{4} \in \left(0,\tfrac{1}{4}\right).
\end{align}
Suppose first $\sqcol'$ includes a $k$-coloring $\f'^{(t)}$ of $H$ such that
more than $\bar{\epsilon}$-fraction of $Z_v$'s are bad.
By \cref{lem:Cut-hard:quadratic:first},
$\Vquad$ rejects $\f'^{(t)}$ with probability more than
\begin{align}
    p_1 \cdot \frac{2h \cdot \bar{\epsilon}}{3k^2}.
\end{align}
Suppose next that
for every $k$-coloring $\f'^{(t)}$ of $H$ in $\sqcol'$,
at most $\bar{\epsilon}$-fraction of $Z_v$'s are bad.
Construct then a reconfiguration sequence $\sqcol = (\f^{(1)}, \ldots, \f^{(T)})$
from $\f_\sss$ to $\f_\ttt$, where
each $\f^{(t)}$ is a $2$-coloring of $G$ defined as follows: 
\begin{align}
    \f^{(t)}(v) \defeq
    \begin{cases}
        1 & \text{if } \f'^{(t)}(v) = 1 \\
        2 & \text{if } \f'^{(t)}(v) = 2 \\
        1 & \text{otherwise}
    \end{cases}
    \text{ for all } v \in V.
\end{align}
By assumption,
$\sqcol$ includes a $2$-coloring $\f^{(t)}$ of $G$ such that
more than $\epsilon_s$-fraction of edges of $G$ are monochromatic,
each of which must be illegal regarding $\f'^{(t)}$ due to the construction of $\f^{(t)}$.
Since at most $\bar{\epsilon}$-fraction of $Z_v$'s are bad regarding $\f'^{(t)}$,
the fraction of edges of $G$ incident to any bad $Z_v$ can be bounded by
\begin{align}
    \frac{\bar{\epsilon} \cdot |V|\cdot \Delta}{|E|} \leq 2 \bar{\epsilon}.
\end{align}
There are thus more than $(\epsilon_s - 2 \bar{\epsilon})$-fraction of \emph{illegal} edges $(v,w) \in E$
such that both $Z_v$ and $Z_w$ are good (regarding $\f'^{(t)}$).
By \cref{lem:Cut-hard:quadratic:second},
$\Vquad$ rejects $\f'^{(t)}$ with probability more than
\begin{align}
    p_2 \cdot \frac{\epsilon_s - 2\bar{\epsilon}}{2k-3}.
\end{align}
Setting the values of $p_1$ and $p_2$ so that
\begin{align}
    p_1 \cdot \frac{2h \cdot \bar{\epsilon}}{3k^2}
    = p_2 \cdot \frac{\epsilon_s - 2\bar{\epsilon}}{2k-3}
    \text{ and } p_1+p_2 = 1,
\end{align}
we find $\Vquad$'s rejection probability to be more than
\begin{align}
    \min\left\{
        p_1 \cdot \frac{2h \cdot \bar{\epsilon}}{3k^2},
        p_2 \cdot \frac{\epsilon_s - 2\bar{\epsilon}}{2k-3}
    \right\}
    = p_2 \cdot \frac{\frac{\epsilon_s + \epsilon_c}{2}}{2k-3}
    = \delta_s,
\end{align}
as desired.
Observe finally that $\delta_c < \delta_s$, completing the proof.
\end{proof}

%% file: main.bib
@Article{raz1998parallel,
  Title                    = {A parallel repetition theorem},
  Author                   = {Raz, Ran},
  Journal                  = {SIAM Journal on Computing},
  Year                     = {1998},
  Number                   = {3},
  Pages                    = {763--803},
  Volume                   = {27}
}

@Article{hastad1999clique,
  Title                    = {Clique is hard to approximate within $n^{1-\epsilon}$},
  Author                   = {H{\aa}stad, Johan},
  Journal                  = {Acta Mathematica},
  Year                     = {1999},
  Pages                    = {105--142},
  Volume                   = {182}
}

@article{hoory2006expander,
  title={Expander graphs and their applications},
  author={Hoory, Shlomo and Linial, Nathan and Wigderson, Avi},
  journal={Bulletin of the American Mathematical Society},
  volume={43},
  number={4},
  pages={439--561},
  year={2006}
}

@article{dinur2007pcp,
  title={The {PCP} Theorem by Gap Amplification},
  author={Dinur, Irit},
  journal   = {Journal of the ACM},
  volume={54},
  number={3},
  pages={12},
  year={2007}
}

@article{arora1998probabilistic,
  title={Probabilistic Checking of Proofs: A New Characterization of {NP}},
  author={Arora, Sanjeev and Safra, Shmuel},
  journal   = {Journal of the ACM},
    volume={45},
  number={1},
  pages={70--122},
  year={1998}
}

@article{arora1998proof,
  title={Proof Verification and the Hardness of Approximation Problems},
  author={Arora, Sanjeev and Lund, Carsten and Motwani, Rajeev and Sudan, Madhu and Szegedy, Mario},
  journal   = {Journal of the ACM},  volume={45},
  number={3},
  pages={501--555},
  year={1998}
}

@Article{feige1998threshold,
  Title                    = {A Threshold of $\ln n$ for Approximating Set Cover},
  Author                   = {Feige, Uriel},
  Journal                  = {Journal of the ACM},
  Year                     = {1998},
  Number                   = {4},
  Pages                    = {634--652},
  Volume                   = {45}
}

@inproceedings{hirahara2024probabilistically,
  title={Probabilistically Checkable Reconfiguration Proofs and Inapproximability of Reconfiguration Problems},
  author={Hirahara, Shuichi and Ohsaka, Naoto},
  booktitle={STOC},
  pages={1435--1445},
  year={2024}
}

@article{karthik2023inapproximability,
  title={On Inapproximability of Reconfiguration Problems: {PSPACE}-Hardness and some Tight {NP}-Hardness Results},
  author={{Karthik {C.~S.}} and Manurangsi, Pasin},
  journal={CoRR},
  Volume = {abs/2312.17140},
  year={2023}
}

@Article{bonsma2009finding,
  Title                    = {Finding paths between graph colourings: {PSPACE}-completeness and superpolynomial distances},
  Author                   = {Bonsma, Paul and Cereceda, Luis},
  Journal                  = {Theoretical Computer Science},
  Year                     = {2009},
  Number                   = {50},
  Pages                    = {5215--5226},
  Volume                   = {410}
}

@Article{cereceda2011finding,
  Title                    = {Finding paths between 3-colorings},
  Author                   = {Cereceda, Luis and {van den Heuvel}, Jan and Johnson, Matthew},
  Journal                  = {Journal of Graph Theory},
  Year                     = {2011},
  Number                   = {1},
  Pages                    = {69--82},
  Volume                   = {67}
}

@Article{cereceda2008connectedness,
  Title                    = {Connectedness of the graph of vertex-colourings},
  Author                   = {Cereceda, Luis and {van den Heuvel}, Jan and Johnson, Matthew},
  Journal                  = {Discrete Mathematics},
  Year                     = {2008},
  Number                   = {5-6},
  Pages                    = {913--919},
  Volume                   = {308}
}

@Article{gopalan2009connectivity,
  Title                    = {The Connectivity of {Boolean} Satisfiability: Computational and Structural Dichotomies},
  Author                   = {Gopalan, Parikshit and Kolaitis, Phokion G. and Maneva, Elitza and Papadimitriou, Christos H.},
  Journal                  = {SIAM Journal on Computing},
  Year                     = {2009},
  Number                   = {6},
  Pages                    = {2330--2355},
  Volume                   = {38}
}

@Book{hearn2009games,
  Title                    = {Games, Puzzles, and Computation},
  Author                   = {Hearn, Robert A. and Demaine, Erik D.},
  Publisher                = {A K Peters, Ltd.},
  Year                     = {2009}
}

@Article{hearn2005pspace,
  Title                    = {{PSPACE}-Completeness of Sliding-Block Puzzles and Other Problems through the Nondeterministic Constraint Logic Model of Computation},
  Author                   = {Hearn, Robert A. and Demaine, Erik D.},
  Journal                  = {Theoretical Computer Science},
  Year                     = {2005},
  Number                   = {1-2},
  Pages                    = {72--96},
  Volume                   = {343}
}

@InCollection{heuvel13complexity,
  Title                    = {The Complexity of Change},
  Author                   = {{van den Heuvel}, Jan},
  Booktitle                = {Surveys in Combinatorics 2013},
  Publisher                = {Cambridge University Press},
  Year                     = {2013},
  Pages                    = {127--160},
  Volume                   = {409}
}

@Article{ito2014approximability,
  Title                    = {Approximability of the subset sum reconfiguration problem},
  Author                   = {Ito, Takehiro and Demaine, Erik D.},
  Journal                  = {Journal of Combinatorial Optimization},
  Year                     = {2014},
  Number                   = {3},
  Pages                    = {639--654},
  Volume                   = {28}
}

@Article{ito2011complexity,
  Title                    = {On the Complexity of Reconfiguration Problems},
  Author                   = {Takehiro Ito and Erik D. Demaine and Nicholas J. A. Harvey and Christos H. Papadimitriou and Martha Sideri and Ryuhei Uehara and Yushi Uno},
  Journal                  = {Theoretical Computer Science},
  Year                     = {2011},
  Number                   = {12-14},
  Pages                    = {1054--1065},
  Volume                   = {412}
}

@Article{nishimura2018introduction,
  Title                    = {Introduction to Reconfiguration},
  Author                   = {Nishimura, Naomi},
  Journal                  = {Algorithms},
  Year                     = {2018},
  Number                   = {4},
  Pages                    = {52},
  Volume                   = {11}
}

@Article{wrochna2018reconfiguration,
  Title                    = {Reconfiguration in Bounded Bandwidth and Treedepth},
  Author                   = {Wrochna, Marcin},
  Journal                  = {Journal of Computer and System Sciences},
  Year                     = {2018},
  Pages                    = {1--10},
  Volume                   = {93}
}

@phdthesis{mouawad2015reconfiguration,
  title={On Reconfiguration Problems: Structure and Tractability},
  author={Mouawad, Amer},
  year={2015},
  school={University of Waterloo}
}

@inproceedings{ohsaka2022reconfiguration,
  title={Reconfiguration Problems on Submodular Functions},
  author={Ohsaka, Naoto and Matsuoka, Tatsuya},
  booktitle={WSDM},
  pages={764--774},
  year={2022}
}

@inproceedings{bonamy2020shortest,
  title={Shortest Reconfiguration of Colorings Under {Kempe} Changes},
  author={Bonamy, Marthe and Heinrich, Marc and Ito, Takehiro and Kobayashi, Yusuke and Mizuta, Haruka and M{\"u}hlenthaler, Moritz and Suzuki, Akira and Wasa, Kunihiro},
  booktitle={STACS},
  pages     = {35:1--35:14},
    year={2020}
}

@article{garey1976some,
  title={Some Simplified {NP}-complete Graph Problems},
  author={Garey, Michael. R. and Johnson, David S. and Stockmeyer, Larry J.},
  journal   = {Theoretical Computer Science},
    volume={1},
  number={3},
  pages={237--267},
  year={1976}
}

@article{papadimitriou1991optimization,
  title     = {Optimization, Approximation, and Complexity Classes},
    author={Papadimitriou, Christos H. and Yannakakis, Mihalis},
  journal   = {Journal of Computer and System Sciences},
    volume={43},
  number={3},
  pages={425--440},
  year={1991}
}

@article{hastad2001some,
  title={Some optimal inapproximability results},
  author={H{\aa}stad, Johan},
  journal   = {Journal of the ACM},
    volume={48},
  number={4},
  pages={798--859},
  year={2001}
}

@InProceedings{ohsaka2023gap,
  Title                    = {Gap Preserving Reductions Between Reconfiguration Problems},
  Author                   = {Ohsaka, Naoto},
  Booktitle                = {STACS},
  Year                     = {2023},
  Pages                    = {49:1--49:18}
}

@Misc{hoang2023combinatorial,
  Title                    = {Combinatorial Reconfiguration},
  Author                   = {Hoang, Duc A.},
  HowPublished             = {\url{https://reconf.wikidot.com/}},
  Year                     = {2023}
}

@inproceedings{ohsaka2024gap,
  title={Gap Amplification for Reconfiguration Problems},
  author={Ohsaka, Naoto},
  booktitle={SODA},
  pages={1345--1366},
  year={2024}
}

@inproceedings{ohsaka2024alphabet,
  title={Alphabet Reduction for Reconfiguration Problems},
  author={Ohsaka, Naoto},
  booktitle={ICALP},
  pages = {113:1-113:17},
  year={2024}
}

@inproceedings{hirahara2024optimal,
  title={Optimal {PSPACE}-hardness of Approximating Set Cover Reconfiguration},
  author={Hirahara, Shuichi and Ohsaka, Naoto},
  booktitle={ICALP},
  pages = {85:1--85:18},
  year={2024}
}

@article{austrin2014new,
  title={New {NP}-Hardness Results for 3-Coloring and 2-to-1 Label Cover},
  author={Austrin, Per and O'Donnell, Ryan and Tan, Li{-}Yang and Wright, John},
  journal={ACM Transactions on Computation Theory},
  volume={6},
  number={1},
  pages={1--20},
  year={2014}
}

@article{guruswami2013improved,
  title={Improved Inapproximability Results for Maximum $k$-Colorable Subgraph},
  author={Guruswami, Venkatesan and Sinop, Ali Kemal},
  journal={Theory of Computing},
  volume={9},
  number = {11},
  pages={413--435},
  year={2013}
}

@article{kann1997hardness,
  author       = {Kann, Viggo and
                  Khanna, Sanjeev and
                  Lagergren, Jens and
                  Panconesi, Alessandro},
  title        = {On the Hardness of Approximating \textsc{Max} $k$-Cut and its Dual},
  journal      = {Chicago Journal of Theoretical Computer Science},
  volume       = {1997},
  year         = {1997}
}

@article{frieze1997improved,
  title={Improved Approximation Algorithms for {MAX $k$-CUT} and {MAX BISECTION}},
  author={Frieze, Alan M. and Jerrum, Mark},
  journal={Algorithmica},
  volume={18},
  number={1},
  pages={67--81},
  year={1997}
}

@article{ohsaka2025approximate,
  title={On Approximate Reconfigurability of Label Cover},
  author={Ohsaka, Naoto},
  journal={Information Processing Letters},
  pages={106556},
  volume = {189},
  year={2025}
}

@book{alon2016probabilistic,
  title={The Probabilistic Method},
  author={Alon, Noga and Spencer, Joel H.},
  year={2016},
  publisher={Wiley}
}

@article{gavinsky2015tail,
  title={A tail bound for read-$k$ families of functions},
  author={Gavinsky, Dmitry and Lovett, Shachar and Saks, Michael E. and Srinivasan, Srikanth},
  journal={Random Structures \& Algorithms},
  volume={47},
  number={1},
  pages={99--108},
  year={2015}
}

@article{cereceda2009mixing,
  title={Mixing 3-colourings in bipartite graphs},
  author={Cereceda, Luis and {van den Heuvel}, Jan and Johnson, Matthew},
  journal={European Journal of Combinatorics},
  volume={30},
  number={7},
  pages={1593--1606},
  year={2009}
}

@article{gabber1981explicit,
  title={Explicit Constructions of Linear-Sized Superconcentrators},
  author={Gabber, Ofer and Galil, Zvi},
  journal={Journal of Computer and System Sciences},
  volume={22},
  number={3},
  pages={407--420},
  year={1981}
}

@article{reingold2002entropy,
  title={Entropy Waves, the Zig-Zag Graph Product, and New Constant-Degree Expanders},
  author={Reingold, Omer and Vadhan, Salil and Wigderson, Avi},
  journal={Annals of Mathematics},
  volume={155},
  number={1},
  pages={157--187},
  year={2002}
}

@InProceedings{radhakrishnan2006gap,
  Title                    = {Gap Amplification in {PCPs} Using Lazy Random Walks},
  Author                   = {Radhakrishnan, Jaikumar},
  Booktitle                = {ICALP},
  Year                     = {2006},
  Pages                    = {96--107}
}

@Article{radhakrishnan2007dinurs,
  Title                    = {On {Dinur}'s Proof of the {PCP} Theorem},
  Author                   = {Radhakrishnan, Jaikumar and Sudan, Madhu},
  Journal                  = {Bulletin of the American Mathematical Society},
  Year                     = {2007},
  Number                   = {1},
  Pages                    = {19--61},
  Volume                   = {44}
}

@article{stockmeyer1973planar,
  title={Planar 3-colorability is polynomial complete},
  author={Stockmeyer, Larry},
  journal      = {ACM {SIGACT} News},
    volume={5},
  number={3},
  pages={19--25},
  year={1973}
}

@InProceedings{lovasz1973coverings,
  Title                    = {Coverings and coloring of hypergraphs},
  Author                   = {Lov{\'a}sz, L{\'a}szl{\'o}},
  Booktitle                = {Proceedings of the 4th Southeastern Conference on Combinatorics, Graph Theory, and Computing},
  Year                     = {1973},
  Pages                    = {3--12}
}

@article{goemans1995improved,
  title={Improved Approximation Algorithms for Maximum Cut and Satisfiability Problems Using Semidefinite Programming},
  author={Goemans, Michel X. and Williamson, David P.},
  journal={Journal of the ACM},
  volume={42},
  number={6},
  pages={1115--1145},
  year={1995}
}

@inproceedings{khot2002power,
  title={On the Power of Unique 2-Prover 1-Round Games},
  author={Khot, Subhash},
  booktitle={STOC},
  pages={767--775},
  year={2002}
}

@article{khot2007optimal,
  title={Optimal Inapproximability Results for {MAX-CUT} and Other 2-Variable {CSPs}?},
  author={Khot, Subhash and Kindler, Guy and Mossel, Elchanan and O'Donnell, Ryan},
  journal={SIAM Journal on Computing},
  volume={37},
  number={1},
  pages={319--357},
  year={2007}
}

@article{mossel2010noise,
  title={Noise stability of functions with low influences: Invariance and optimality},
  author={Mossel, Elchanan and O'Donnell, Ryan and Oleszkiewicz, Krzysztof},
  journal={Annals of Mathematics},
  volume={171},
  number={1},
  pages={295--341},
  year={2010}
}

@article{zuckerman2007linear,
  title={Linear Degree Extractors and the Inapproximability of Max Clique and Chromatic Number},
  author={Zuckerman, David},
  journal={Theory of Computing},
  volume={3},
  number={1},
  pages={103--128},
  year={2007}
}

@article{feige1998zero,
  title={Zero Knowledge and the Chromatic Number},
  author={Feige, Uriel and Kilian, Joe},
  journal={Journal of Computer and System Sciences},
  volume={57},
  number={2},
  pages={187--199},
  year={1998}
}

@article{jerrum1995very,
  title={A Very Simple Algorithm for Estimating the Number of k-Colorings of a Low-Degree Graph},
  author={Jerrum, Mark},
  journal={Random Structures \& Algorithms},
  volume={7},
  number={2},
  pages={157--165},
  year={1995}
}

@article{bousquet2024note,
  title={A Note on the Complexity of Graph Recoloring},
  author={Bousquet, Nicolas},
  journal={CoRR},
  year={2024},
  Volume                   = {abs/2401.03011}
}

@article{johnson2016finding,
  title={Finding Shortest Paths Between Graph Colourings},
  author={Johnson, Matthew and Kratsch, Dieter and Kratsch, Stefan and Patel, Viresh and Paulusma, Dani{\"e}l},
  journal={Algorithmica},
  volume={75},
  number={2},
  pages={295--321},
  year={2016}
}

@incollection{mynhardt2019reconfiguration,
  title={Reconfiguration of Colourings and Dominating Sets in Graphs},
  author={Mynhardt, Christina M. and Nasserasr, Shahla},
  booktitle={50 years of Combinatorics, Graph Theory, and Computing},
  pages={171--191},
  year={2019},
  chapter={10},
  publisher={Chapman and Hall/CRC}
}

@article{petrank1994hardness,
  title={The Hardness of Approximation: Gap Location},
  author={Petrank, Erez},
  journal={Computational Complexity},
  volume={4},
  pages={133--157},
  year={1994}
}

@article{dyer2006randomly,
  title={Randomly coloring sparse random graphs with fewer colors than the maximum degree},
  author={Dyer, Martin E. and Flaxman, Abraham D. and Frieze, Alan M. and Vigoda, Eric},
  journal={Random Structures \& Algorithms},
  volume={29},
  number={4},
  pages={450--465},
  year={2006}
}

@article{molloy2004glauber,
  title={The {Glauber} Dynamics on Colorings of a Graph with High Girth and Maximum Degree},
  author={Molloy, Michael},
  journal={SIAM Journal on Computing},
  volume={33},
  number={3},
  pages={721--737},
  year={2004}
}

@article{bonamy2014reconfiguration,
  title={Reconfiguration graphs for vertex colourings of chordal and chordal bipartite graphs},
  author={Bonamy, Marthe and Johnson, Matthew and Lignos, Ioannis and Patel, Viresh and Paulusma, Dani{\"e}l},
  journal={Journal of Combinatorial Optimization},
  volume={27},
  number={1},
  pages={132--143},
  year={2014}
}

@phdthesis{cereceda2007mixing,
  title={Mixing Graph Colourings},
  author={Cereceda, Luis},
  year={2007},
  school={London School of Economics and Political Science}
}

@article{bellare1998free,
  title={Free Bits, {PCPs}, and Nonapproximability --- {T}owards Tight Results},
  author={Bellare, Mihir and Goldreich, Oded and Sudan, Madhu},
  journal={SIAM Journal on Computing},
  volume={27},
  number={3},
  pages={804--915},
  year={1998}
}

@article{bousquet2024survey,
  title={A survey on the parameterized complexity of reconfiguration problems},
  author={Bousquet, Nicolas and Mouawad, Amer E. and Nishimura, Naomi and Siebertz, Sebastian},
  journal={Computer Science Review},
  volume={53},
  pages={100663},
  year={2024}
}

@article{bonamy2013recoloring,
  title={Recoloring bounded treewidth graphs},
  author={Bonamy, Marthe and Bousquet, Nicolas},
  journal={Electronic Notes in Discrete Mathematics},
  volume={44},
  pages={257--262},
  year={2013}
}

@article{bonamy2011diameter,
  title={On the diameter of reconfiguration graphs for vertex colourings},
  author={Bonamy, Marthe and Johnson, Matthew and Lignos, Ioannis and Patel, Viresh and Paulusma, Dani{\"e}l},
  journal={Electronic Notes in Discrete Mathematics},
  volume={38},
  pages={161--166},
  year={2011}
}

@inproceedings{bonsma2014complexity,
  title={The Complexity of Bounded Length Graph Recoloring and {CSP} Reconfiguration},
  author={Bonsma, Paul and Mouawad, Amer E. and Nishimura, Naomi and Raman, Venkatesh},
  booktitle={IPEC},
  pages={110--121},
  year={2014}
}

@article{hatanaka2019coloring,
  title={The Coloring Reconfiguration Problem on Specific Graph Classes},
  author={Hatanaka, Tatsuhiko and Ito, Takehiro and Zhou, Xiao},
  journal={IEICE Transactions on Information and Systems},
  volume={102},
  number={3},
  pages={423--429},
  year={2019}
}

@article{ohsaka2024tight,
  title={Tight Inapproximability of Target Set Reconfiguration},
  author={Ohsaka, Naoto},
  journal={CoRR},
  year={2024},
    Volume                   = {abs/2402.15076}
}
